\documentclass[acmsmall]{acmart}
\PassOptionsToPackage{dvipsnames}{xcolor}
\usepackage{boldline}
\usepackage{xcolor}
\usepackage{xargs}
\usepackage{framed} \usepackage[framemethod=default]{mdframed}
\usepackage{mathpartir}
\usepackage{mathtools}
\usepackage{soul}
\usepackage{cmll} \usepackage{amsthm}
\usepackage[ruled,vlined]{algorithm2e}
\usepackage{cleveref}
\usepackage{makecell}
\usepackage{fancyvrb}
\usepackage{thmtools, thm-restate}
\usepackage{listings}
\usepackage{todonotes}
\usepackage[justification=centering]{subfig}

\newcommand{\incode}[1]{\lstinline{#1}}

\definecolor{working}{rgb}{0.9,1,0.9}
\newmdenv[hidealllines=true,backgroundcolor=working,innerleftmargin=0pt,innerrightmargin=0pt,innertopmargin=-3pt,innerbottommargin=-3pt,skipabove=3pt,skipbelow=3pt]{working}
\newcommand{\workingcolorname}{light green}

\definecolor{notyet}{rgb}{1,1,0.85}
\newmdenv[hidealllines=true,backgroundcolor=notyet,innerleftmargin=0pt,innerrightmargin=0pt,innertopmargin=-3pt,innerbottommargin=-3pt,skipabove=3pt,skipbelow=3pt]{notyet}
\newcommand{\notyetcolorname}{light yellow}

\definecolor{noway}{rgb}{1,0.9,0.9}
\newmdenv[hidealllines=true,backgroundcolor=noway,innerleftmargin=0pt,innerrightmargin=0pt,innertopmargin=-3pt,innerbottommargin=-3pt,skipabove=3pt,skipbelow=3pt]{noway}
\newcommand{\nowaycolorname}{light red}

\definecolor{limitation}{rgb}{1.0, 0.875, 0.75}
\newmdenv[hidealllines=true,backgroundcolor=limitation,innerleftmargin=0pt,innerrightmargin=0pt,innertopmargin=-3pt,innerbottommargin=-3pt,skipabove=3pt,skipbelow=3pt]{limitation}
\newcommand{\limitationcolorname}{light orange}

\DefineVerbatimEnvironment{code}{Verbatim}{fontsize=\small}
\DefineVerbatimEnvironment{example}{Verbatim}{fontsize=\small}

\newcommand{\lolli}{\multimap}

\newcommand{\llet}[2]{\mathsf{let}~#1~\mathsf{in}~#2}
\newcommand{\lletrec}[2]{\mathsf{letrec}~#1~\mathsf{in}~#2}
\newcommand{\ccase}[2]{\mathsf{case}~#1~\mathsf{of}~#2}

\definecolor{darkblue}{rgb}{0,0,0.5}
\definecolor{darkgreen}{rgb}{0,0.3,0}
\definecolor{darkpink}{rgb}{0.4,0,0.3}
\definecolor{graygreen}{rgb}{0.3,0.5,0.3}
\definecolor{grayblue}{rgb}{0.2,0.2,0.6}
\definecolor{grayred}{rgb}{0.5,0.2,0.2}

\newcommand{\ROUNDTWO}[1]{#1}

\newcommand{\G}{\Gamma}
\newcommand{\D}{\Delta}

\newcommand{\s}{\sigma}
\newcommand{\vp}{\varphi}

\newcommandx{\var}[3][1=x, 2=1, 3=\tau]{\ensuremath{#1{:}_{#2}#3}}
\newcommand{\x}{\var[x]}
\newcommand{\y}{\var[y]}
\newcommand{\z}{\var[z]}

\newcommand{\xl}{\x[1]}
\newcommand{\xo}{\x[\omega]}

\newcommand{\xD}{\x[\Delta]}

\newcommand{\yo}{\y[\omega][\sigma']}

\newcommand{\yD}{\y[\Delta][\sigma']}

\newcommand{\zD}{\z[\overline{\Delta}]}

\newcommand{\ov}[1]{\ensuremath{\overline{#1}}}

 \renewcommand{\G}{\Gamma}
\renewcommand{\D}{\Delta}

\newcommand{\irr}[1]{\left[#1\right]}

\renewcommand{\s}{\sigma}
\renewcommand{\vp}{\varphi}

\newcommandx{\judg}[8][1=\Gamma, 2=\Delta, 3=\delta, 4=\Delta', 7, 8]{#1;#2;#3;\irr{#4} \vdash_{#7} #5 :_{#8} #6}
\newcommand{\splitjudg}[2]{\judg[\Gamma][\Delta_1,\Delta_2][\delta,\delta'][\Delta_3,\Delta_4]{#1}{#2}}

\renewcommandx{\var}[3][1=x, 2=1, 3=\sigma]{\ensuremath{#1{:}_{#2}#3}}
\renewcommand{\x}{\var[x]}
\renewcommand{\y}{\var[y]}
\renewcommand{\z}{\var[z]}

\renewcommand{\xl}{\x[1]}
\renewcommand{\xo}{\x[\omega]}

\renewcommand{\xD}{\var[x][\Delta,\irr{\Delta'}][\sigma]}

\renewcommand{\yo}{\y[\omega][\sigma']}

\renewcommand{\yD}{\y[\Delta][\sigma']}

\renewcommand{\zD}{\z[\Delta,\irr{\Delta'}][\sigma]}

\renewcommand{\ov}[1]{\ensuremath{\overline{#1}}}

 \renewcommand{\G}{\Gamma}
\renewcommand{\D}{\Delta}

\renewcommand{\irr}[1]{\left[#1\right]}

\renewcommand{\s}{\sigma}
\renewcommand{\vp}{\varphi}

\renewcommandx{\judg}[7][1=\Gamma, 2=\Delta, 3=\delta, 6, 7]{#1;#2;#3 \vdash_{#6} #4 :_{#7} #5}
\renewcommand{\splitjudg}[2]{\judg[\Gamma][\Delta,\Delta'][\delta,\delta']{#1}{#2}}

\renewcommandx{\var}[3][1=x, 2=1, 3=\sigma]{\ensuremath{#1{:}_{#2}#3}}
\renewcommand{\x}{\var[x]}
\renewcommand{\y}{\var[y]}
\renewcommand{\z}{\var[z]}

\renewcommand{\xl}{\x[1]}
\renewcommand{\xo}{\x[\omega]}

\renewcommand{\xD}{\var[x][\Delta][\sigma]}
\newcommand{\xt}[1]{\x[1][\s]\##1}
\newcommand{\xr}[2]{\x[#1][#2]}

\renewcommand{\yo}{\y[\omega][\sigma']}

\renewcommand{\yD}{\y[\Delta][\sigma']}

\renewcommand{\zD}{\z[\Delta][\sigma]}

\newcommand{\zr}[2]{\z[#1][#2]}

\renewcommand{\ov}[1]{\ensuremath{\overline{#1}}}

\newcommand{\hasargs}[2]{#1~\textrm{has}~#2~\textrm{arguments}}
\newcommand{\lctag}[2]{#1\##2}
\newcommand{\subst}[3]{#1\left[#2/#3\right]}

\newcommand{\konstructor}{K~\ov{\xo},\ov{y_i{:}_1\s_i}^n}

 \renewcommand{\llet}[2]{\mathbf{let}~#1~\mathbf{in}~#2}
\renewcommand{\lletrec}[2]{\mathbf{let~rec}~#1~\mathbf{in}~#2}
\renewcommand{\ccase}[2]{\mathbf{case}~#1~\mathbf{of}~#2}
\newcommand{\judgment}[1]{
    \begin{tabular}{V{2.7}cV{2.7}}
    \hlineB{2.7}
    $#1$\\
    \hlineB{2.7}
    \end{tabular}
}
\newcommand{\datatype}[2]{
  \mathbf{data}~#1~\mathbf{where}~#2
}

\renewcommand{\G}{\Gamma}
\renewcommand{\D}{\Delta}

\renewcommand{\irr}[1]{\left[#1\right]}

\renewcommand{\s}{\sigma}
\renewcommand{\vp}{\varphi}

\renewcommandx{\judg}[7][1=\Gamma, 2=\Delta, 5, 6, 7]{#1;#2 \vdash_{#5} #3 :_{#6}^{#7} #4}
\renewcommand{\splitjudg}[2]{\judg[\Gamma][\Delta,\Delta']{#1}{#2}}

\renewcommandx{\var}[3][1=x, 2=1, 3=\tau]{\ensuremath{#1{:}_{#2}#3}}
\renewcommand{\x}{\var[x]}
\renewcommand{\y}{\var[y]}
\renewcommand{\z}{\var[z]}

\renewcommand{\xl}{\x[1]}
\renewcommand{\xo}{\x[\omega]}

\renewcommand{\xD}{\var[x][\Delta][\sigma]}
\renewcommand{\xt}[1]{\x[1][\tau]\##1}
\renewcommand{\xr}[2]{\x[#1][#2]}

\renewcommand{\yo}{\y[\omega][\sigma']}

\renewcommand{\yD}{\y[\Delta][\sigma']}

\renewcommand{\zD}{\z[\Delta][\sigma]}

\renewcommand{\zr}[2]{\z[#1][#2]}

\renewcommand{\ov}[1]{\ensuremath{\overline{#1}}}

\renewcommand{\hasargs}[2]{#1~\textrm{has}~#2~\textrm{linear arguments}}
\renewcommand{\lctag}[2]{#1\##2}
\renewcommand{\subst}[3]{#1\left[#2/#3\right]}

\renewcommand{\konstructor}{K~\ov{\xo},\ov{y_i{:}_1\s_i}^n}

 \newcommand{\SyntaxTypes}{
\ensuremath{
\begin{array}{lcll}
    \ROUNDTWO{\tau,\sigma}  & ::=  & T~\overline{p}         & \textrm{Datatype} \\
                    & \mid & \ROUNDTWO{\tau} \to_\pi \sigma & \textrm{Function with multiplicity}\\
                    & \mid & \forall p.~\ROUNDTWO{\tau}     & \textrm{Multiplicity universal scheme}\\
\end{array}
}
}

\newcommand{\SyntaxTerms}{
\ensuremath{
\begin{array}{lcll}                                             
    e                & ::=  & x,y,z \mid K                                                                 & \textrm{Variables and data constructors}\\
                     & \mid & \Lambda p.~e~\mid~e~\pi                                          & \textrm{Multiplicity abstraction/application}\\
                     & \mid & \lambda x{:}_\pi\ROUNDTWO{\tau}.~e~\mid~e~e'                           & \textrm{Term abstraction/application}\\
                     & \mid & \llet{x{:}_{\Delta}\sigma = e}{e'}
                                                                                                           &
                                                                                                             \textrm{Let
                                                                                                             \ROUNDTWO{with
                                                                                                             usage
                                                                                                             environment}}
  \\
                     & \mid & \lletrec{\overline{x{:}_{\Delta}\sigma =
                              e}}{e'}             & \textrm{Recursive
                                                    Let \ROUNDTWO{with
                                                                                                             usage
                                                                                                             environment}} \\
                     & \mid &
                     \ccase{e}{z{:}_{\Delta}\ROUNDTWO{\tau}~\{\overline{\rho\ROUNDTWO{\Rightarrow} e'}\}}   & \textrm{Case} \\
\rho             & ::=  & K~\overline{x{:}_\pi\ROUNDTWO{\tau}} \mid \_                              & \textrm{Pattern and wildcard} \\
\end{array}
}
}

\newcommand{\SyntaxEnvironments}{
\ensuremath{
\begin{array}{lcll}
\Gamma   & ::=  & \ROUNDTWO{\cdot \mid \Gamma,x{:}_\omega\tau \mid
                    \Gamma,K{:}\tau \mid \Gamma,p \mid
                    \Gamma,z{:}_{\Delta}\tau} & \ROUNDTWO{\textrm{Unrestricted
                                                typing environment}} \\
  \Delta   & ::=  & \cdot \mid \Delta,x{:}_\pi\tau \mid
                    \Delta,[x{:}_\pi\tau] & \ROUNDTWO{\textrm{Linear
                                            typing environment}} \\
\end{array}
}
}

\newcommand{\SyntaxFull}{
\begin{figure}[ht]
\[
{\small
  \begin{array}{l}
\textbf{Types} \\
\SyntaxTypes\\
\textbf{Terms}\\
\SyntaxTerms\\\\
\textbf{Environments}\\
\SyntaxEnvironments\\\\
\textbf{Multiplicities} \qquad\qquad \textbf{Declarations}\\
\begin{array}{lcl}
  \pi & ::= & 1 \mid \omega \mid p \qquad decl  ::=  \datatype{T~\overline{p}}{\overline{K:\overline{\ROUNDTWO{\tau} \to_\pi}~T~\overline{p}}} \end{array}
\end{array}}
\]
               
\caption{Linear Core Syntax}
\label{fig:full-linear-core-syntax}
\end{figure}
}

 \newcommand{\TypeVarOmega}{
    \infer*[right=(Var$_\omega$)]
    {\,}
    {\judg[\G,\xo][\cdot]{x}{\tau}}
}

\newcommand{\TypeLinearVar}{
    \infer*[right=(Var$_1$)]
    {\,}
    {\judg[\G][\xl]{x}{\tau}}
}

\newcommand{\TypeVarDelta}{
    \infer*[right=(Var$_{\D}$)]
    { \ROUNDTWO{\Delta = \Delta'} }
    {\judg[\G,\xD][\ROUNDTWO{\D'}]{x}{\tau}}
  }

  \newcommand{\TypeVarP}
  {
     \ROUNDTWO{\infer*[right=(Var$_{p}$)]
    { \rho \in \G }
    {\judg[\G][x:_\rho \tau ]{x}{\tau}}}
  }

\newcommand{\TypeVarSplit}{
  \infer[(Split)]
    {\judg[\G][\D,x:_1\tau]{e}{\tau} \and \hasargs{K}{n}}
    {\judg[\G][\D,\ov{\xt{K_i}}^n]{e}{\s}}
}

\newcommand{\TypeMultLamIntro}{
    \infer*[right=($\Lambda I$)]
    {\judg[\G,p]{e}{\tau} \and p\notin\Gamma}
    {\judg{\Lambda p.~e}{\forall p.~\tau}}
}

\newcommand{\TypeMultLamElim}{
    \infer*[right=($\Lambda E$)]
    {\judg{e}{\forall p.~\tau} \and \Gamma\vdash_{mult}\pi}
    {\judg{e~\pi}{\tau[\pi/p]}}
}

\newcommand{\TypeLamIntroL}{
    \infer*[right=($\lambda I_1$)]
    {\judg[\G][\D,\xl]{e}{\s} \and x\notin\Delta}
    {\judg{\lambda \xl.~e}{\tau\to_{1}\s}}
}

\newcommand{\TypeLamIntroW}{
    \infer*[right=($\lambda I_\omega$)]
    {\judg[\G,\xo]{e}{\vp} \and x\notin\Gamma}
    {\judg{\lambda \xo.~e}{\s\to_{\omega}\vp}}
}

\newcommand{\TypeLamElimL}{
    \infer*[right=($\lambda E_1$)]
    {\judg[\G][\D]{e}{\tau\to_{1}\s} \and \judg[\G][\D']{e'}{\tau}}
    {\splitjudg{e~e'}{\s}}
}

\newcommand{\TypeLamElimW}{
\infer*[right=($\lambda E_\omega$)]
    {\judg{e}{\tau\to_{\omega}\s} \and \judg[\G][\cdot]{e'}{\tau}}
    {\judg{e~e'}{\s}}
}

\newcommand{\TypeLet}{
     \infer*[right=(Let)]
{\judg[\G][\D]{e}{\tau} \\ \judg[\G,\xr{\D}{\tau}][\D,\D']{e'}{\s}}
    {\judg[\G][\D,\D']{\llet{\x[\D] = e}{e'}}{\s}}
}

\newcommand{\TypeLetRec}{
     \infer*[right=(LetRec)]
{\ov{\judg[\G,\ov{x_i{:}_{\D}\tau_i}][\D]{e_i}{\s}} \\ \judg[\G,\ov{x_i{:}_{\D}\tau_i}][\D,\D']{e'}{\s}}
    {\judg[\G][\D,\D']{\lletrec{\ov{\var[x_i][\D][\tau_i] = e_i}}{e'}}{\s}}
  }

  \newcommand{\TypeConstr}{
    \infer*[right=(Constr)]
    {\ov{\G;\cdot\vdash e_\omega : \tau_i} \and \ov{\G ; \D_j \vdash
          e_j : \tau_j} \and \ov{\D_j} = \D \and K : \ov{ \tau_i
          \rightarrow_\omega \tau_j \lolli} \s }
    {\judg[\G][\D]{K~\ov{e_\omega}\ov{e_i} }{\s}}
    }

\newcommand{\TypeCaseWHNFIntermediate}{
    \mprset{flushleft}
    \infer[(Case$_\textrm{WHNF}$)]
    {\textrm{e is in \emph{WHNF}} \\ \G;\D \Vdash e:\tau \gtrdot \ov{\D_i}
    \\ {\ov{\G,\var[z][\ov{\D_i}];\ov{\D_i},\D' \vdash \rho\,\,\ROUNDTWO{\Rightarrow}\,\, e' :^z_{\ov{\D_i}} \s}}}
    {\judg[\G][\D,\D']{\ccase{e}{\var[z][\ov{\D_i}]~\{\ov{\rho \,\,\ROUNDTWO{\Rightarrow}\,\, e'}\}}}{\s}}
}

\newcommand{\TypeCaseWHNF}{
    \infer[(Case$_\textrm{WHNF}$)]
    {\textrm{e is in \emph{WHNF}} \and e~\textsf{matches}~\rho_j
    \and \ROUNDTWO{\ov{\judg[\G,z{:}_{\irr{\D}}\tau][\irr{\D},\D']{\rho \Rightarrow_\mathsf{NWHNF}  e'}{
        \s}[][\irr{\D}][z]}}\\
    \G;\D \Vdash e:\tau \gtrdot \ov{\D_i} \and
    \ROUNDTWO{\G,\var[z][\ov{\D_i}];\ov{\D_i},\D' \vdash \rho_j \Rightarrow e'
    :^z_{\ov{\D_i}} \s}
    }
    {{\judg[\G][\D,\D']{\ccase{e}{\var[z][\ov{\D_i}]~\{\ov{\rho \,\ROUNDTWO{\Rightarrow}\, e'}\}}}{\s}}}
}

\newcommand{\TypeCaseNotWHNF}{
    \mprset{flushleft}
    \infer[(Case$_\textrm{Not WHNF}$)]
    {\textrm{e is not in \emph{WHNF}} \\
    \judg{e}{\tau}
    \\ \ov{\judg[\G,z{:}_{\irr{\D}}\tau][\irr{\D},\D']{\ROUNDTWO{\rho \Rightarrow_\mathsf{NWHNF}}\, e'}{ \s}[][{\irr{\D}}][{z}]}
    }
    {\splitjudg{\ccase{e}{\z[\irr{\D}]~\{\ov{\rho\,
            \ROUNDTWO{\Rightarrow}\, e'}\}}}{\s}}
}

\newcommand{\TypeAltNNotWHNF}{
    \mprset{flushleft}
    \infer[(AltN$_{\textrm{Not WHNF}}$)]
    { \judg[\G,\ov{\xo},\ov{y_i{:}_{\D_i}\tau_i}][\D]{e}{\s}
    \\ \ov{\D_i} = \ov{\lctag{\D_s}{K_i}}^n
}
    {\judg{K~\ov{\xo},\ov{y_i{:}_1\tau_i}^n \,\ROUNDTWO {\Rightarrow_\mathsf{NWHNF}}\,  e}{\s}[][\D_s][z]}
}

\newcommand{\TypeAltNWHNF}{
    \mprset{flushleft}
    \infer[(AltN$_{\textrm{WHNF}}$)]
    { \judg[\G,\ov{\xo},\ov{y_i{:}_{\D_i}\tau_i}^n][\D]{e}{\s}
}
    {\judg{K~\ov{\xo},\ov{y_i{:}_1\tau_i}^n \,\ROUNDTWO{\Rightarrow_\mathsf{WHNF}}\,e}{ \s}[][\ov{\D_i}^n][z]}
}

\newcommand{\TypeAltZero}{
    \mprset{flushleft}
    \infer[(Alt0)]
    { \judg[\subst{\G}{\cdot}{\D_s}_z,\ov{\xo}][\subst{\D}{\cdot}{\D_s}]{e}{\s}
}
    {\judg{K~\ov{\xo}\,\ROUNDTWO{\Rightarrow}\, e}{ \s}[][\D_s][z]}
}

\newcommand{\TypeAltWild}{
    \mprset{flushleft}
    \infer[(Alt\_)]
    { \judg{e}{\s} }
    {{\judg{\_ \,\ROUNDTWO{\Rightarrow}\ e}{ \s}[][\D_s][z]}}
}

\newcommand{\TypeWHNFCons}{
    \infer[($\textrm{WHNF}_K$)]
    { \ov{\G; \cdot \vdash e_\omega : \tau_i} \\ \ov{\G; \D_j \vdash e_j
        : \tau_j} \\ \ov{\D_j} = \D \\ K : \ov{\tau_i \rightarrow_\omega
        \tau_j \lolli} \s}
    { \G; \D \Vdash K~\ov{e_\omega}\ov{e_j} : \s \gtrdot \ov{\D_j} }
}

\newcommand{\TypeWHNFLam}{
    \infer[($\textrm{WHNF}_\lambda$)]
    { \judg{\lambda x.~e}{\s} }
    { \G; \D \Vdash \lambda x.~e : \s \gtrdot \D}
}

\newcommand{\TypeWellFormedMult}{
    \infer*[right=$(1)$]
    { }
    {\Gamma \vdash 1}
\qquad
    \infer*[right=$(\omega)$]
    { }
    {\Gamma \vdash \omega}
\qquad
    \infer*[right=$(\rho)$]
    { }
    {\Gamma, \rho \vdash \rho}
}

\newcommand{\TypingRules}{
\begin{figure}[ht]
\small
\[
\begin{array}{c}
    \judgment{\judg{e}{\tau}}
\\[0.4cm]
    \TypeMultLamIntro
\qquad
    \TypeMultLamElim
\\[0.2cm]
    \TypeLamIntroL
\qquad
    \TypeLamIntroW
\\[0.2cm]
    \TypeVarDelta
  \qquad
    \TypeVarSplit
\\[0.2cm]
    \TypeVarP
  \qquad
    \TypeVarOmega
\qquad
    \TypeLamElimL
\\[0.2cm]
    \TypeLinearVar
\qquad
  \TypeLamElimW
\\[0.2cm]
  \TypeConstr
\\[0.2cm]
    \TypeLet
\quad
    \TypeLetRec
\\[0.2cm]
    \TypeCaseWHNF
\\
    \TypeCaseNotWHNF
\\[1cm]
    \ROUNDTWO{\judgment{\judg{\rho \Rightarrow e}{ \s}[][\Delta_s][z]}}
\\[0.4cm]
    \TypeAltNWHNF
\quad
    \TypeAltNNotWHNF
\\[0.2cm]
    \TypeAltZero
\qquad
    \TypeAltWild
\end{array}
\]
\caption{Linear Core Type System}
\label{fig:linear-core-typing-rules}
\end{figure}
}

\newcommand{\TypingRulesOther}{
\begin{figure}[ht]
\small
\[
\begin{array}{c}
    \judgment{\Gamma \vdash_{mult} \pi}
\\[0.4cm]
    \TypeWellFormedMult
\\[1cm]
    \judgment{\G; \D \Vdash e : \s \gtrdot \ov{\D_i}}
\\[0.4cm]
    \TypeWHNFCons
\qquad
    \TypeWHNFLam
\end{array}
\]
\caption{Linear Core Auxiliary Judgements}
\label{fig:linear-core-other-judgements}
\end{figure}
}

\newcommand{\WHNFConvSoundness}{
\begin{restatable}[Irrelevance]{lemma}{irrelevancelemma}~\label{lem:irrev}
If $\Gamma, \z[\irr{\D}]; \irr{\Delta}, \D' \vdash \ROUNDTWO{\rho \Rightarrow_\mathsf{NWHNF}}\,\,  e :^z_{\irr{\D}} \sigma $
then $\Gamma, \z[\D^\dag]; \D^\dag, \D' \vdash \ROUNDTWO{\rho \Rightarrow e} :^z_{\D^\dag} \sigma$, for \emph{any} $\D^\dag$
\end{restatable}
}

\newcommand{\DeltaLinearRelationLemma}{
\begin{assumption}[$\D \Rightarrow 1$]
A $\Delta$-variable can replace its usage environment $\D$ as a linear variable if $\D$ is decidedly consumed through it.\\
If $\G,\xD; \D,\D' \vdash e : \s$ and $\D$ is consumed through $\xD$ in $e$
then $\G[x/\D]; \D',\xl \vdash e :\s$.
\end{assumption}
}

\newcommand{\LinearDeltaRelationLemma}{
\begin{assumption}[$1 \Rightarrow \D$]
A linear variable can be moved to the unrestricted context as a $\D$-var with usage environment $\D$ by introducing $\D$ in the linear resources.\\
If $\G; \D',\xl \vdash e :\s$
then $\G[\D/x],\xD; \D,\D' \vdash e : \s$.
\end{assumption}
}

\newcommand{\DeltaUnrestrictedRelationLemma}{
\begin{assumption}[$x{:}_\omega\sigma = x{:}_{\cdot}\sigma$]
An unrestricted variable is equivalent to a $\D$-var with an empty usage environment.\\
$\G,\xo; \D \vdash e : \s$ iff $\G,\x[\cdot]; \D \vdash e : \s$
\end{assumption}
}

\newcommand{\LinearSubstitutionLemma}{
\begin{restatable}[Substitution of linear variables preserves typing]{lemma}{linearsubstlemma}~
  \begin{enumerate}
  \item If $\judg[\G][\D,\xl]{e}{\vp}$ and $\judg[\G][\D']{e}{\s}$
    then $\judg[\subst{\G}{\D'}{x}][\D,\D']{e[e'/x]}{\vp}$
  \item If $\judg[\G][\D,\xl]{\rho\,\ROUNDTWO{\Rightarrow}\, e}{\s }[][\D_s][z]$ and
    $\judg[\G][\D']{e'}{\tau}$ and
    $\D_s \subseteq \D,x$ then
    $\judg[\subst{\G}{\D'}{x}][\D,\D']{\rho \,\ROUNDTWO{\Rightarrow}\, e[e'/x]}{\s }[][\subst{\D_s}{\D'}{x}][z]$
  \end{enumerate}
\end{restatable}
}

\newcommand{\UnrestrictedSubstitutionLemma}{
\begin{restatable}[Substitution of unrestricted variables preserves
  typing]{lemma}{unrestrictedsubstlemma}~\\
  \begin{enumerate}
\item If $\Gamma, x{:}_\omega\sigma; \Delta \vdash e : \varphi$ and
  $\Gamma; \cdot \vdash e' : \sigma$ then $\G,\D \vdash e[e'/x] :
  \varphi$
 \item If $\G, \xo; \D \vdash \rho \Rightarrow  e :^z_{\Delta_s} \s$ and $\G; \D \vdash e' : \tau$
    then $\G; \D \vdash\rho \Rightarrow e[e'/x] :^z_{\Delta_s} \s$
  \end{enumerate}

\end{restatable}
}

\newcommand{\DeltaSubstitutionLemma}{
\begin{restatable}[Substitution of $\Delta$-variables preserves
  typing]{lemma}{deltasubstlemma}~\\
  \begin{enumerate}
\item If $\G,\xD; \D, \D' \vdash e : \varphi$ and $\G; \D \vdash e' :
  \sigma$ then $\G; \D, \D' \vdash e[e'/x] : \varphi$
 \item If $\G,\xD; \D,\D' \vdash \rho \Rightarrow e :^z_{\Delta_s} \s'$ and
    $\G; \D \vdash e' : \s$ and $\Delta_s \subseteq (\Delta,\Delta')$ then $\G; \D, \D' \vdash \rho \Rightarrow e[e'/x] :^z_{\Delta_s} \s' $
  \end{enumerate}

\end{restatable}
}

\newcommand{\BetaReductionTheorem}{
\begin{theorem}[$\beta$-reduction preserves types]~\\
If $\G; \D \vdash (\lambda \x[\pi][\s].~e)~e' : \vp$ then $\G; \D \vdash e[e'/x] : \vp$
\end{theorem}
}

\newcommand{\BetaReductionSharingTheorem}{
\begin{theorem}[$\beta$-reduction with sharing preserves types]~\\
    If $\G; \D \vdash (\lambda \xo.~e)~e' : \vp$ then $\G; \D \vdash \llet{x = e'}{e} : \vp$\\
    NB: We only apply this transformation when $x$ is unrestricted, otherwise beta-reduction without sharing is the optimization applied.
\end{theorem}
}

\newcommand{\BetaReductionMultTheorem}{
\begin{theorem}[$\beta$-reduction on multiplicity abstractions preserves types]~\\
    If $\G; \D \vdash (\Lambda p.~e)~\pi : \vp$ then $\G; \D \vdash e[\pi/p] : \vp$
\end{theorem}
}

\newcommand{\BinderSwapTheorem}{
\begin{theorem}[Binder-swap preserves types]~\\
If $\G; \D \vdash \ccase{x}{z~\{\ov{\rho_i\Rightarrow  e_i}\}} : \vp$ then $\G; \D \vdash \ccase{x}{z~\{\ov{\rho_i\Rightarrow e_i[z/x]}\}} : \vp$
\end{theorem}
}

\newcommand{\InliningTheorem}{
\begin{theorem}[Inlining preserves types]~\\
If $\G; \D,\D' \vdash \llet{\xD = e}{e'} : \vp$ then $\G; \D, \D' \vdash \llet{\xD = e}{e'[e/x]} : \vp$
\end{theorem}
}

\newcommand{\FullLazinessTheorem}{
\begin{theorem}[Full-laziness preserves types]~\\
    If $\G; \D,\D' \vdash \lambda \y[\pi].~\llet{\xD = e}{e'} : \vp$ and $y$ does not occur in $e$ \\then $\G;\D,\D' \vdash \llet{\xD=e}{\lambda \y[\pi].~e'}$
\end{theorem}
}

\newcommand{\LocalTransformationsTheorem}{
\begin{theorem}[commuting lets preserve types]
\[
  \begin{array}{llcl}
  1. & \G; \D \vdash (\llet{v = e}{b})~a : \vp & \Rightarrow & \G; \D \vdash \llet{v = e}{b~a} : \vp\\
  2. & \G; \D \vdash \ccase{(\llet{v = e}{b})}{\dots} : \vp & \Rightarrow & \G; \D \vdash \llet{v = e}{\ccase{b}{\dots}} : \vp\\
  3. & \G; \D \vdash \llet{x = (\llet{v = e}{b})}{c} : \vp & \Rightarrow & \G; \D \vdash \llet{v = e}{\llet{x = b}{c}} : \vp\\
  \end{array}
  \]
\end{theorem}
}

\newcommand{\CaseOfKnownConstructorTheorem}{
\begin{theorem}[Case-of-known-constructor preserves types]~\\
If $\G; \D, \D' \vdash \ccase{K~\ov{e}}{\z[\D][\s]~\{..., K~\ov{x} \Rightarrow e_i\}} : \vp$ then $\G; \D,\D' \vdash e_i\ov{[e/x]}[K~\ov{e}/z] : \vp$
\end{theorem}
}

\newcommand{\CaseOfCaseTheorem}{
\begin{restatable}[Case-of-case preserves types]{theorem}{caseofcasethm}~\\
    If $\G; \D, \D',\D'' \vdash \ccase{\ccase{e_c}{\z[\D]~\{\ov{\rho_{c_i} \Rightarrow e_{c_i}}\}}}{\var[w][\irr{\D,\D'}][\s']~\{\ov{\rho_i \Rightarrow e_i}\}} : \vp$\\
    then $\G; \D, \D',\D'' \vdash \ccase{e_c}{\z[\D]~\{\ov{\rho_{c_i} \Rightarrow \ccase{e_{c_i}}{w~\{\ov{\rho_i \Rightarrow e_i}\}}}\}} : \vp$\\
\end{restatable}
}

\newcommand{\EtaExpansionTheorem}{
\begin{theorem}[$\eta$-expansion preserves types]~\\
    If $\G; \D \vdash f : \s \to_\pi \vp$ then $\G; \D \vdash \lambda x.~f~x : \s \to_\pi \vp$
\end{theorem}
}

\newcommand{\EtaReductionTheorem}{
\begin{theorem}[$\eta$-reduction preserves types]~\\
    If $\G; \D \vdash \lambda x.~f~x : \s \to_\pi \vp$ then $\G; \D \vdash f : \s \to_\pi \vp$
\end{theorem}
}



\setcopyright{cc}
\setcctype{by}
\acmDOI{10.1145/3776711}
\acmYear{2026}
\acmJournal{PACMPL}
\acmVolume{10}
\acmNumber{POPL}
\acmArticle{69}
\acmMonth{1}
\received{2025-07-10}
\received[accepted]{2025-11-06}

\begin{document}

\title{Lazy Linearity for a Core Functional Language}

\author{Rodrigo Mesquita}
\orcid{0009-0005-1467-115X}
\affiliation{\institution{Well-Typed LLP}
  \city{London}
  \country{United Kingdom}
}
\email{rodrigo@well-typed.com}

\author{Bernardo Toninho}
\orcid{0000-0002-0746-7514}
\affiliation{\institution{Instituto Superior Técnico, University of Lisbon}
  \city{Lisbon}
  \country{Portugal}
}
\affiliation{\institution{INESC-ID}
  \city{Lisbon}
  \country{Portugal}
}
\email{bernardo.toninho@tecnico.ulisboa.pt}

\begin{abstract}
Traditionally, \ROUNDTWO{in linearly typed languages, consuming} a linear resource is synonymous
with its syntactic occurrence in the program. However, under the lens of
non-strict evaluation, linearity can be further understood semantically,
where a syntactic occurrence of a resource does not necessarily entail
using that resource when the program is executed.
\ROUNDTWO{While this distinction has been largely unexplored, it turns
  out to be inescapable in} Haskell's optimising compiler, which heavily rewrites the source
    program in ways that break syntactic linearity but preserve the program's
    semantics.
We introduce Linear Core, \ROUNDTWO{a novel system which accepts the lazy
semantics of linearity statically} and
is suitable for \ROUNDTWO{lazy languages such as } the Core intermediate language of the
Glasgow Haskell Compiler. We prove \ROUNDTWO{that} Linear Core is sound,
\ROUNDTWO{guaranteeing} linear resource usage, and \ROUNDTWO{that} multiple
optimising transformations preserve linearity in Linear Core while failing
to do so in Core. We have implemented Linear Core as a compiler plugin to
validate the system against linearity-heavy libraries, including
\texttt{linear-base}.\end{abstract}

\maketitle

\section{Introduction}\label{sec:intro}

Linear types~\cite{cite:linear-logic,cite:barberdill} increase the
expressiveness of type systems by \ROUNDTWO{guaranteeing that certain
resources} are used \emph{exactly once}.
In programming languages with a linear type system, \ROUNDTWO{discarding so-called linear} resources,
or using them twice, is flagged as a type error. Linear types can \ROUNDTWO{be used to}, for instance,
statically enforce correct usage of file descriptors, that heap-allocated
memory is freed exactly once \ROUNDTWO{s.t.}~leaks and double-frees become type errors,
or guarantee deadlock freedom in channel-based communication protocols~\cite{10.1007/978-3-642-15375-4_16},
among other correctness properties~\cite{10.1145/3373718.3394765,10.1145/3527313,cite:linearhaskell}.
The following program with a runtime error (allocated memory is freed twice) is
accepted by a C-like type system. Conversely, under the lens of a linear type
system where $p$ is deemed a linear resource created by the call to
\texttt{malloc}, the program is rejected:
\begin{code}
let p = malloc(4) in free(p); free(p);
\end{code}

Despite their promise and extensive presence in research
literature~\cite{Wadler1990LinearTC,CERVESATO2000133,10.1093/logcom/2.3.297},
an effective design combining linear and non-linear typing is both
challenging and necessary to bring the advantages of linear typing to
mainstream languages.
Consequently, few general purpose programming languages have linear or substructural
type systems. Among them are Idris~2~\cite{brady:LIPIcs.ECOOP.2021.9},
Rust~\cite{10.1145/2692956.2663188}, whose
ownership types build on linear types to guarantee memory safety
without garbage collection, and, more recently,
Haskell~\cite{10.1145/3158093}, a pure, functional, and
\emph{lazy} general purpose language.

Linearity in Haskell stands out from linearity in other languages because
linear types permeate Haskell down to Core, the intermediate language into
which Haskell is translated by the Glasgow Haskell Compiler (GHC).
Core is a minimal, explicitly typed, \ROUNDTWO{pure} lazy functional language with both
linear and unrestricted types, to which multiple Core-to-Core optimising
transformations are applied during compilation.
Notably, Core \ROUNDTWO{is typechecked after each transformation}. This serves
as a sanity check to the correction of \ROUNDTWO{the implementation} of
\ROUNDTWO{such} transformations, since unsound ones are likely to introduce
ill-typed expressions.

Just as Core's type system \cite{10.1145/1190315.1190324} provides a
degree of validation to the 
translation from Haskell, dubbed \emph{desugaring}, and the subsequent
optimising transformations, a linearly typed Core would guarantee that
linear resource usage in the source language is not violated by desugaring
and the compiler optimisation passes. Moreover, linearity information in
Core can inform Core-to-Core optimising 
transformations~\cite{cite:let-floating,peytonjones1997a,10.1145/3158093},
such as inlining and $\beta$-reduction, to produce more efficient programs.
\ROUNDTWO{However}, despite the advantages of \ROUNDTWO{having} linear types in Core, the status quo is
that linearity is effectively ignored in Core, being only checked
in the source Haskell code. The reason is that, \ROUNDTWO{after} desugaring, the
Core-to-Core optimising transformations eliminate and rearrange most of the
syntactic constructs through which linearity checking is performed -- often
resulting in programs that are syntactically very different from the original,
especially due to the flexibility provided by laziness with regard to the
optimisations a compiler may perform.

While compiler optimisations aim to preserve the semantics of
programs and therefore a program's resource consumption,
a naive syntactic analysis of the optimised programs as
performed by the current Core type checker often fails to recognize
\ROUNDTWO{resulting} programs as linear.
As a trivial example, let $x$ be a linear resource in the next two expressions,
where the latter results from inlining $y$ without eliminating the
binder -- a common optimising transformation in GHC. Even
though the second expression no longer appears linear (as there are now two occurrences
of $x$), it is indeed linear because the let-bound expression is never
evaluated under \ROUNDTWO{non-strict} evaluation and $x$ is still consumed exactly once (when it is freed in the let body):
\[
\begin{array}{ccc}
\llet{y = \textsf{free}~x}{y} & \Longrightarrow_{Inlining} & \llet{y = \textsf{free}~x}{\textsf{free}~x}
\end{array}
\]

The Core optimising transformations expose a fundamental limitation of Core's
\ROUNDTWO{linear} type system: it does not account for the non-strict evaluation model
of Core and, thus, a whole class of programs that are linear under the lens of
non-strict evaluation are rejected by Core's current \ROUNDTWO{checker}.
In this work, we address this limitation by exploring how syntactic linearity
breaks down in the presence of non-strictness and design a linear type
system for Core which accepts the vast majority of programs and optimising
transformations that, while correct under non-strict evaluation, are rejected
by traditional linear typing disciplines.
Additionally, we show \ROUNDTWO{that} our system
is suitable for the intermediate language of an
optimising compiler by implementing it as a Glasgow Haskell Compiler plugin.
Our contributions are as follows:
\begin{itemize}

\item We \ROUNDTWO{illustrate} the \ROUNDTWO{lazy semantics} of linearity, in
contrast to the \ROUNDTWO{traditional} syntactic \ROUNDTWO{understanding
of} linearity, in Haskell, by example (Section~\ref{sec:linearity-semantically}).
A precise \ROUNDTWO{semantic definition of} linearity, \ROUNDTWO{closely
following that of Linear Haskell~\cite{cite:linearhaskell}}, is given by
our instrumented operational semantics under which a program using a
linear resource non-linearly gets stuck
(Section~\ref{sec:main:metatheory}).

\item We introduce Linear Core (Section~\ref{sec:main:linear-core}), a \ROUNDTWO{lazy} linear language with
the key features of Core (except for type equality coercions), whose type
system \ROUNDTWO{is designed to accept the lazy} semantics of
linearity. \ROUNDTWO{Our system achieves this via a novel combination
  of features: explicit tracking of captured linear resources through
        \emph{usage environments} (Section~\ref{sec:usage-environments});
  a distinct treatment of case
  scrutinees in weak-head normal form (Section~\ref{sec:onwhnf});
  and a notion of irrelevant and tagged resources
  (Section~\ref{sec:irrev}), which allow for fine-grained control over
  variable usage.} 
To the best of our
knowledge, this is the first type system \ROUNDTWO{to accept programs
which depend on non-strict evaluation to be understood as linear}.

\item We prove that our type system is sound, guarantees linear resource usage, and
    prove that multiple optimising transformations, which are rejected in Core,
        are validated by Linear Core (Section~\ref{sec:main:metatheory}).

\item We discuss our implementation of Linear Core as a GHC
  plugin which checks linearity in
all intermediate Core programs produced during the compilation process, showing
it accepts the programs resulting from transformations in libraries such as
\texttt{linear-base} (Section~\ref{sec:discuss:implementation}).

\end{itemize}

\ROUNDTWO{While GHC's Core language is our main focus, we
note that the ideas explored in this work can potentially be applied to any
language with non-strict features, even if the language itself is
strict. We assume some familiarity with Haskell and linear types for
functional languages, but try to explain our work and motivations
without assuming much knowledge of GHC and Core. }

Additional proofs and definitions are given in Appendix.

\section{Linearity, Semantically\label{sec:linearity-semantically}}

Linear type systems guarantee that linear resources are consumed exactly once.
\ROUNDTWO{Thus}, whether a program is \ROUNDTWO{accepted as linear} intrinsically depends on
the definition of consuming a resource.
\ROUNDTWO{Traditionally}, consuming a resource \ROUNDTWO{has been} equated with its 
syntactic occurrence along a control-flow path in the program. However, as this
section \ROUNDTWO{aims to} make clear, \ROUNDTWO{in a non-strict context, a
syntactic definition of linearity is overly conservative}.
For instance, in the following Haskell-like program, \lstinline{f} is a linear
function (a function that must use its argument
exactly once), but there are two syntactic
occurrences of \lstinline{handle}: one is in a let-bound computation that
closes this handle, the other \ROUNDTWO{is in the then branch of a
  conditional}. If the program is 
evaluated lazily, the handle will, \ROUNDTWO{dynamically, either be
  used exactly once (i.e,~closed) in
the else branch, or be used exactly once (i.e.,~returned) in the then branch}, since the 
\ROUNDTWO{let-bound variable is not needed in that branch and
so the expression which closes the handle is never evaluated}.
\begin{lstlisting}
f : Handle *' $\lolli$ '* Handle
f handle =
    let closed = close handle in
    if <cond>
        then handle
        else closed
\end{lstlisting}

Intuitively, a computation that depends on a linear resource to produce some
result consumes that resource iff the result is fully consumed; in
contrast, a computation that depends on a linear resource, but is never run,
will not consume that resource.
We argue that a linear type system for a language with non-strict evaluation
semantics should accept the above program, unlike a linear type system for the
same program if it were to be evaluated eagerly, where the let-bound
expression would always be evaluated.
Given this observation, we turn our focus to linearity under lazy evaluation
\ROUNDTWO{specifically} since the distinction between semantically and
syntactically consuming a resource is only exposed under non-strict semantics,
and because GHC Core is lazy.

\subsection{Semantic Linearity by Example\label{sec:semantic-linearity-examples}}

This section aims to develop an intuition for semantic linearity through simple
Core programs that can be understood as linear but are rejected by Core's
linear type system (or, in fact, by all linear type systems we are
aware of prior to this work).
\begingroup
\setlength{\fboxsep}{0pt}In the examples, a \colorbox{working}{\workingcolorname} background highlights
programs that are syntactically linear and accepted by Core's linear type
system.
A \colorbox{notyet}{\notyetcolorname} or
\colorbox{limitation}{\limitationcolorname} background mark programs that are
semantically linear, but not seen as linear by Core's type system.
\ROUNDTWO{Our system accepts all \colorbox{notyet}{\notyetcolorname} programs,
but not the \colorbox{limitation}{\limitationcolorname} examples.}
A \colorbox{noway}{\nowaycolorname} background indicates that the program
violates linearity semantically, i.e.~the program effectively
discards or duplicates linear resources.
\endgroup

\paragraph{Let bindings}

The example at the start of Section~\ref{sec:linearity-semantically}
highlighted the subtlety of linearity in a non-strict language with
a single unused let binding.
In general, a let-bound expression that captures a linear variable
can be conditionally needed at runtime, depending on the branch taken.
Both optimising transformations (float-out) and programmers used to
non-strict evaluation are keen to write programs with bindings that are selectively
used in case alternatives, such as in the following program. Note that we write
$\lolli$ for linear functions and \lstinline{->} for non-linear functions:
\begin{notyet}
\begin{lstlisting}
f :: (a *' $\lolli$ '* a) -> Bool -> a *' $\lolli$ '* a
f use bool x = let y = use x in
  case bool of
    True -> x 
    False -> y
\end{lstlisting}
\end{notyet}
This example returns \lstinline{x} in one branch and the let-bound
\lstinline{y} in the other. Semantically, as opposed to syntactically, this
program is linear because the linear resource \lstinline{x} is used exactly once regardless of which branch is
taken: either directly in the \lstinline{True} branch, or \emph{via}
the use of the let-bound \lstinline{y} in the \lstinline{False} branch.

That said, a let-bound \lstinline{y} that captures a linear resource
\lstinline{x} must still be used exactly once, regardless of whether evaluation
inlines the let (\emph{à la} call-by-name) or creates a thunk (\emph{à
  la} call-by-need):
\begin{noway}
\begin{lstlisting}
f1 :: (a *' $\lolli$ '* b) -> a *' $\lolli$ '* (b, b)
f1 use x = let y = use x in (y, y)
\end{lstlisting}
\end{noway}
If \lstinline{y} is inlined, linearity is trivially violated.
On the other hand, if \lstinline{y} becomes a thunk, \lstinline{use x} will
only be evaluated once regardless of how many times \lstinline{y} is used.
Even so, this program must be rejected under call-by-need.
Intuitively, linearity is violated if the computation leaks the resource or
parts of it. For instance, if \lstinline{use = id}, the program must clearly be
rejected.

Lastly, consider an ill-typed program which defines two let bindings \incode{z}
and \incode{y}, where \incode{z} uses \incode{y} which in turn uses the linear
resource \incode{x}:
\begin{noway}
\begin{lstlisting}
f2 :: (a *' $\lolli$ '* a) -> a *' $\lolli$ '* ()
f2 use x = let y = use x in let z = use y in ()
\end{lstlisting}
\end{noway}
Even though \incode{y} occurs in \incode{z}, \incode{x} is never consumed
\ROUNDTWO{because \incode{z} is discarded}. \ROUNDTWO{This example highlights
that naively checking for syntactic occurrences of let-bound variables such as
\incode{y} is also inadequate to track resource usage}.
\ROUNDTWO{A better mental model is that} using \incode{y} implies
using \incode{x}, and thus using \incode{z} implies using \incode{x}
\ROUNDTWO{too}. If neither \incode{z} nor \incode{y} are used, then \incode{x}
is \ROUNDTWO{effectively discarded}.

In essence, an unused let binding doesn't consume any resources, and a let
binding used exactly once consumes the resources it captures exactly once. Let
binders that depend on linear resources must be used \emph{at most once} in
mutual exclusion with the linear resources themselves (in a sense,
let-bound variables are affine in the let body).
We discuss how to encode this principle between let bindings and their dependencies using so called \emph{usage
environments}, in Section~\ref{sec:usage-environments}.

\paragraph{Recursive let bindings\label{sec:semantic-linearity-examples:recursive-lets}}

Recursive let-bindings
behave similarly to non-recursive lets as far as their usage in the let
continuation is concerned.
The challenge lies in understanding linearity in the definitions of the
recursive binders and determining which resources are consumed when one of the
binders in the recursive group is used.
Consider the following program, which calls a recursive let-bound function
defined in terms of the linear resource \incode{x} and itself:
\begin{notyet}
\begin{lstlisting}
f3 :: Bool -> a *' $\lolli$ '* a
f3 bool x = let go b = case b of
                     True -> x
                     False -> go (not b)
                   in go bool
\end{lstlisting}
\end{notyet}
Function \incode{f3} is semantically linear because, iff it is consumed exactly once,
then \incode{x} is consumed exactly once. We can see this by case analysis on \incode{go}'s argument:
when \incode{bool} is \incode{True}, we'll use the resource \incode{x};
when \incode{bool} is \incode{False}, we recurse by calling \incode{go} on \incode{True}, which in turn will use the resource \incode{x}.
In \incode{go}'s body, \incode{x} is used directly in one branch and indirectly in the
other (by recursing, which we know will result in using \incode{x} linearly).

It so happens that \incode{go} will terminate on any input, and will always consume
\incode{x}. However, termination is not a requirement for a binding to use a resource linearly,
and we could have a similar example in which \incode{go} might never terminate but still
uses \incode{x} linearly if evaluated:

\begin{notyet}
\begin{lstlisting}
f4 :: Bool -> a *' $\lolli$ '*  a
f4 bool x = let go b = case b of
                     True -> x
                     False -> go b
                   in go bool
          \end{lstlisting}
        \end{notyet}

The key to linearity in the presence of non-termination is Linear Haskell's
definition of a linear function: \emph{if a linear function application (\incode{f u}) is
consumed exactly once, then the argument (\incode{u}) is consumed exactly once}.
If \incode{f u} doesn't terminate, it is never consumed, thus the claim holds
vacuously.
If \incode{go} doesn't terminate, we aren't able to compute (or consume) the result
of \incode{f4}, so we do not promise anything about \incode{x} being consumed (\incode{f4}'s
linearity holds trivially). If it did terminate, it would consume \incode{x} exactly
once (e.g. if \incode{go} was applied to \incode{True}).

Naturally, a linear type system must statically reject any recursive group that could
lead to non-linear resource usage.
For instance, if the linear resource \incode{x} is used in one case alternative and
recursion is used more than once in another, \incode{x} may end up being consumed
more than once:
\begin{noway}
\begin{lstlisting}
f5 :: Bool -> Bool *' $\lolli$ '* Bool
f5 bool x =
  let go b
        = case b of
           True -> x
           False -> go (not b) && go True
  in go bool
\end{lstlisting}
\end{noway}

To determine whether the usage of a binding from a recursive group is linear,
we must \ROUNDTWO{check} the superset of resources that may end up being consumed in a
single use of any of the bindings. For a strongly connected recursive group,
this superset is the same for all bindings by definition.
It follows from the resources used being the same for all bindings
that at most a single variable from the recursive group can be used in the
continuation.
The treatment of recursive let bindings in our system is discussed in Section~\ref{sec:recursivelets}.

\paragraph{Case expressions}
We have observed that more programs can be accepted as linear by delaying
resource consumption in (recursive) let bindings until they are evaluated.
On the other hand, \emph{case expressions drive evaluation} -- and, therefore,
do consume resources.

To understand exactly how,
we turn to the definition of consuming a resource from Linear
Haskell~\cite{10.1145/3158093}:
to consume a value of atomic base type (such as~\texttt{Int} or
        \texttt{Ptr}) exactly once, just evaluate it;
to consume a function value exactly once, apply it to one argument,
        and consume its result exactly once;
to consume a value of an algebraic datatype exactly once,
        pattern-match on it, and consume all its linear components exactly once.
Thus, we can consume a linear resource by fully evaluating it, and evaluation
happens with case expressions.

In Core, cases are of the form $\ccase{e_s}{z~\{\ov{\rho_i \to e_i}\}}$,
where $e_s$ is the case \emph{scrutinee}, $z$ is the case \emph{binder}, and
$\ov{\rho_i \to e_i}$ are the case \emph{alternatives}, composed of a pattern
$\rho_i$ and of the corresponding expression $e_i$. 
The case scrutinee is always evaluated to Weak Head Normal Form (WHNF) and
the case binder is an alias to the result of evaluating the scrutinee
to WHNF. Moreover, in Core, the alternative patterns are always exhaustive.
\ROUNDTWO{The first} example program constructs a tuple from linear resources,
matches on it, then uses both linearly-bound variables from the tuple pattern.
This is well-typed in Linear Haskell:
\begin{working}
\begin{lstlisting}
f6 :: a *' $\lolli$ '* b *' $\lolli$ '* (a *' $\lolli$ '* b *' $\lolli$ '* c) -> c
f6 x y use = case (x,y) of z { (a,b) -> use a b }
\end{lstlisting}
\end{working}
\ROUNDTWO{Perhaps surprisingly}, a similar program which discards the pattern
variables, and instead uses the resources \ROUNDTWO{from} the scrutinee
\ROUNDTWO{again}, is still linear, despite not being accepted by Linear Haskell:
\begin{notyet}
\begin{lstlisting}
f7 :: a *' $\lolli$ '* b *' $\lolli$ '* (a *' $\lolli$ '* b *' $\lolli$ '* c) -> c
f7 x y use = case (x,y) of z { (a,b) -> use x y }
\end{lstlisting}
\end{notyet}
\incode{f7} \ROUNDTWO{can be seen as} linear by \ROUNDTWO{realizing} that
the tuple being scrutinized is already in WHNF, so evaluating it again does
not consume either $x$ \ROUNDTWO{nor} $y$. Even if the tuple \ROUNDTWO{paired}
two \ROUNDTWO{arbitrary} expressions which \ROUNDTWO{captured} $x$ and $y$,
\ROUNDTWO{because $a$ and $b$ are not forced in the continuation, $x$ and $y$ would remain unused.}
\ROUNDTWO{On the other hand, as seen in the next example}, if $a$ were used in 
the case body, then $x$ \ROUNDTWO{is consumed via $a$ and must not be used again.}
Thus, \incode{f8} must be rejected:
\begin{noway}
\begin{lstlisting}
f8 :: a *' $\lolli$ '* b *' $\lolli$ '* (a *' $\lolli$ '* a *' $\lolli$ '* c) -> c
f8 x y use = case (x,y) of z { (a,b) -> use a x }
\end{lstlisting}
\end{noway}
This idea that $x$ and $a$ are mutually exclusive is \ROUNDTWO{akin to how} let
bindings \ROUNDTWO{are} mutually exclusive \ROUNDTWO{with the set} of resources
\ROUNDTWO{they capture, in the let-continuation}.
\ROUNDTWO{Forcing a pattern variable evaluates the expression for the
corresponding constructor field, and all linear variables captured by that
expression may end up being used}.
A third option, semantically \ROUNDTWO{linear}, but rejected by Linear Haskell, is to use
the case binder $z$ instead of $a,b$ or $x,y$:
\begin{notyet}
\begin{lstlisting}
f9 :: a *' $\lolli$ '* b *' $\lolli$ '* (a *' $\lolli$ '* b *' $\lolli$ '* c) -> c
f9 x y use = case (x,y) of z { (a,b) -> uncurry use z }
\end{lstlisting}
\end{notyet}
Again, $z$ is mutually exclusive with $(a,b)$ and with $(x,y)$, but at least
one of the three must \ROUNDTWO{be used} to ensure the
linear resources are consumed. Essentially, in this example, using $a$ entails
using the resource $x$, $b$ the resource $y$, and the case binder $z$ entails
using both $a$ and $b$.

\ROUNDTWO{Conversely}, if the scrutinee is an expression that's not in WHNF,
evaluation to WHNF \ROUNDTWO{may indeed} consume \ROUNDTWO{some or all of the}
linear resources \ROUNDTWO{captured by said scrutinee}. \ROUNDTWO{Consider the following
non-WHNF scrutinee example},
\incode{f10}, \ROUNDTWO{which is accepted} by our system:
\begin{notyet}
\begin{lstlisting}
f10 :: a *' $\lolli$ '* b *' $\lolli$ '* (a *' $\lolli$ '* b *' $\lolli$ '* (c,d)) -> (c,d)
f10 x y use = case use x y of z { (a,b) -> z }
\end{lstlisting}
\end{notyet}
\ROUNDTWO{In contrast to} a scrutinee in WHNF, \ROUNDTWO{the resources captured
by a non-WHNF scrutinee} (here, $x,y$), \ROUNDTWO{conservatively, can no longer
be used} in the case alternatives, \ROUNDTWO{since reducing \incode{use x y} to
WHNF may consume (parts of) $x$ and/or $y$}. \ROUNDTWO{Consequently}, either the case
binder $z$ or the linear pattern variables $a,b$ must be used exactly
once,
since the linear resources in the scrutinee are \ROUNDTWO{only} fully consumed
if the result of evaluating the scrutinee \ROUNDTWO{is also} fully consumed . For instance, take
\incode{use} in \incode{f10} to be the pairing function \incode{(,)}: it is not
sufficient for \incode{use x y} to be evaluated to WHNF to consume \incode{x}
and \incode{y}.
If all the resources were considered to be fully consumed after
\ROUNDTWO{evaluating the} scrutinee to \ROUNDTWO{WHNF}, the pattern variables
could simply be ignored and thus linear resources potentially discarded.
In \ROUNDTWO{practice}, if the scrutinee is not in WHNF then \ROUNDTWO{either}
the case binder or the linear components of the pattern \ROUNDTWO{must be consumed}.

\ROUNDTWO{Alternatively}, if a constructor without any linear components \ROUNDTWO{is matched},
all linear resources used in the scrutinee must \ROUNDTWO{already} have been
fully consumed in that branch since the result of evaluating it can be consumed
unrestrictedly (by definition, a constructor application to unrestricted
arguments is also unrestricted).
Hence, in the \ROUNDTWO{case-}branch for a constructor without linear fields, the case
binder can also be used unrestrictedly. For example, the program in
Figure~\ref{fig:nolin} is semantically linear.
\begin{figure}[h]
 \begin{minipage}{0.5\textwidth}
\begin{notyet}
\begin{lstlisting}
f11 :: () *' $\lolli$ '* ()
f11 x = case x of z { () -> z <> z }
\end{lstlisting}
\end{notyet}
\vspace{-0.5cm}
\caption{No Linear Fields\label{fig:nolin}}
\end{minipage}\begin{minipage}{0.5\textwidth}
\begin{limitation}
\begin{lstlisting}
f12 :: a *' $\lolli$ '* a
f12 x = case K1 x of z { K2 -> x; K1 a -> x }
\end{lstlisting}
\end{limitation}
\vspace{-0.5cm}
\caption{Absurd Branches\label{fig:absurd}}
\end{minipage}
\end{figure}
A second example of an unrestricted pattern, where \incode{K2} has no fields
and \incode{K1} has one linear field, is shown in Figure~\ref{fig:absurd}.
The use of \incode{x} in the \incode{K2} branch is valid despite \incode{x}
also occurring in the scrutinee because this branch is never taken. In fact,
any arbitrary resource could be used in absurd branches and be understood as
linear (despite not being seen as such by our system).

While many of these examples may seem artificial due to the pattern match on a
known constructor, we note that in between the transformations programs
undergo in an optimising compiler, many such programs naturally occur (e.g.,~if
the definition of a function is inlined in the scrutinee).

As for \emph{default} case alternatives, also known as \emph{wildcards}
(written $\_$): matching a wildcard doesn't provide any linearity information,
but the
scrutinee is still evaluated to WHNF.
When matching a wildcard, if the scrutinee is already in WHNF,
programs that either use the resources from the scrutinee directly, or the case
binder in that alternative, can be seen as linear. If the scrutinee is not in
WHNF, we \emph{must} use the case binder, as it's the only way to linearly
consume the result of evaluating the scrutinee.

\section{A Type System for Semantic Linearity in Core\label{sec:main:linear-core}}

In this section, we develop a pure linear calculus $\lambda_\Delta^\pi$,
dubbed \emph{Linear Core}, that combines the linearity-in-the-arrow
and multiplicity polymorphism introduced by Linear
Haskell~\cite{10.1145/3158093} with all the key features of GHC Core: algebraic
datatypes, case expressions, and recursive lazy let bindings; except for equality
coercions, which we discuss as future work in Section~\ref{sec:future-work}.

Linear Core makes precise the various insights discussed in the
previous section in a linear type system that we prove \ROUNDTWO{to be} sound and guarantee
linear resource usage at runtime.
This result is obtained (see Section~\ref{sec:main:metatheory}) via a soundness argument on a
call-by-need big-step operational semantics that becomes stuck when linear
variables in the runtime heap are used more than once, \ROUNDTWO{a
  technical approach that is natural in the literature~\cite{10.5555/1076265,10.1145/3471874.3472980}}. We show this
semantics to be bisimilar to a
Launchbury-style~\cite{10.1145/158511.158618} natural semantics. 
Crucially,
Linear Core typing accepts all the semantically linear examples
(highlighted with \colorbox{notyet}{\notyetcolorname})
from Section~\ref{sec:semantic-linearity-examples}, which Core currently
rejects.
Additionally, we prove that multiple optimising Core-to-Core transformations
preserve linearity in Linear Core, while violating it under Core's
current type system.
Despite the focus on GHC Core, the fundamental ideas for \ROUNDTWO{typing} linearity in
our call-by-need calculus may be applied to other languages \ROUNDTWO{which
combine non-strict features or semantics with linearity}.

\ROUNDTWO{The key ideas that make Linear Core work are:}

\ROUNDTWO{
\begin{enumerate}
    \item Usage environments (Section~\ref{sec:usage-environments})
    \item Distinct treatment of case scrutinees in WHNF (Section~\ref{sec:onwhnf})
    \item Irrelevant and tagged resources (Sections~\ref{sec:irrev}~and~\ref{sec:tagging})
\end{enumerate}

\noindent (1) Encodes the idea that lazy bindings do not consume resources upon
definition, but rather when the bound variables themselves are consumed.
Consuming a variable with some usage environment $\Delta$ equates to consuming
the variables contained in $\Delta$.
(2) Case expressions in WHNF may capture ambient linear resources, but these
resources can be safely used in branches since the case will not further
evaluate its scrutinee. On the other hand, scrutinizing a non-WHNF expression
must be treated more conservatively.
(3) Irrelevant resources make existing bindings unusable while forcing the use
of others. Namely, linear resources occurring in non-WHNF case scrutinees can
no longer be used in the alternatives, while the pattern variables (or the case
binder) must necessarily be used. Tagged resources follow the observation that
some variables must be used jointly or not at all.}

\ROUNDTWO{The reader may wonder whether the complexities of our
  system hinted above are a better alternative to simply specializing
  the optimisation passes to produce code that passes a more naive
  linearity check. This avenue was already explored during the
  implementation of Linear
  Haskell in GHC,
 exhibiting the prohibitive cost of too often restricting
 optimisations and producing less efficient code. Our work adheres to
 the principle that the type system bends to the optimisations, not
 the other way around\footnote{\url{https://gitlab.haskell.org/ghc/ghc/-/blob/c85c845dc5ad539bf28f1b8c5c1dbb349e3f3d25/compiler/GHC/Core/Lint.hs\#L3225-3235}}.}

We present Linear Core's syntax and type system incrementally, starting with the
judgements and base linear calculi rules (Section~\ref{sec:base-calculi}).
Then, usage environments, the rule for $\Delta$-bound variables, and rules for
(recursive) let bindings (Section~\ref{sec:usage-environments}).
Finally, we introduce the rules to type case expressions and case alternatives,
along with the key insights to do so, namely \ROUNDTWO{discriminating by} the WHNF-ness of
the scrutinee, irrelevant resources, and tagging (Section~\ref{sec:lc-case-exps}).

\subsection{Language Syntax}

The complete syntax of Linear Core is given in Figure~\ref{fig:full-linear-core-syntax}.
The types of Linear Core are algebraic datatypes, function types, and
multiplicity schemes to support multiplicity polymorphism: datatypes
($T~\ov{p}$) are parametrised by multiplicities, function types
($\vp\to_\pi\s$) are also annotated with a multiplicity, which can be $1$,
$\omega$ (read \emph{many}), or a multiplicity variable $p$ introduced by a
multiplicity universal scheme ($\forall p.~\vp$).

Terms are variables $x,y,z$, data constructors $K$, multiplicity
abstractions $\Lambda p.~e$ and applications $e~\pi$, term abstractions
$\lambda \x[\pi].~e$ and applications $e~e'$, where lambda binders are
annotated with a multiplicity $\pi$ and a type $\s$. Then, there are
non-recursive let bindings $\llet{\xD = e}{e'}$, recursive let bindings
$\lletrec{\ov{\xD = e}}{e'}$, where the overline denotes a set of distinct
bindings $x_1{:}_{\D}\s_1\dots x_n{:}_{\D}\s_n$ and associated expressions
$e_1\dots e_n$, and case expressions $\ccase{e}{\zD~\{\ov{\rho \Rightarrow e'}\}}$,
where $z$ is the case binder and the overline denotes a set of distinct
patterns $\rho_1\dots \rho_n$ and corresponding right hand sides $e'_1\dots
e'_n$. We often omit the case binder $z$ when it is unused. Notably, (recursive) let-bound
binders and case-bound binders are annotated
with a so-called \emph{usage environment} $\Delta$ -- a fundamental construct
for typing semantic linearity in the presence of laziness that we present in
Section~\ref{sec:usage-environments}.
Case patterns $\rho$ can be either the \emph{default/wildcard}
pattern $\_$, which matches any expression, or a constructor $K$ and a set of
variables that bind its arguments, where each field of the constructor has an
associated multiplicity denoting which pattern-bound variables must be
consumed linearly or unrestrictedly.
Additionally, the set of patterns in a case expression is assumed to be exhaustive,
i.e. there is always at least one pattern which matches the scrutinized expression.

Datatype declarations $\datatype{T~\overline{p}}{\overline{K:\overline{\sigma
\to_\pi}~T~\overline{p} }}$ have the name of the type being declared $T$
parametrized over multiplicity variables $\ov{p}$, and a set of the data
constructors $K$ with signatures indicating the type and multiplicity of the
constructor arguments. Note that a linear resource is used many times when a
constructor with an unrestricted field is applied to it, since, dually, pattern
matching on the same constructor with an unrestricted field allows it to be
used unrestrictedly.

\SyntaxFull

\subsection{Base Typing System\label{sec:base-calculi}}

Linear Core ($\lambda^\pi_\Delta$) is a linear lambda calculus akin to Linear
Haskell's $\lambda^q_\to$ in that both have multiplicity polymorphism,
(recursive) let bindings, case expressions, and algebraic data types.
\ROUNDTWO{We use the $\pi$ superscript to stand for multiplicity
  and the $\Delta$ subscript to emphasize usage
  environments, a key feature of our work.}
$\lambda^\pi_\Delta$ diverges from $\lambda^q_\to$ primarily when typing lets,
case expressions, and alternatives.
The core rules of the calculus for abstraction and application are similar to
that of $\lambda^q_\to$ and mostly standard. 
We note that we handle multiplicity polymorphism differently from
Linear Haskell by \ROUNDTWO{effectively} ignoring the multiplicity semiring and
treating all multiplicity polymorphic \ROUNDTWO{variables}
\ROUNDTWO{parametrically}, for the sake of simplicity. \ROUNDTWO{Such
  variables cannot be duplicated in our system.}
The full type system is given in Figure~\ref{fig:linear-core-typing-rules},
with auxiliary judgements given in
Figure~\ref{fig:linear-core-other-judgements}.

\TypingRules
\TypingRulesOther

The main judgement is written $\G;\D \vdash e : \tau$ to denote that expression
$e$ has type $\tau$ under the unrestricted environment $\G$ and linear
environment $\D$.
Variables in $\G$ can be freely discarded (\emph{weakened}) and duplicated
(\emph{contracted}), while resources in $\D$ must be used exactly once. Despite
not having explicit weakening and contraction rules in our system, they are
available as admissible rules for $\G$ (but not for $\D$), since, equivalently
\cite{91621fae-5e53-3497-8291-32b2fab5a743}, resources from $\G$ are duplicated
for sub-derivations and may unrestrictedly exist in the variable rules.

Occurrences of unrestricted variables from $\G$ are well-typed as long as the linear
environment is empty, while occurrences of linear \ROUNDTWO{(or multiplicity polymorphic)} variables
are only well-typed
when the variable being typed is the only resource available in the
linear context (rules $Var_\omega$, $Var_1$ \ROUNDTWO{and $Var_p$}, respectively).
In \ROUNDTWO{the $Var_1$ and $Var_p$ cases}, the linear context must contain exactly \ROUNDTWO{the
  appropriate variable binding}, whereas the unrestricted context may contain arbitrary
variables.
Variables in contexts are annotated with their type and multiplicity, so
$\x[\pi]$ is a variable named $x$ of type $\s$ and multiplicity $\pi$.

Linear functions are introduced via the function type ($\s \to_\pi \vp$) with
$\pi = 1$, i.e. a function of type $\s \to_1 \vp$ (or, equivalently, $\s \lolli \vp$)
introduces a linear resource of type $\s$ in the linear environment $\D$ to then type an expression of type $\varphi$.
Unrestricted functions are introduced via the function type ($\s \to_\pi \vp$) with $\pi =
\omega$, and the $\lambda$-bound variable is introduced in $\G$ (rules
$\lambda I_1$ and $\Lambda I$). Function application is standard, with
linear function application splitting the linear context in two
disjoint parts, one used to type the argument and the other the
function.
Arguments of unrestricted functions must also be unrestricted, i.e. no
linear variables can be used to type them. 
Typing linear and unrestricted function application separately is less
general than typing applications of functions of any multiplicity
$\pi$ by scaling (per the multiplicity semiring) the multiplicities of
the resources used to type the argument by $\pi$, however, our
objective of typing semantic linearity does not benefit much from
doing so and therefore we opt for a simpler design.

Multiplicity abstractions ($\Lambda I$) introduce a multiplicity variable $p$,
and construct expressions of type $\forall p.~\s$, i.e. a type universally
quantified over a multiplicity variable $p$. \ROUNDTWO{In the body
of the abstraction, function types and datatype fields annotated with the multiplicity
variable $p$ are typed parametrically in their multiplicity, since} $p$ can be instantiated at both $\omega$ and $1$.
A multiplicity application ($\Lambda E$) instantiates a multiplicity-polymorphic type
$\forall p.~\s$ at a particular (argument) multiplicity $\pi$, resulting in an
expression of type $\s$ where occurrences of $p$ are substituted by $\pi$, i.e.
$\s[\pi/p]$.
The rule additionally requires that $\pi$ be \emph{well-formed} in order
for the expression to be well-typed, using the judgement $\G \vdash_{mult}
\pi$, where well-formedness is given by $\pi$ either being $1$, $\omega$, or an
in-scope  multiplicity variable in $\G$.

\subsection{Usage environments\label{sec:usage-environments}}

A \emph{usage environment} $\Delta$ encodes the idea that lazy
\ROUNDTWO{variable} bindings do not consume resources upon definition, but rather when the bindings themselves are consumed.
Specifically, we annotate so-called $\Delta$-bound variables with \ROUNDTWO{their} \emph{usage
environment} \ROUNDTWO{$\Delta$} to denote that consuming
\ROUNDTWO{such variables must be equated with} consuming the
resources in the annotated \ROUNDTWO{typing} environment.
\ROUNDTWO{Such} variables are introduced in the unrestricted environment by
\ROUNDTWO{multiple} constructs such as (recursive) let binders, case binders, and case
pattern variables.
For example, the let-bound variable $u$ in the following program is annotated with a
usage environment $\{x{:}_1\s,y{:}_1\s\}$ tracking the resources that are used in its
body. Thus, in the expression $e$, consuming $u$ will equate to consuming the
linear resources $x$ and $y$:
\[
f = \lambda \x:_1\tau.~\lambda \y[1].~\llet{u_{\{x:_1\tau,y:_1\s\}} = (x,y)}{e}
\]
Furthermore, usage environments guarantee that using a $\Delta$-bound variable
is mutually exclusive with directly using the resources it is annotated with --
using the $\Delta$-bound variable consumes all linear resources listed in its
usage environment, meaning they are no longer available for direct usage.
Dually, using the linear resources directly means they are no longer available
to consume through the usage environment of the $\Delta$-bound variable. The
$\Delta$-bound variable can be left unused since it is unrestricted, regardless
of the variables it is annotated with. \ROUNDTWO{In the example above,
  our type system ensures that \emph{either} $x$ and $y$ are consumed
  or $u$ is consumed in $e$, since $u$ has usage environment $\{x:_1\tau,y:_1\s\}$.}

\paragraph{\texorpdfstring{$\D$}{Delta}-bound variables}

A $\Delta$-bound variable $u$ is a variable annotated with
\ROUNDTWO{some} usage environment \ROUNDTWO{$\Delta'$}.
\ROUNDTWO{For conciseness, we use the meta-variable $\Delta$ to refer to the usage
  environment of a $\Delta$-bound variable.}
Crucially, for any $\Delta$-bound variable $u$:
(1) using $u$ is equivalent to using all the linear resources in $\Delta$;
(2) using $u$ is mutually exclusive with using the $\Delta$ resources it depends on;
(3) $u$ can be safely discarded as long as the resources in $\Delta$ are consumed.
Fortunately, since linear resources must be linearly split across
sub-derivations, (2) follows from (1) since consuming the linear
resources in $\Delta$ to type $u$ makes them unavailable in the
context of any other derivation.
Similarly, (3) also follows from (1), because if the linear resources aren't
consumed \ROUNDTWO{by a use of $u$}, they must be consumed in 
some other derivation (or otherwise the \ROUNDTWO{overall} expression is ill-typed).
These observations lead to the typing rule for $\Delta$-bound
variables:

\[
\TypeVarDelta
\]
The rule states that an occurrence of a $\Delta$-bound variable is well-typed if
the linear environment is made up exactly of the resources in the usage environment of
that variable, \ROUNDTWO{denoted by the $\Delta = \Delta'$ premise of
  the rule.}
\ROUNDTWO{$\Delta$-bound variables} are unrestricted and can be discarded and duplicated, despite
multiple occurrences of the \ROUNDTWO{variable} in an expression not
being well-typed
\ROUNDTWO{due to} non-linear usage of linear resources.

\subsection{Lazy let bindings}

In Section~\ref{sec:semantic-linearity-examples}, we discussed how linear
resources used in let-bound expressions are only consumed when the respective
binders are evaluated in the let body.
Moreover, resources captured by a let-bound expression cannot be used
simultaneously with the let-binder in its body, since \ROUNDTWO{those}
resources would end up being consumed more than once and violate
\ROUNDTWO{the semantics of}
linearity. Either the binding or the
resources \emph{must} be used.

Let-bound variables are the canonical example of a
$\Delta$-bound variable.
Annotating let-bound variables with \ROUNDTWO{their} usage environment $\D$ delays the
consumption of resources to when the variables themselves are used and
guarantees the desired mutual exclusion property between the binder and the resources:
\[
\TypeLet
\]
The rule for (non-recursive) let bindings splits the linear environment in $\D$
and $\D'$. $\D$ is used to type the body $e$ of the let binding $x$. Perhaps
surprisingly, the resources $\D$ are still available to type the let body $e'$, alongside the unrestricted $x$ binding
annotated with the usage environment $\D$. The resources of $e$ being
available in $e'$ reflects how the body of a lazy binder
doesn't consume resources outright, and $x$ being an
\ROUNDTWO{unrestricted $\Delta$-bound variable}
ensures that using $x$ will consume the resources $\D$ that $e$
captures.

\subsection{Recursive let bindings}\label{sec:recursivelets}

Recursive let bindings also bind expressions lazily and so
introduce a $\D$-bound variable for each binding. \ROUNDTWO{The} resources required to type the
let-bindings are still available in the \ROUNDTWO{let body}. \ROUNDTWO{These
resources may} be consumed
via \ROUNDTWO{the corresponding} $\D$-bound variables, or directly if the let-bindings are unused.

The typing rule for recursive groups of bindings leverages both the assumption
that all mutually recursive groups are strongly connected and the corollary
that every binder in such a group must be typed with the same linear context,
as observed in Section~\ref{sec:semantic-linearity-examples:recursive-lets}.
Consequently, all bindings of the recursive group are introduced as
$\D$-bound variables
sharing the same usage environment -- using any one of the bindings in a
recursive group entails consuming all resources required to type that same
group -- so, at most a single binder from the group can be used in the
let body.
\[
\TypeLetRec
\]
This formulation is not syntax-directed 
since \ROUNDTWO{it does not describe how to annotate binders with} a
particular \ROUNDTWO{usage environment} $\D$.
\ROUNDTWO{Our system simply assumes the annotations are explicitly
  provided in the source and checks the correctness of the overall
  recursive let.} 

In practice, determining this typing environment $\D$ amounts to finding a
least upper bound of the resources needed to type each mutually-recursive
binding that (transitively) uses all binders in the recursive group.
Our implementation uses a naive $O(n^2)$ algorithm for inferring usage
environments of recursive bindings.
Moreover, we note that inference of usage environments for recursive binding
groups bears some resemblance to the inference of principle types for recursive
bindings traditionally achieved through the Hindley–Milner inference
algorithm~\cite{DBLP:conf/popl/DamasM82}. We leave the exploration of this connection
as future work.

\subsection{Case Expressions\label{sec:lc-case-exps}}

A case expression drives evaluation by evaluating its scrutinee to Weak Head
Normal Form (WHNF)~\cite{10.5555/1096899}. Then, the case
alternative whose pattern matches the result of evaluating the
scrutinee is taken\footnote{In our calculus, the alternatives are assumed to be
  exhaustive, i.e. there always exists at least one pattern which
  matches the scrutinee in its WHNF, so we're guaranteed to have an
  expression to progress evaluation.}.
An expression in WHNF can either be a $\lambda$-abstraction $(\lambda x.~e)$ 
or a datatype constructor application $(K~\ov{e})$.
In both cases, the sub-expressions $e$ or $\ov{e}$ occurring in the lambda body
or as constructor arguments need not be evaluated for the lambda or constructor
application to be in WHNF (and they may capture linear resources).

A key detail is that when the scrutinee is already in WHNF, evaluation continues
in the alternative simply by substituting the pattern variables by the
scrutinee sub-expressions. No resource from the scrutinee can be consumed since
no computation happens at all.
We introduce a typing judgement $\G;\D \Vdash e : \s \gtrdot \ov{\D_i}$ to
extract additional information from the static structure of terms in WHNF:
\[
    \TypeWHNFCons
\qquad
    \TypeWHNFLam
\]
This judgement differs from the main typing judgement in that (1) it only
applies to expressions in weak head normal form, and (2) it ``outputs'' (to the right of $\gtrdot$) a
disjoint set of linear environments ($\ov{\D_i}$), where each environment corresponds to the
linear resources used by a sub-expression of the WHNF expression.
To type a constructor application $K~\ov{e_\omega e_i}$, where $e_\omega$
are the unrestricted arguments and $e_i$ the linear arguments of the
constructor, we split the resources $\D$ into a disjoint set of resources
$\ov{\D_i}$ required to type each linear argument individually and return exactly
that split of the resources; the unrestricted $e_\omega$ expressions must be
typed on an empty linear environment. A lambda expression is typed with
the main typing judgement and trivially ``outputs'' the whole $\D$ environment,
as there is always only a single sub-expression in lambda abstractions.

\subsubsection{\ROUNDTWO{Discriminating by} WHNF-ness}\label{sec:onwhnf}

The dichotomy between evaluation (and resource usage) of a case expression
whose scrutinee is in weak head normal form and one whose scrutinee is not
leads to one of the key insights of $\lambda^\pi_\Delta$:
\ROUNDTWO{accurate typing of case expressions depends on whether the scrutinee is in WHNF.}
When the scrutinee is already in WHNF, the resources are
unused upon evaluation and thus available in the alternatives.
When it is not, resources will be consumed and cannot be used in the
alternative.
To illustrate, consider the following case expressions:
\[
\begin{array}{ccc}
    (1)~\lambda x.~\ccase{K~x}{z~\{\_ \to x\}} &  & (2)~\lambda x.~\ccase{free~x}{z~\{\_ \to x\}}
\end{array}
\]
The first function uses $x$ linearly, but the second does not.
Alternatives may also use the case binder or pattern variables, referring to, respectively,
the whole scrutinee and all its used resources; or
constructor arguments and the resources used to type them.
\ROUNDTWO{The reader may wonder about the significance of case expressions
  where the scrutinee is in WHNF. We note that in an
  optimizing compiler (such as GHC) many intermediate programs that
  arise due to program transformations commonly scrutinize expressions
  in WHNF.}

There are three competing ways
to use the resources from a scrutinee in WHNF in a case alternative: directly, via
the case binder, or by using pattern-bound variables.
Recall how $\D$-bound variables encode mutual exclusion between alternative
ways of consuming resources -- it follows that case binders and pattern-bound
variables are yet another instance of $\D$-bound variables.
This suggests the following first rule for cases of scrutinees already in WHNF:
\[
\TypeCaseWHNFIntermediate
\]

First, we assert this rule is only applicable to expressions in weak head
normal form. Second, we appeal to the typing judgement for expressions in WHNF
to determine the split of resources amongst the scrutinee
sub-expressions. Finally, we type all case alternatives with the same context, using the
introducing
the case binder $z$ in the unrestricted environment as a $\D$-bound
variable whose usage environment is the linear resources $\ov{\D_i}$ used to type the
scrutinee. Those same resources $\ov{\D_i}$ are again made available in the
linear typing environment of the alternative, similarly to how the resources
used to type a $Let$ binder are still available in the continuation.
Although the main idea for typing ``WHNF case expressions'' is conveyed by
this rule, the full rule of our system is slightly more involved, as will be seen after
discussing cases of scrutinees not in WHNF.

The alternative judgement \ROUNDTWO{$\G;\D \vdash \rho \Rightarrow e :^z_\D 
\s$} is used to type case alternatives, encompassing three
modes that are distinguished by the kind of arrow that is used:
for alternatives of case expressions whose scrutinee is in WHNF
($\Rightarrow_\mathsf{WHNF}$), not in WHNF ($\Rightarrow_\mathsf{NWHNF}$),
and for alternatives agnostic to the WHNF-ness of the scrutinee
($\Rightarrow$), with $\Rightarrow$ also generalizing the other two.

First, an alternative whose pattern is a constructor with $n > 0$ linear components
under the $\Rightarrow_{\mathsf{WHNF}}$ mode is typed as:
\[
\TypeAltNWHNF
\]
The rule states that, for such a match on a scrutinee already in
WHNF, we introduce the linear components of the pattern as $\D$-bound variables
with usage environment matching the linear resources required to type the
corresponding constructor argument in the scrutinee. The resources required to
type each constructor sub-expression come annotated in the judgement (as
${:}_{\ov{\D_i}}$). Unrestricted fields of the constructor are introduced as
unrestricted variables. We note that the typing environment $\D$ always
contains the resources $\ov{\D_i}$ in uses of the alternative judgement.

Second, the rule for alternatives that match on the wildcard pattern,
regardless of the WHNF status of the scrutinee (note the use of $\Rightarrow$, which
is applicable both under $\Rightarrow_{WHNF}$ and $\Rightarrow_{NWHNF}$):
\[
\TypeAltWild
\]
To type a wildcard alternative we simply type the expression with the main
judgement, ignoring all annotations. The case binder will
have already been introduced in the environment with the appropriate usage
environment by the relevant case expression rule.

Lastly, consider an alternative matching on a data constructor without any
linear components. The linear resources used to type a scrutinee matching such
a pattern are fully consumed during evaluation: the resulting unrestricted
constructor application can only be well-typed under an empty linear
environment.
Consequently, the resources to type the scrutinee which are carried over to the
case-alternative environments by the $Case_{WHNF}$ and $Case_{NWHNF}$ rules
(see the latter below), must be reactively destroyed from branches where the
pattern is unrestricted. In failing to do this, we would allow the fully-consumed
resources to be used again in the branch.
The $Alt0$ rule essentially encodes this insight, and is applicable regardless
of the WHNF status of the scrutinee ($\Rightarrow$ mode), as long as
the constructor pattern has no linear fields:
\[
\TypeAltZero
\]
The rule deletes the annotated scrutinee environment $\D_s$ from
the linear typing context (written $\D[\cdot/\D_s]$, a substitution of
the scrutinee typing environment by the empty linear environment $\cdot$)
and from the usage environment of the case binder $z$, written $\G[\cdot/\D_s]_z$
to delete $\D_s$ from the usage environment of the \ROUNDTWO{$\D$-bound variable} $z$ in $\G$.
Given that the case-binder $z$ is always annotated with $\D_s$ (essentially,
the linear environment typing the scrutinee), this latter substitution makes
the case-binder unrestricted in branches with unrestricted patterns.

This rule observes that any linear resources required to type an expression
that reduces to an unrestricted constructor application are fully consumed
through evaluation and the resulting unrestricted expression can be freely
discarded or duplicated (e.g. via the case binder).

\subsubsection{Irrelevant resources}\label{sec:irrev}

When the scrutinee is not in WHNF
\ROUNDTWO{we cannot statically determine which resources are consumed
  upon scrutinee evaluation or which are transferred into the branch continuation}
e.g. as in $\ccase{f\ x}{z\ \{K\
a \rightarrow e\}}$:
\begin{enumerate}

    \item Cases with non-WHNF scrutinees must be typed such that no resource
    from the scrutinee can be used directly in the branch again, or we
    risk
    consuming the same resource twice. In the example, $x$ must not directly occur
    in $e$ since $f$ \emph{could} have consumed it.

    \item On the other hand, since evaluation to WHNF does not guarantee all
    resources are consumed (unless the result matches an
    unrestricted pattern), we must ensure that any resources that transfer from the
    scrutinee to the alternative are ultimately consumed (in the example, $x$
    is not consumed by evaluating the scrutinee for $f = K$).

\end{enumerate}

\ROUNDTWO{
    Considering (1), we must resort to a more uniform method of guaranteeing we
    ``finish consuming'' the scrutinee. There are only two ways of
    uniformly referring to the remaining linear resources in the newly-evaluated scrutinee:

        \begin{itemize}
            \item The case binder, which refers to the whole result of evaluating the scrutinee;
            \item The linear components of a constructor pattern, which refer to expressions that may contain linear resources.
        \end{itemize}}

To encode  linear resources that cannot be
directly used (i.e.,~to which the $\textsc{Var}$ rule is not
applicable) we introduce \emph{irrelevant} resources, written as linear resources
within square brackets $[ x{:}\s]$ and \ROUNDTWO{pointwise} lifted to contexts
\ROUNDTWO{and usage environments} as $[\D]$.
Irrelevant resources
are linear resources in every sense, meaning they must be used
exactly once. However, since irrelevant resources cannot be
discarded or used directly, they have to be consumed \emph{indirectly} via
$\D$-bound variables, namely, the case binder, or, in mutual exclusion, the linear pattern-bound
variables.

Hence, to type a case expression whose scrutinee is not in WHNF, we
type the scrutinee with linear resources $\D$ and the case
alternatives by introducing the case binder with an \ROUNDTWO{irrelevant} usage environment $[\D]$,
having the same irrelevant $[\D]$ in the typing
environment, and annotating the judgement with the irrelevant resources
for use in the $\Rightarrow_{\textsf{NWHNF}}$ judgement:
\[
\TypeCaseNotWHNF
\]
Note how the rule is similar to the one for scrutinees in WHNF, but
requires the resources in the case binder, typing environment, and
judgement annotation to be made ``irrelevant''.

Finally, we recall the tentative $Case_\textrm{WHNF}$ rule presented before and
highlight its flaw: the $\G;\D \vdash \rho \Rightarrow_{\textsf{WHNF}}
e :^z_{\D_s} \s$ judgement is only well-defined for patterns $\rho$
matching the WHNF form 
of the scrutinee, as the distribution of resources per constructor components
only makes sense for the constructor pattern matching the scrutinee.
Alternatives not matching the scrutinee could use resources arbitrarily as they
will never be executed. We uniformly treat non-matching alternatives
as if the scrutinee were not in WHNF. Having introduced irrelevant
resources, we can now present the full $Case_\textrm{WHNF}$ rule:
\[
\TypeCaseWHNF
\]
\ROUNDTWO{where we appeal to the pattern judgment for WHNF scrutinees
  for pattern $\rho_j$, allowing the $\ov{\Delta_i}$ resources to be
  consumed either through the case-binder or outright; and to the
  non-WHNF judgment for the remaining branches, where the
  $\Delta$ resources may only be consumed via the case-binder.}
We again note that it might seem unusual to specialize a rule for expressions in
WHNF, as programs scrutinizing an expression in WHNF are rarely
written by humans.
Yet, our system is designed to be suitable for optimising compilers
in which intermediate programs commonly scrutinize expressions in WHNF.

\subsubsection{Splitting and tagging fragments}\label{sec:tagging}

As opposed to scrutinees in WHNF, where resources 
can be cleanly divided amongst the various sub-expressions of a constructor
application and, henceforth, each pattern variable, there is no direct mapping
between the resources typing a scrutinee not in WHNF and the usage
environments of pattern variables in any alternative.

We introduce \emph{tagged resources} as a means to enforce that all pattern-bound
variables for a scrutinee not in WHNF are either \emph{jointly} used to consume
all resources occurring in the environment, or not at all (instead, the case
binder may be used). Given \ROUNDTWO{irrelevant} resources $[\D_s]$ used to type a scrutinee,
and a pattern $K~\ov{x_\omega},\ov{y_i}$ with $i$ linear components, we assign
a usage environment $\D_i$ to each linear pattern variable where $\D_i$ is
obtained from the scrutinee environment \emph{tagged} with the constructor name
and linear-variable index $\lctag{\irr{\D_s}\!}{K_i}$. Then, $y_i{:}_{\D_i}\s$ is
introduced in $\G$, just like other $\D$-bound variables.
\[
  \begin{array}{cc}
    \TypeAltNNotWHNF & \TypeVarSplit
    \end{array}
\]
Having uniquely tagged the resources in the usage environment of each pattern
variable, we need only express that (1) the pattern $\D$-bound variables can
be used (i.e., the tagged resources need be available in the linear
environment in order for the $\textsc{Var}_\D$ rule to be applicable) and (2)
if a $K$-pattern $\D$-bound variable is consumed, all remaining
$K$-$\D$-bound variables,
standing for the remaining linear components of the same pattern, must also be
consumed.

To satisfy these two constraints, we introduce a $\textsc{Split}$ rule that
allows a linear resource $x{:}_1\s$ to be split into $n$ resources at a given
constructor $K$, where $n$ is the number of linear components of the
constructor, and each resource resulting from the split is \emph{tagged} with
$K$ and the positional index of a linear component (i.e. $x{:}_1\s$ can be
split into $\ov{\lctag{x{:}_1\s}{K_i}}^n$).
By assigning to each pattern variable a \emph{fragment} of the scrutinee
resources (with a tag), we require the scrutinee resources to be
$\textsc{Split}$ in order to use pattern variables at all.  Moreover, the
remaining tagged resources cannot be used directly, yet need to still be consumed
exactly once. Thus, the choice to consume the tagged resources via the usage
environments of the other pattern variables is forced, ensuring that no
variable for a constructor's linear component can go unused (the tagged environments are disjoint).
Tagged resources are inspired by fractional permissions in separation logic~\cite{10.5555/1760267.1760273}.

For instance, in the term $\lambda x.~\ccase{f~x}{z~\{K~a~b\to
  (a,b)\}}$, where $x$ is a linear variable, the case alternative is
typed with $\irr{x}$,
the case binder $z$ is introduced as $\z[\irr{x}]$, and the pattern variables
are introduced as $\var[a][\lctag{\irr{x}}{K_1}]$ and
$\var[b][\lctag{\irr{x}}{K_2}]$, assuming both components of $K$ are linear.
The occurrences of $a$ and $b$ are well-typed because we can first
$\textsc{Split}$ $[x{:}_1\s]$ into $\lctag{[x{:}_1\s]}{K_1}, \lctag{[x{:}_1\s]}{K_2}$,
noting that $\textsc{Split}$ can be applied both to relevant and proof irrelevant linear
resources in $\D$.

\section{Metatheory of Linear Core\label{sec:main:metatheory}}

We develop the metatheory of Linear Core by first 
presenting the operational semantics of $\lambda^\pi_\D$,
consisting of a lazy natural semantics in the
style of Launchbury~\cite{10.1145/158511.158618}.
This semantics conveys the expected behavior of an implementation of
lazy evaluation and is agnostic to any linearity information.
Such a semantics is ill-suited for reasoning about
linearity and so we develop an \emph{instrumented} linearity-aware semantics that is
enriched with sufficient information to allow us to establish type
safety, thus showing that linearity is preserved. We show that the natural
and instrumented semantics are bisimilar for well-typed terms and so
derive type safety in the natural semantics
(Section~\ref{sec:typesafe}). Finally, we show that multiple optimising
transformations such as those performed by GHC on Linear Core are
type preserving
(Section~\ref{sec:optimisations-preserve-types-meta}), and thus
linearity preserving.

\subsection{Operational Semantics}

The Launchbury-style natural semantics, presented in Figure~\ref{fig:opsem}, captures
standard lazy evaluation while ignoring linearity information.
The natural semantics is equipped with an evaluation environment that
maps variables to expressions and is
mutated in order to express shared evaluation. Terms are explicit about
their sharing (in the sense of shared reductions in lazy evaluation),
and so we perform a translation step that makes all sharing explicit
through let-binders (see~Appendix).
We write $\llet{\ov{x_i{:}_{\Delta_i}\sigma = e_i}}{e'}$ for the iterated
let-binding of variables $\ov{x_i}$. 

\begin{figure}

  {\small
    \[
  \begin{array}{c}
    \infer[]
    {\,}
    {\Theta : \Lambda p.e \Downarrow \Theta : \Lambda p.e}
    \qquad
    \infer[]
    {\Theta : e\Downarrow \Theta' : \Lambda p.e' \quad \Theta' :
    e'[\pi/p]\Downarrow \Theta'' : v}
    {\Theta : e~\pi\Downarrow \Theta'' : v }\\[1.5em]
    \infer[]
    {\,}
    {\Theta : \lambda x:_\pi \sigma . e \Downarrow \Theta : \lambda x:_\pi \sigma . e }
    \qquad
    \infer[]
    {\Theta : e \Downarrow \Theta' : \lambda y:_\pi \sigma . e' \quad
    \Theta' : e'[x/y] \Downarrow \Theta'' : v }
    {\Theta : e~x \Downarrow \Theta'' : v}\\[1.5em]
    \infer[]
    {(\Theta,  x:_\omega \sigma = e) : e \Downarrow \Theta' : v}
    {(\Theta , x:_\omega \sigma = e) : x \Downarrow (\Theta',x:_\omega
    \sigma = v) : v}
    \qquad
    \infer[]
    {(\Theta,x:_\omega \sigma = e) : e' \Downarrow \Theta' : v}
    {\Theta : \llet{x{:}_{\Delta}\sigma = e}{e'}\Downarrow \Theta' : v
    }\\[1.5em]
    \infer[]
    {\,}
    {\Theta : K~\ov{x_i} \Downarrow \Theta : K~\ov{x_i}}
    \qquad
    \infer[]
    {\Theta : e \Downarrow \Theta' : K~\ov{x_i} \quad
    \Theta' : e'\ov{[x_i/y_i]}[K~\ov{x_i}/z] \Downarrow \Theta'' : v }
    {\Theta :
    \ccase{e}{z{:}_{\Delta'}\sigma~\{{\dots,K~\ov{y_i} \Rightarrow e', {\dots}}\}}
    \Downarrow \Theta'' : v}\\[1.5em]
    \infer[]
    {\Theta : e \Downarrow \Theta' : K~\ov{x_i} \quad
    \Theta' : e'[K~\ov{x_i}/z] \Downarrow \Theta'' : v }
    {\Theta :
    \ccase{e}{z{:}_{\Delta'}\sigma~\{\dots , \_\Rightarrow
    e'}\} \Downarrow \Theta'' : v}
    \qquad
    \infer[]
    {(\Theta, \overline{x_i:_\omega \sigma_i = e_i}) : e' \Downarrow \Theta'
    : v
    }
    {\Theta : \lletrec{\overline{x_i{:}_{\Delta}\sigma_i = e_i}}{e'}
    \Downarrow \Theta' : v}
    \end{array}
\]}
    \caption{Natural Semantics of $\lambda^\pi_\D$\label{fig:opsem}}
  \end{figure}

  The natural semantics is defined via the relation
$\Theta : e \Downarrow \Theta' : v$ where $e$ is an expression, $v$ a
value, and $\Theta$ and $\Theta'$ are evaluation environments with bindings of
the form $x:_\omega\sigma = e$ which assigns the expression $e$ to the variable
$x$ of the given type $\s$. The rules are standard, augmenting the runtime
environment with let-bound expressions which are evaluated (only once) if the
corresponding let-bound variable is forced.

While the natural semantics is a mostly direct model of an implementation
of lazy evaluation, it is inconvenient for formal reasoning about
well-typed linear terms.
We follow an approach
similar to that of \cite{cite:linearhaskell,10.1145/3158093}, defining
an \emph{instrumented} operational semantics in which linear variables are
erased from 
the runtime environment once they are forced, ensuring that any term violating linearity
will result in a stuck state. The semantics is also type-aware and
carries with it sufficient data to reconstruct typing derivations for
the purposes of showing type safety.

The instrumented semantics is presented in
Figure~\ref{fig:linopsem}, defining the judgment $\Gamma ; \Delta \vdash (\Theta \mid e) \Downarrow (\Theta'
\mid v) : \sigma , \Sigma$ where $\Gamma$ and $\Delta$ are typing
contexts for expressions in $\Sigma$
($\Gamma$ tracks unrestricted and $\Delta$-bound variables, $\Delta$ tracks linear
assumptions);
$\Theta$ and $\Theta'$ are evaluation environments, consisting of bindings of the form
$x :_\pi \sigma = e$ or $x :_\Delta \sigma = e$; $e$ is the expression being evaluated and $v$
its resulting value, both of type $\sigma$; $\Sigma$ is a list of assignments of the form
$e : \tau$ which are expressions in which
bindings in $\Theta$ may also be used. This additional data is needed to
inductively show that the overall evaluation state is well-typed (and
linearity preserving). The rules mostly mirror those of the natural
semantics, with the exception of the linear variable rule which erases
the binding from the execution environment. Both semantics define
evaluation relations of the form $a \Downarrow b$, where $a$ and $b$
are evaluation states. We also rely on a notion
of partial derivations and partial evaluation (see~Appendix), written $a\Downarrow^* b$, which allows us
to state and reason about a progress property.

\begin{figure}
  {\small
\[
  \begin{array}{c}
    \infer[]
    {\,}
    {\Gamma;\Delta \vdash (\Theta \mid \Lambda p.e) \Downarrow (\Theta \mid
    \Lambda p.e) : \forall p . \sigma , \Sigma }
    \\[1.5em]
    \infer[]
    {\Gamma ; \Delta \vdash (\Theta \mid  e) \Downarrow (\Theta' \mid
    \Lambda p.e') : \forall p . \sigma , \Sigma \quad \Gamma ;\Delta \vdash
      (\Theta' \mid e'[\pi/p]) \Downarrow (\Theta'' \mid v) : \sigma[\pi/p], \Sigma}
    {\Gamma ; \Delta \vdash (\Theta \mid e~\pi) \Downarrow (\Theta''
      \mid v) : \sigma[\pi/p] , \Sigma  }\\[1.5em]
    \infer[]
    {\,}
    {\Gamma ;\Delta \vdash (\Theta \mid \lambda x:_\pi \sigma . e)
    \Downarrow (\Theta \mid \lambda x:_\pi\sigma .e) : \sigma
    \rightarrow_\pi \tau,\Sigma }\\[1.5em]
    \infer[]
    {\Gamma;\Delta \vdash  (\Theta \mid e) \Downarrow (\Theta' \mid
    \lambda y:_1 \sigma . e') : \sigma \rightarrow_1 \tau , x:\sigma, \Sigma
    \quad
    \Gamma ; \Delta \vdash (\Theta' , y:_{1} \sigma = x \mid e') \Downarrow (\Theta''
    \mid v) : \tau , \Sigma  }
    {\Gamma ; \Delta\vdash (\Theta \mid  e~x ) \Downarrow (\Theta''
    \mid v) : \tau , \Sigma }\\[1.5em]
           \infer[]
    {\Gamma;\Delta \vdash  (\Theta \mid e) \Downarrow (\Theta' \mid
    \lambda y:_\omega \sigma . e') : \sigma \rightarrow_\omega \tau , x:\sigma, \Sigma
    \quad
            \Gamma ; \Delta \vdash (\Theta' \mid e'[x/y]) \Downarrow (\Theta''
    \mid v) : \tau , \Sigma  }
    {\Gamma ; \Delta\vdash (\Theta \mid  e~x ) \Downarrow (\Theta''
    \mid v) : \tau , \Sigma }\\[1.5em]
     
       \infer[]
    {\Gamma ; \Delta  \vdash (\Theta  , x:_{\Delta'} \sigma = e\mid e)
    \Downarrow (\Theta' \mid v) : \sigma , \Sigma \and \D'' = \D' \setminus (\Theta \setminus \Theta')}
    {\Gamma ; \Delta \vdash (\Theta , x:_{\Delta'} \sigma = e \mid x)
      \Downarrow (\Theta' , x:_{\Delta''} \sigma = v \mid v) : \sigma,
    \Sigma }
    \\[1.5em]
           \infer[]
    {\Gamma ;\Delta  \vdash (\Theta \mid e)
    \Downarrow (\Theta' \mid v) : \sigma,\Sigma}
    {\Gamma ; \Delta \vdash (\Theta , x:_1 \sigma = e \mid x)
    \Downarrow (\Theta' \mid v) : \sigma,\Sigma}
    \\[1.5em]
    
    \infer[]
    {\Gamma ;\Delta \vdash (\Theta,x:_{\Delta'} \sigma = e \mid e')
    \Downarrow (\Theta' \mid v) : \tau,\Sigma}
    {\Gamma ; \Delta \vdash (\Theta \mid \llet{x{:}_{\Delta'}\sigma =
    e}{e'})\Downarrow (\Theta' \mid v) : \tau , \Sigma
    }\\[1.5em]
    
    \infer[]
    {\,}
    {\Gamma ; \Delta \vdash (\Theta \mid K~\ov{x_i}) \Downarrow
    (\Theta \mid K~\ov{x_i}) : T , \Sigma}
    \\[1.5em]
    
    \infer[]
    { K : \ov{ \sigma_i
    \rightarrow_\omega \sigma_j \lolli} T
    \\
    \Gamma, z{:}_{\ov{y_j}}T, \ov{y_i :_\omega \sigma_i} ;\Delta , \ov{y_j :_1 \sigma_j} \vdash (\Theta \mid e) \Downarrow (\Theta' \mid
      K~\ov{x_i}) : T , e' : \tau ,\Sigma \\
    \Gamma ;\Delta \vdash (\Theta' \mid e'\ov{[x_i/y_{i,j}]}[K~\ov{x_i}/z]
    ) \Downarrow (\Theta'' \mid v) :\tau , \Sigma }
    {\Gamma ;\Delta \vdash (\Theta \mid
    \ccase{e}{z{:}_{\Delta'}T~\{{\dots,K~\ov{y_{i,j}} \Rightarrow e', {\dots} }\}})
    \Downarrow (\Theta'' \mid  v) : \tau , \Sigma}\\[1.5em]

    \infer[]
    {K : \ov{ \sigma_i
    \rightarrow_\omega \sigma_j \lolli} T\\
    \Gamma ; \Delta, z{:}_1\s \vdash (\Theta \mid e) \Downarrow (\Theta' \mid
    K~\ov{x_i}) : T , e', \Sigma \quad
    \Gamma ;\Delta \vdash (\Theta' \mid  e'[K~\ov{x_i}/z])
    \Downarrow (\Theta'' \mid v) : \tau,\Sigma }
    {\Gamma ; \Delta \vdash (\Theta \mid 
    \ccase{e}{z{:}_{\Delta'}T~\{\dots , \_\Rightarrow
    e'}\}) \Downarrow (\Theta'' \mid v) : \tau}
    \\[1.5em]
    \infer[]
    {\Gamma ; \Delta \vdash (\Theta, \overline{x_i:_{\Delta'} \sigma_i =
    e_i} \mid e' )\Downarrow (\Theta'
    \mid v) : \tau
    }
    {\Gamma ;\Delta \vdash (\Theta \mid \lletrec{\overline{x_i{:}_{\Delta'}\sigma_i = e_i}}{e'})
    \Downarrow (\Theta' \mid v) : \tau}

   \end{array}
 \]}
 \caption{Instrumented Operational Semantics~\label{fig:linopsem}}
  
\end{figure}

\subsection{Type Safety}\label{sec:typesafe}

Our results rely on a series of lemmas relating linear, unrestricted and
$\Delta$-bound variables, as well as a lemma characterizing irrelevant
resources.

\renewcommand{\DeltaLinearRelationLemma}{
\begin{restatable}[$\Delta$-bound to Linear]{lemma}{deltaone}\label{lem:deltaone}
If $\G,\x[\irr{\D}]; \irr{\D},\D' \vdash e : \s$ 
then $\G[x/\irr{\D}]; \D',\xl \vdash e :\s$.
\end{restatable}
}

\renewcommand{\LinearDeltaRelationLemma}{
  \begin{restatable}[Linear to $\Delta$-bound]{lemma}{onedelta}\label{lem:onedelta}
If $\G; \D',\xl \vdash e :\s$
then $\G[\D/x],\xD; \D,\D' \vdash e : \s$ ($\Delta$ fresh).
\end{restatable}
}

\renewcommand{\DeltaUnrestrictedRelationLemma}{
\begin{restatable}[Unrestricted and $\Delta$-bound]{lemma}{undelta}\label{lem:undelta}
$\G,\xo; \D \vdash e : \s$ iff $\G,\x[\ROUNDTWO{\emptyset}]; \D \vdash e : \s$
\end{restatable}
}

\paragraph{$\Delta$-bound variables} A well-typed program with a linear variable ($\xl$) is equivalently
well-typed if the linear variable were instead $\D$-bound ($\xD$) with
usage environment $\D$, with $\D$ available in the linear context and
\ROUNDTWO{with any occurrences of $x$ in usage environments in $\G$
  expanded into $x$'s usage environment $\Delta$, with the abuse of notation $\Gamma[\Delta/x]$}.

\LinearDeltaRelationLemma

\noindent Dually, any $\Delta$-bound variable with an irrelevant usage
environment can be made into a linear variable as long as its
(irrelevant) environment is dropped from the linear context
\ROUNDTWO{and all the occurrences of the set of variables $[\D]$ in
  usage environments in $\Gamma$ are collapsed to $x$, written $\Gamma[x/[\Delta]]$}.

\DeltaLinearRelationLemma

\noindent Finally, unrestricted and $\D$-bound
variables with an empty usage environment are equivalent.

\DeltaUnrestrictedRelationLemma

\paragraph{Irrelevance}
Irrelevant resources \ROUNDTWO{(Section~\ref{sec:irrev})} can only be consumed
indirectly, essentially encoding that the non-WHNF case scrutinee
resources must be consumed through the case binder or the linear
pattern-bound variables.
As a case expression is evaluated, the scrutinee will eventually
reduce to WHNF,
which must then be typed with rule $Case_{\textrm{WHNF}}$.
Crucially, the two case rules must be in harmony in the system, in the sense that
case expressions typed using the $Case_{\textrm{Not WHNF}}$ rule must also be
well-typed by the $Case_{\textrm{WHNF}}$ once the scrutinee is evaluated to WHNF.

\WHNFConvSoundness

The \emph{Irrelevance} lemma \ROUNDTWO{shows that if a case
  alternative is well-typed with irrelevant resources, then it is
  well-typed with arbitrary resources.}
However, typing
a case alternative with irrelevant resources is not complete wrt using
arbitrary resources -- a counter example needs only to use a resource
directly.

We prove type safety of the Linear Core system via the standard type
preservation and progress results. As is customary, we make use of multiple
substitution lemmas, one for each kind of variable: unrestricted variables
$\xo$, linear variables $\xl$, and $\D$-bound variables $\xD$. The
development is reported in the Appendix, where type safety is established directly for the
instrumented semantics and ported to the natural semantics via a
bisimulation argument.

\ROUNDTWO{We adopt the standard semantic technique of dealing with badly-formed
terms, such as trying to add two values that are not numbers, by
having such terms be stuck. In our setting, linear variables used more
than once result in stuck states. Thus,} progress ensures that linear
functions consume their arguments at most once if their result is
consumed exactly once.
Type preservation ensures that evaluation of a
closed term returns an evaluation environment with no linear variables.

  \begin{theorem}[Type Preservation]
    For any well-typed evaluation state $a$, if $a \Downarrow
    b$ or $a \Downarrow^* b$ then $b$ is well-typed.
  \end{theorem}

\begin{theorem}[Progress]
Let $a$ be a well-typed evaluation state. For any partial evaluation $a \Downarrow^* b$ 
the evaluation can be extended.
\end{theorem}

\subsection{Optimisations preserve linearity\label{sec:optimisations-preserve-types-meta}}
A main goal of the Linear Core type system is 
to serve as a typed intermediate representation for optimising
compilers of lazy languages with linearity. In light of this,
we show that multiple optimising transformations are type
preserving in Linear Core. 

\begin{figure}[t]

  {\small
  \[
    \begin{array}{ll}
   (\llet{\xD = e}{e'} )  \Longrightarrow \llet{\xD =       e}{e'[e/x]} &                    \mbox{Inlining}\\[1ex]
       (\lambda \x[\pi][\s].~e)~e'  \Longrightarrow e[e'/x]  &                \mbox{$\beta$-reduction}\\[1ex]
      (\ccase{(K~\ov{e})}{\z[\D][\s]~\{\overline{\rho_j\to e_j}, K~\ov{x} \to e_i, \overline{\rho_k\to e_k}\}} )\Longrightarrow  e_i\ov{[e/x]}[K~\ov{e}/z] &
                                  \mbox{Case of known constr.}\\[1.5ex]
      (\ccase{(\ccase{e_c}{\z[\D]~\{\ov{\rho_{c_i}\to e_{c_i}}\}})}{\var[w][\irr{\D,\D'}][\s']~\{\ov{\rho_i\to e_i}\}})\\
            \Longrightarrow 
\ccase{e_c}{\z[\D]~\{\ov{\rho_{c_i} \to \ccase{e_{c_i}}{w~\{\ov{\rho_i\to e_i}\}}}\}}
      & \mbox{Case-of-case}\\[1.5ex]
      (\lambda \y[\pi].~\llet{\xD = e}{e'})
      \Longrightarrow
      \llet{\xD=e}{\lambda\y[\pi].~e'} & \mbox{Commuting let-$\lambda$}\\[1ex]
      ((\llet{v = e}{b})~a) \Longrightarrow \llet{v = e}{b~a} &\mbox{Commuting let-app.}\\[1ex]
      (\ccase{(\llet{v = e}{b})}{\{\overline{\rho\to e'}\}}) \Longrightarrow
\llet{v =  e}{\ccase{b}{\{\overline{\rho\to e'}\}}}
                               &
                                 \mbox{Commuting let-case}\\[1ex]
      (\ccase{e}{\{\overline{\rho_j\to e_j}, \rho \to E[e_1],
      \overline{\rho_k\to e_k}\}} ) \qquad (\mbox{$x$ fresh}) \\[1.5ex]
\Longrightarrow \llet{x = e_1}{\ccase{e}{\{\overline{\rho_j\to e_j}, \rho \to E[x], \overline{\rho_k\to e_k}\}}} 
                                                                        & \mbox{Commuting case-let}\\[1.5ex]
      (\llet{x = (\llet{v = e}{b})}{c}) \Longrightarrow \llet{v = e}{\llet{x = b}{c}}
                                                                      & \mbox{Commuting let-let}\\[1ex]
      f \Longrightarrow \lambda x.(f~x) & \mbox{$\eta$-expansion}\\[1ex]
      \lambda x.(f~x) \Longrightarrow f & \mbox{$\eta$-reduction}\\[1ex]
      (\ccase{x}{z~\{\ov{\rho_i\to e_i}\}}) \Longrightarrow \ccase{x}{z~\{\ov{\rho_i\to  e_i[z/x]}\}}
                                                                      & \mbox{Binder swap}\\[1ex]
\end{array}
  \]}
  \caption{Optimising transformations validated by Linear
    Core\label{fig:opttrans} (for commuting case-let we
      require $e_1$ not to capture the case binder, variables bound by the
      pattern $\rho$ and the context $E[{-}]$ -- see Appendix for
    details).} 
\end{figure}

We describe Core-to-Core transformations as $e_1 \Longrightarrow e_2$,
where $e_1$ is an arbitrary \emph{well-typed} expression
matching a certain shape (or of a certain type), which is
transformed into expression $e_2$. Validating a transformation
entails showing that $e_2$ is well-typed. We often annotate the arrow
with the name of the transformation when describing chains of such
transformations.

The list of transformations we have
validated is given in Figure~\ref{fig:opttrans}.
The behaviour of most transformations can be inferred from its name:
Inlining substitutes a let-bound expression in the let body; $\beta$-reduction
and case of known constructor effectively perform evaluation of the given expression;
case-of-case commutes cases in scrutinee position inside the outer case alternatives;
the five commuting conversions allow let-bound expressions to be commuted with other
Core constructs; $\eta$-expansion and reduction apply the $\eta$-laws to terms of function type;
finally, binder swap substitutes occurrences of a scrutinee variable for the
case binder. While many of the transformations above are not obviously optimising,
they are chained in the compiler pipeline to expose optimisation
opportunities.

\paragraph{Reverse Binder Swap\label{sec:reverse-binder-swap-considered-harmful}}

Reverse binder swap is dual to the binder swap transformation, substituting
occurrences of the case binder $z$ by the scrutinee in the special case when
the scrutinee is a variable $x$.
By using the scrutinee $x$ instead of the case binder, we might be able to
float out expressions from the alternative using the case binder:
\[
    \ccase{x}{z~\{\rho \rightarrow e\}} \Rightarrow \ccase{x}{z~\{\rho \rightarrow e[x/z]\}}
\]

Even though GHC applies the reverse binder swap transformation in the
Core-to-Core passes, this optimisation violates linearity in Linear Core
because the transformed program is rejected -- the scrutinee $x$, a single
variable, is not considered to be in WHNF by our system, thus occurrences of
$x$ in any alternative are ill-typed in it. This is a conservative choice in $\lambda^\pi_\Delta$ for
typing the subtle case where the scrutinee is a single variable.

Consider $\ccase{x}{\_ \to x}$. Perhaps surprisingly, this expression can be
considered semantically linear under call-by-need but violates
linearity under call-by-name:
Under call-by-need, if $x$ refers to an unevaluated expression, scrutinizing it forces the expression
bound by $x$ to WHNF. In a subsequent use of
$x$ in an alternative, $x$ refers to the already evaluated scrutinee in WHNF.
Since $x$ is just another name (like the case binder) for the scrutinee in
WHNF, its use is also valid in the alternatives.
Dually, if $x$ refers to an expression already in WHNF, no evaluation takes
place and $x$ is still just another name for the scrutinee in WHNF.

However, under call-by-name, evaluation of $x$ is not shared:
\[
\begin{array}{l}
(\lambda x.~\ccase{x}{\_ \to x})~(use~y)
\Longrightarrow_\textrm{CBName}
\ccase{use~y}{\_ \to use~y}
\end{array}
\]
After the reduction, $y$ is a linear variable consumed both in the scrutinee
and in the case alternative -- $y$ has been duplicated!
If the application is instead reduced call-by-need, $y$ is not duplicated:
\[
\begin{array}{l}
(\lambda x.~\ccase{x}{\_ \to x})~(use~y)
\Longrightarrow_\textrm{CBNeed}
\llet{x = use~y}{\ccase{x}{\_ \to x}}
\end{array}
\]

In an optimising compiler such as GHC, there is no definitive evaluation
strategy. It is mostly up to heuristics to determine what to inline
\emph{à la} call-by-name and what to create let-bindings for \emph{à
la} call-by-need.
In this flexible evaluation setting, we highlight the reverse binder swap is
only a linearity-preserving transformation as long as a linear variable with
more than one occurrence is never substituted for an expression with free
linear variables.

\section{Linear Core as a GHC Plugin\label{sec:discuss:implementation}}

We implemented Linear Core as a plugin for the Glasgow Haskell Compiler.
GHC plugins allow developers to inspect and modify programs being compiled by
GHC, at different stages of compilation~\cite{10.1145/3331545.3342599}.
The Linear Core GHC plugin~\cite{cite:linear-core-plugin} is a typechecker. For
every single intermediate Core program produced by GHC, our plugin checks the
linearity of all expressions, failing if a linear resource is not used
\emph{exactly once} according to \ROUNDTWO{our typing discipline}.
The plugin implementation of Linear Core provides us with strong confidence
that the Linear Core type system is suitable for the intermediate language of
an optimizing compiler, since it is successful in accepting the vast majority
of real-world linear programs produced by the GHC optimizer, and even
identified some which violate linearity.

The implementation of a Linear Core checker does not follow directly from the
type system because Linear Core is not syntax-directed \ROUNDTWO{(the design
aims to be as simple to reason with as possible).}
Specifically, the \ROUNDTWO{features of Linear Core that are not
syntax-directed are essentially the}
splitting of linear resources amongst sub-derivations and
consuming fragments of resources through pattern-bound variables.
\ROUNDTWO{As standard in substructural typing systems, our rules
  non-deterministically split linear resources as needed.}
\ROUNDTWO{In our implementation,} we thread input/output linear resources through each
derivation~\cite{CERVESATO2000133}.
To use pattern variables with usage environment $\lctag{\irr{\D}}{K_i}$, the
scrutinee resources must first be $\textsc{Split}$ into fragments. Instead of
guessing which resources need to be split,
the implementation consumes tagged fragments of a resource as needed,
i.e., only when consuming a resource tagged $K_j$ do we $\textsc{Split}$ that
resource on $K$ and consume the $K_j$ fragment.
Finally, our implementation must infer the usage environments of binders in
a recursive let group before using them to typecheck the let body. We use a
naive $O(n^2)$ algorithm (where $n$ is the number of let bindings) to determine
these usage environments. We leave the study of the soundness and completeness
of this implementation strategy to future work.

The results of running the Linear Core GHC plugin on real-world libraries with
linear types are reported in Figure~\ref{fig:core-plugin-res}.
We compiled the Haskell standard library for programming with linear types,
\texttt{linear-base} (4000 lines, comprising over 100 modules),
\texttt{linear-smc}~\cite{10.1145/3471874.3472980} (1500 lines),
\texttt{priority-sesh}~\cite{10.1145/3471874.3472979} (1400 lines),
\texttt{text-builder-linear} \cite{cite:text-builder-linear} (1800 lines), and
\texttt{linear-generics}~\cite{cite:linear-generics} (2000 lines), using our
plugin.
We count the number of programs accepted by our implementation, where
each top-level binding in a module counts as a program, and every such
binding is checked once per optimisation pass, i.e.~we check
all intermediate programs produced by GHC. The total of rejected programs
is counted in the same way, without halting compilation. We note that our implementation is not performance
conscious and types every intermediate program from scratch.

{\small
  \begin{figure}[h]
\centering{
\begin{tabular}{|l|c|c|c|c|}
\thead{Library} & \thead{Total\\Accepted} & \thead{Total\\Rejected} & \thead{Unique\\Rejected} & \thead{Accept Rate}\\
linear-base & 48018 & 735 & 51 & 98.50\% \\
    linear-smc & 18759 & 5 & 3 & 99.97\% \\
    priority-sesh & 760 & 95 & 6 & 88.89\% \\
    text-builder-linear & 4858 & 34 & 6 & 99.30\% \\
    linear-generics & 63332 & 0 & 0 & 100.00\% \\
\end{tabular}
}
\caption{Linear Core Plugin on Linear Libraries with GHC 9.10.3}
\label{fig:core-plugin-res}
\end{figure}
   }
The results indicate that our mostly direct implementation of Linear Core is
successful in accepting the vast majority of the thousands of intermediate
programs produced by GHC when compiling libraries that make extensive use of
linear types.
However, Linear Core does not accept every single program produced by the
optimisation pipeline of GHC. Our manual analysis of rejected programs revealed
excellent results. Programs considered ill-typed by the Linear Core
implementation include programs that:
\begin{itemize}
\item we expected GHC to produce, but know are not seen as linear by $\lambda^\pi_\Delta$.
        Namely, programs resulting from the reverse binder swap which
        scrutinize a variable and then use it again in the case alternatives.
        As discussed in Section~\ref{sec:reverse-binder-swap-considered-harmful}, these
        programs could be understood as linear as long as all applications of
        linear functions are reduced call-by-need
\item were rejected due to the lack of ``linearity coercions'' in
        $\lambda^\pi_\Delta$. Extending Linear Core with these coercions would
        make these programs be accepted (as discussed in Section~\ref{sec:future-work});
\item actually violate linearity. These are the programs
        we ultimately want to spot and reject by having a linearly typed Core
        validate intermediate programs. We found both a program which outright
        discarded a linear resource and that rewrite rules implementing stream
        fusion~\cite{10.1145/165180.165214,10.1145/1291151.1291199}
        (among other rules) could produce linearly-invalid programs (in one such case,
        the type of \lstinline{build} was not general enough to accommodate linearity).
\end{itemize}
Besides validating that our implementation is faithful to the
$\lambda^\pi_\Delta$ system, these programs provide insight as to what more is
needed to fully understand linearity in a mature optimising compiler.

\section{Related and Future Work\label{sec:discussion}}

\paragraph{Linear Haskell\label{sec:related-work-linear-haskell}.}
Linear Haskell~\cite{10.1145/3158093} augments the Haskell surface language
with linear types, but it is not concerned with extending GHC's intermediate
language(s), which presents its own challenges. 
However, the implementation of Linear Haskell in GHC does modify and extend Core
with linearity/multiplicity annotations. Core's type system is unable to type
semantic linearity of programs in the sense elaborated in our work,
which results in Core rejecting most linear programs resulting from optimising
transformations that leverage the non-strict semantics of Core.
Our work overcomes the limitations of Core's linear type system derived from
Linear Haskell by understanding linearity semantically in the presence of
laziness, and by showing that multiple Core-to-Core
optimisations employed by GHC are linearity preserving, using a
linearity-aware semantics that is essentially identical to that of~\cite{cite:linearhaskell}.
Linear Core can also be seen as a system that validates the programs written in
Linear Haskell and are compiled by GHC, by guaranteeing (through typing) that
linear resources are still used exactly once throughout the optimising
transformations.
\ROUNDTWO{We note that both Linear Haskell and our work has no
  special treatment of exceptions. This is also the case in GHC Core,
  where exceptions have no special status and so there is
  no natural way of dealing with the interaction between exceptions and
  linearity in Core itself without a major overhaul of the exception
  mechanisms of the language. }

\paragraph{Linearity-aware Semantics}
Several other work explore linearity aware semantics,
dating back to \cite{Chirimar_Gunter_Riecke_1996} and
\cite{TURNER1999231}.
The former uses a computational interpretation of linear logic to give
an account of reference counting. The latter present a heap-based
operational interpretation of linear logic and explore memory usage
properties under call-by-name and call-by-need, showing both satisfy
different properties under this lens. Their semantics is similar to
ours, but does not account for pattern matching.
More recently, \cite{DBLP:journals/pacmpl/ChoudhuryEEW21,DBLP:conf/esop/MarshallVO22} and
\cite{cite:linearhaskell,10.1145/3158093} all make use of usage-aware
semantics. \citet{DBLP:journals/pacmpl/ChoudhuryEEW21} do so in the context of graded dependent
type theory, using the usage-aware heap-based operational semantics to
show correct accounting of resource usage in their setting, as we do
in ours, but they do not model
call-by-need, only call-by-name. The work of \cite{DBLP:conf/esop/MarshallVO22} is
applies the ideas of \cite{DBLP:journals/pacmpl/ChoudhuryEEW21}
to a non-graded call-by-name setting. \citet{10.1145/3158093} provide a similar semantics
to ours, but present only a naive linear typing system that does not
leverage non-strict evaluation.
\ROUNDTWO{We also note the non-strict language
  Clean~\cite{10.1007/3-540-18317-5_20} that employs uniqueness types
  to ensure resources are used at most once.}

\paragraph{Inspirations for Linear Core\label{sec:linear-mini-core}}

Linear Mini-Core~\cite{cite:minicore} is a specification of linear types in
Core as they were being implemented in GHC, and doubles as the (unpublished)
precursor to our work. Linear Mini-Core first observes the incapacity of
Core's type system to accept linear programs after transformations, and 
introduces usage environments for let-bound variables with the same goal of
Linear Core of specifying a linear type system for Core that accepts the
optimising transformations.
We draw from Linear Mini-Core the rule for non-recursive let expressions and
how let-bound variables are annotated with a usage environment. However, our
work further explores the interaction of laziness with linearity in depth, and
diverges significantly in rules for typing other constructs, notably, case expressions and
case alternatives. Furthermore, unlike Mini-Core, we prove Linear Core is type
safe, guarantees linear resource usage, and that multiple optimising
transformations, when applied to Linear Core programs, preserve linearity.
\ROUNDTWO{Our irrelevant resources
  have similarities with 0-multiplicity
  values from QTT~\cite{brady:LIPIcs.ECOOP.2021.9} which are
  disallowed from being used directly in executable code and, as previously noted,
  our tagged resources are related to fractional permissions from separation
  logic~\cite{10.5555/1760267.1760273}.}

\paragraph{Linearity-directed optimisations}
Core-to-Core transformations appear in multiple papers~\cite{cite:let-floating,peytonjones1997a,santos1995compilation,peytonjones2002secrets,baker-finch2004constructed,maurer2017compiling,Breitner2016_1000054251,sergey_vytiniotis_jones_breitner_2017},
all designed in the context of a typed language (Core) which does not have
linear types. However, some authors~\cite{cite:let-floating,peytonjones1997a,10.1145/3158093} have observed that
optimisations (in particular, let-floating and inlining) can greatly
benefit from linearity analysis and, in order to improve those transformations,
special purpose linear-type-inspired systems were created and implemented.
Preserving linear types in Core throughout the compilation pipeline allows the
optimiser to leverage non-heuristic linearity information where, previously, it
would resort solely to ad-hoc or incomplete custom-built linearity inference
passes (naturally, these passes are still necessary to optimise programs not
using explicit linearity or multiplicity annotations). \ROUNDTWO{Linearity may
potentially be used to the benefit of other transformations.
An obvious candidate is
\emph{inlining}, which is applied based on heuristics from information provided
by the \emph{cardinality analysis} pass that counts occurrences of bound
variables.  Linearity can be used to non-heuristically inform
the inliner~\cite{10.1145/3158093}.}

\subsection{Future Work\label{sec:future-work}}

We highlight some avenues of future work. Briefly,
these include \emph{multiplicity coercions}, optimisations leveraging
linearity, and incorporating Linear Core in GHC Haskell. 

\paragraph{Multiplicity Coercions.}
Linear Core does not have type equality coercions, a flagship feature of GHC
Core's type system.
Coercions allow the Core intermediate language to encode a panoply of Haskell source
type-level features such as GADTs, type families, or newtypes.
In Linear Haskell, multiplicities are introduced as annotations to function
arrows which specify the linearity of the function. In practice,
multiplicities are simply types of kind \incode{Multiplicity}, where \incode{One} and \incode{Many}
are the type constructors of the kind \incode{Multiplicity}; multiplicity polymorphism
follows from type polymorphism, where multiplicity variables are
just type variables. Encoding multiplicities as types allows Haskell programs
to leverage \ROUNDTWO{type-level features when handling multiplicities}.
For instance, it is possible via the use of GADTs to define a
function whose linearity of its second argument depends on the value
of its first argument, internally realized through so-called
\emph{multiplicity coercions}. Currently, Core cannot make use of such
coercions to determine whether the usages of linear resources match
their intended multiplicity. Studying the interaction between coercions and multiplicities is a main avenue
of future work for Linear Core.

\paragraph{Linear Core in the Glasgow Haskell Compiler.}
Integrating Linear Core in the Glasgow Haskell Compiler is one of the ultimate
goals of our work. Core's current type system ignores linearity due to
its limitation in understanding semantic linearity, and our work fills this gap
and would allow Core to be linearly typed all throughout.
Implementing Linear Core in GHC is a challenging endeavour, since we must
account for all other Core features (e.g. strict constructor fields), propagate
new or modified data throughout the entire pipeline, and accept more
optimisations. Despite our initiative in this direction\footnote{\url{https://gitlab.haskell.org/ghc/ghc/-/merge_requests/10310}},
we leave this as future work.

\paragraph{Generalizing Linear Core to Haskell.}
Linear types, despite their compile-time correctness guarantees regarding
resource management, impose a burden on programmers in being a restrictive
typing discipline (witnessed, e.g., by Linear
Constraints~\cite{cite:linearconstraints}). Linear Core eases the restrictions
of linear typing by being more flexible in understanding linearity for lazy
languages. It is future
work to apply our ideas to the surface Haskell language.

\begin{acks}

We would like to thank Arnaud Spiwack for many insightful discussions on
linearity in Core and the reviewers for their feedback, which improved the
paper significantly. This work was supported by Well-Typed LLP and by
national funds through Fundação para a Ciência e a Tecnologia, I.P. (FCT)
under projects UID/50021/2025 and UID/PRR/50021/2025.

\end{acks}

\bibliography{references}

\appendix

 \newpage

\section{Natural and Instrumented Operational Semantics}\label{app:opsem}

We define a lazy big-step operational semantics for $\lambda_\Delta^\pi$ in the
style of Launchbury~\cite{10.1145/158511.158618} natural semantics. The semantics is
equipped with an environment which assigns expressions to variables and is
mutated in order to express lazy evaluation. Terms are explicit about
their sharing (in the sense of shared reductions in lazy evaluation),
and so we perform a translation that makes all sharing explicit
through let-binders. This translation is presented in Figure~\ref{fig:explicitsharing}.
We write $\llet{\ov{x_i{:}_{\Delta_i}\sigma = e_i}}{e'}$ for the iterated
let-binding of variables $\ov{x_i}$.
\begin{figure}
\[
  \begin{array}{lclr}
    x^* & \triangleq & x\\

    (\Lambda p.~e)^* & \triangleq & \Lambda p.~(e)^*\\
    (\lambda x:_\pi\sigma .e)^* & \triangleq & \lambda x:_\pi\sigma . (e)^*\\
    (e~\pi)^* & \triangleq & (e)^*~\pi\\
    (e~x)^* & \triangleq & (e)^*~x\\
    (e~e')^* & \triangleq & \llet{x{:}_{\Delta}\sigma =
                            (e')^*}{(e)^*~x}     \\ &&\qquad \mbox{with
                            $\Gamma;\Delta' \vdash e :
                            \sigma\rightarrow_\pi \tau$} \\ &&\qquad \mbox{and
                                       $\Gamma;\Delta \vdash e' : \sigma$}\\
      K~\ov{e_i} & \triangleq & \llet{\ov{x_i:_{\Delta_i} \sigma_i =
                                   (e_i)^*}}{K~\ov{x_i}} \\ &&
                                       \qquad\mbox{with
                                       $\ov{\Gamma;\Delta_i\vdash
                                       e_i :
                                       \sigma_i}$}\\                                              
(\ccase{e}{z{:}_{\Delta}\sigma~\{\overline{\rho\ROUNDTWO{\Rightarrow} e'}\}})^* &
                                                               \triangleq
                                &
                                  \ccase{(e)^*}{z{:}_{\Delta}\sigma~\{\overline{\rho\ROUNDTWO{\Rightarrow} (e')^*}\}}\\
    (\llet{x{:}_{\Delta}\sigma = e}{e'}      )^* & \triangleq & \llet{x{:}_{\Delta}\sigma  = (e)^*}{(e')^*}    \\
    (\lletrec{\overline{x{:}_{\Delta}\sigma = e}}{e'}  )^* & \triangleq &\lletrec{\overline{x{:}_{\Delta}\sigma = (e)^*}}{(e')^*}  
      \end{array}
    \]
    \caption{Explicit Sharing Translation\label{fig:explicitsharing}}
  \end{figure}

The natural semantics (Figure~\ref{fig:natsemapp}) is defined via the relation
$\Theta : e \Downarrow \Theta' : v$ where $e$ is an expression, $v$ a
value, $\Theta$ and $\Theta'$ are environments with bindings of the
form $x:_\omega\sigma = e$ which assigns the expression $e$ to the
variable $x$ of the given type.

\begin{figure}
\[
  \begin{array}{c}
    \infer[]
    {\,}
    {\Theta : \Lambda p.e \Downarrow \Theta : \Lambda p.e}
    \qquad
    \infer[]
    {\Theta  : e\Downarrow \Theta' : \Lambda p.e' \quad \Theta' :
    e'[\pi/p]\Downarrow \Theta'' : v}
    {\Theta : e~\pi\Downarrow \Theta'' : v }\\[1.5em]
    \infer[]
    {\,}
    {\Theta : \lambda x:_\pi \sigma . e \Downarrow \Theta : \lambda x:_\pi \sigma . e }
    \qquad
    \infer[]
    {\Theta : e \Downarrow \Theta' : \lambda y:_\pi \sigma . e' \quad
    \Theta' : e'[x/y] \Downarrow \Theta'' : v }
    {\Theta : e~x \Downarrow \Theta'' : v}\\[1.5em]
    \infer[]
    {(\Theta , x:_\omega \sigma = e): e \Downarrow \Theta' : v}
    {(\Theta , x:_\omega \sigma = e) : x \Downarrow (\Theta',x:_\omega
    \sigma = v) : v}
    \qquad
    \infer[]
    {(\Theta,x:_\omega \sigma = e) : e' \Downarrow \Theta' : v}
    {\Theta : \llet{x{:}_{\Delta}\sigma = e}{e'}\Downarrow \Theta' : v
    }\\[1.5em]
    \infer[]
    {\,}
    {\Theta : K~\ov{x_i} \Downarrow \Theta : K~\ov{x_i}}
    \qquad
    \infer[]
    {\Theta : e \Downarrow \Theta' : K~\ov{x_i} \quad
    \Theta' : e'\ov{[x_i/y_i]}[K~\ov{x_i}/z] \Downarrow \Theta'' : v }
    {\Theta :
    \ccase{e}{z{:}_{\Delta'}\sigma~\{{\dots,K~\ov{y_i} \to e', {\dots}}\}}
    \Downarrow \Theta'' : v}\\[1.5em]
    \infer[]
    {\Theta : e \Downarrow \Theta' : K~\ov{x_i} \quad
    \Theta' : e'[K~\ov{x_i}/z] \Downarrow \Theta'' : v }
    {\Theta :
    \ccase{e}{z{:}_{\Delta'}\sigma~\{\dots , \_\to
    e'}\} \Downarrow \Theta'' : v}
    \qquad
    \infer[]
    {(\Theta, \overline{x_i:_\omega \sigma_i = e_i}) : e' \Downarrow \Theta'
    : v
    }
    {\Theta : \lletrec{\overline{x_i{:}_{\Delta}\sigma_i = e_i}}{e'}
    \Downarrow \Theta' : v}
    \end{array}
  \]
  \caption{Natural Semantics for $\lambda_\Delta^\pi$\label{fig:natsemapp}}
\end{figure}

Our instrumented, linearity-aware semantics (akin to that
of~\cite{cite:linearhaskell,10.1145/3158093})
is defined in Figure~\ref{fig:instrsemapp}.
In this semantics, linear variables are erased from
the runtime environment, ensuring that any term violating linearity
produces a stuck state. The semantics is defined by
the judgment $\Gamma ; \Delta \vdash (\Theta \mid e) \Downarrow (\Theta'
\mid v) : \sigma , \Sigma$ where $\Gamma$ and $\Delta$ are typing judgments
($\Gamma$ tracks unrestricted variables, $\Delta$ tracks linear
assumptions) for variables occurring in $\Sigma$;
$\Theta$ and $\Theta'$ are runtime environments, consisting of bindings of the form
$x :_\pi \sigma = e$ or $x :_\Delta \sigma = e$ ; $e$ is the expression being evaluated and $v$
its resulting value, both of type $\sigma$; $\Sigma$ is a list of assignments of the form
$e : \tau$, which are expressions in which variables in $\Theta$ are
used.

\begin{figure}
\[
  \begin{array}{c}
    \infer[]
    {\,}
    {\Gamma;\Delta \vdash (\Theta \mid \Lambda p.e) \Downarrow (\Theta \mid
    \Lambda p.e) : \forall p . \sigma , \Sigma }
    \\[1.5em]
    \infer[]
    {\Gamma ; \Delta \vdash (\Theta \mid  e) \Downarrow (\Theta' \mid
    \Lambda p.e') : \forall p . \sigma , \Sigma \quad \Gamma ;\Delta \vdash
      (\Theta' \mid e'[\pi/p]) \Downarrow (\Theta'' \mid v) : \sigma[\pi/p], \Sigma}
    {\Gamma ; \Delta \vdash (\Theta \mid e~\pi) \Downarrow (\Theta''
      \mid v) : \sigma[\pi/p] , \Sigma  }\\[1.5em]
    \infer[]
    {\,}
    {\Gamma ;\Delta \vdash (\Theta \mid \lambda x:_\pi \sigma . e)
    \Downarrow (\Theta \mid \lambda x:_\pi\sigma .e) : \sigma
    \rightarrow_\pi \tau,\Sigma }\\[1.5em]
    \infer[]
    {\Gamma;\Delta \vdash  (\Theta \mid e) \Downarrow (\Theta' \mid
    \lambda y:_1 \sigma . e') : \sigma \rightarrow_1 \tau , x:\sigma, \Sigma
    \quad
    \Gamma ; \Delta \vdash (\Theta' , y:_{1} \sigma = x \mid e') \Downarrow (\Theta''
    \mid v) : \tau , \Sigma  }
    {\Gamma ; \Delta\vdash (\Theta \mid  e~x ) \Downarrow (\Theta''
    \mid v) : \tau , \Sigma }\\[1.5em]
        \infer[]
    {\Gamma;\Delta \vdash  (\Theta \mid e) \Downarrow (\Theta' \mid
    \lambda y:_\omega \sigma . e') : \sigma \rightarrow_\omega \tau , x:\sigma, \Sigma
    \quad
            \Gamma ; \Delta \vdash (\Theta' \mid e'[x/y]) \Downarrow (\Theta''
    \mid v) : \tau , \Sigma  }
    {\Gamma ; \Delta\vdash (\Theta \mid  e~x ) \Downarrow (\Theta''
    \mid v) : \tau , \Sigma }\\[1.5em]
       \infer[]
    {\Gamma ; \Delta  \vdash (\Theta, x:_{\Delta'} \sigma = e \mid e)
    \Downarrow (\Theta' \mid v) : \sigma , \Sigma \\
      \D'' = \D' \setminus (\Theta \setminus \Theta')
      }
    {\Gamma ; \Delta \vdash (\Theta , x:_{\Delta'} \sigma = e \mid x)
      \Downarrow (\Theta' , x:_{\Delta''} \sigma = v \mid v) : \sigma, \Sigma}
    \\[1.5em]
           \infer[]
    {\Gamma ;\Delta  \vdash (\Theta \mid e)
    \Downarrow (\Theta' \mid v) : \sigma,\Sigma}
    {\Gamma ; \Delta \vdash (\Theta , x:_1 \sigma = e \mid x)
    \Downarrow (\Theta' \mid v) : \sigma,\Sigma}
    \\[1.5em]
    
    \infer[]
    {\Gamma ;\Delta \vdash (\Theta,x:_{\Delta'} \sigma = e \mid e')
    \Downarrow (\Theta' \mid v) : \tau,\Sigma}
    {\Gamma ; \Delta \vdash (\Theta \mid \llet{x{:}_{\Delta'}\sigma =
    e}{e'})\Downarrow (\Theta' \mid v) : \tau , \Sigma
    }\\[1.5em]
    
    \infer[]
    {\,}
    {\Gamma ; \Delta \vdash (\Theta \mid K~\ov{x_i}) \Downarrow
    (\Theta \mid K~\ov{x_i}) : T , \Sigma}
    \\[1.5em]
    
    \infer[]
    { K : \ov{ \sigma_i
    \rightarrow_\omega \sigma_j \lolli} T
    \\
    \Gamma, z{:}_{\ov{y_j}}T, \ov{y_i :_\omega \sigma_i} ;\Delta , \ov{y_j :_1 \sigma_j} \vdash (\Theta \mid e) \Downarrow (\Theta' \mid
      K~\ov{x_i}) : T , e' : \tau ,\Sigma \\
    \Gamma ;\Delta \vdash (\Theta' \mid e'\ov{[x_i/y_{i,j}]}[K~\ov{x_i}/z]
    ) \Downarrow (\Theta'' \mid v) :\tau , \Sigma }
    {\Gamma ;\Delta \vdash (\Theta \mid
    \ccase{e}{z{:}_{\Delta'}T~\{{\dots,K~\ov{y_{i,j}} \,\ROUNDTWO{\Rightarrow}\, e', {\dots} }\}})
    \Downarrow (\Theta'' \mid  v) : \tau , \Sigma}\\[1.5em]

    \infer[]
    {K : \ov{ \sigma_i
    \rightarrow_\omega \sigma_j \lolli} T\\
    \Gamma ; \Delta, z{:}_1\s \vdash (\Theta \mid e) \Downarrow (\Theta' \mid
    K~\ov{x_i}) : T , e', \Sigma \quad
    \Gamma ;\Delta \vdash (\Theta' \mid  e'[K~\ov{x_i}/z])
    \Downarrow (\Theta'' \mid v) : \tau,\Sigma }
    {\Gamma ; \Delta \vdash (\Theta \mid 
    \ccase{e}{z{:}_{\Delta'}T~\{\dots , \_\,\ROUNDTWO{\Rightarrow}\,
    e'}\}) \Downarrow (\Theta'' \mid v) : \tau}
    \\[1.5em]

     \infer[]
     {(\Gamma ; \Delta \vdash (\Theta, \overline{x_i:_{\Delta'} \sigma_i =
     e_i} \mid e' )\Downarrow (\Theta'
     \mid v) : \tau
     }
     {\Gamma ;\Delta \vdash (\Theta \mid \lletrec{\overline{x_i{:}_{\Delta'}\sigma_i = e_i}}{e'})
     \Downarrow (\Theta' \mid v) : \tau}

   \end{array}
 \]
 \caption{Instrumented Semantics for $\lambda_\Delta^\pi$\label{fig:instrsemapp}}
 \end{figure}

 Following~\cite{Gunter1993APA,cite:linearhaskell}, we also consider
 partial derivations of our semantics in order to be able to state a
 progress result. Given our set of rules defining judgments of the
 form $a \Downarrow b$ with ordered premises, we define a total
 derivation of $a \Downarrow b$ as a tree in the standard way. A
 partial derivation $a \Downarrow?$ is either a root labelled with $a
 \Downarrow?$ or an application of a rule matching $a$ on the left of
 the evaluation arrow, where exactly one of the premises $a'
 \Downarrow?$ has a partial  derivation, all the premises to the left
 of $a' \Downarrow?$ have a total derivation, and the premises to the
 right of $a'\Downarrow?$ are not known yet.

 By construction, in a partial derivation, there is only one premise
 $b\Downarrow?$ with no sub-derivation. We call $b$ the head of such a
 partial derivation.
We write $a \Downarrow^* b$ for the relation which holds when $b$ is
the head of some partial derivation with root $a \Downarrow?$.  We
write $a \Downarrow^* b$  the partial evaluation relation.
 
Below, we write $\Theta {\uparrow^1}$ for the linear bindings in evaluation
environment $\Theta$ and $\Theta {\uparrow^\omega}$ for all other bindings.
$\Theta\downarrow^1_e$ refers to the expressions bound to linear variables in
the runtime environment and $\Theta\downarrow^1_\sigma$ refers to the types of
those expressions.

 \begin{definition}[Environment Expansion]
   We rely on an expansion of a runtime environment $\Theta$ for the evaluation of a term $e$ into a (typed) term, written $[\Theta \mid e]$, defined inductively as follows:
   \[
     \begin{array}{lcl}
       \left[ \cdot \mid e \right] & \triangleq & e\\
       \left[ \Theta , \xl = e' \mid e \right] & \triangleq  & [\Theta \mid (\lambda \xl . e)~e']\\
       \left[\Theta , x:_\Delta \sigma = e' \mid e \right] & \triangleq & [\Theta \mid \llet{x:_\Delta \sigma = e'}{e}]
    \end{array}
   \]
 \end{definition}

We note that for the treatment of recursive let groups, environment
expansion needs to happen for the entire group of let-bound
variables. This is easily achieved by extending the runtime
environment with information about recursive let groupings, which we
omit for the sake of presentation.
 
 \begin{definition}[Well-typed State]
   We write $\G;\D \vdash \Theta \mid e : \tau,\Sigma$ to denote the following typing judgment:

   \[
\G ; \D \vdash [ \Theta \mid (e,\mathit{terms}(\Sigma))] : \tau \otimes \bigotimes(\Sigma)
   \]
   where $\mathit{terms}(e_1 : \sigma_1 , \dots , e_n : \sigma_n)$ stands for the tuple $(e_1 ,( \dots , (e_n , ())))$ and
   $\bigotimes(e_1 : \sigma_1 , \dots , e_n : \sigma_n)$ for $\sigma_1 \otimes ( \dots (\sigma_n \otimes ()))$.
 \end{definition}

 \begin{lemma}~\label{lem:unwrapenv}
   $\G ; \D \vdash [ \Theta \mid e] : \sigma$
     iff $\G , \Theta\uparrow^{\omega} ; \D , \Theta\uparrow^1 \vdash  (e, \Theta\downarrow^1_e) : \sigma \otimes \bigotimes(\Theta\downarrow^1_\sigma)$
   \end{lemma}
\begin{proof}
By structural induction on the environment expansion definition in the
forward direction and by a straightforward induction on $\Theta$ in the
converse direction. We show only the forward direction.
\begin{description}
\item[Case:] $\left[\Theta, \xl = y \mid e\right] \triangleq \left[\Theta \mid (\lambda \xl. e)~y\right]$
\begin{tabbing}
    (1) $\G; \D \vdash \left[ \Theta, \xl = y \mid e \right] : \phi$\\
    (2) $\G; \D \vdash \left[ \Theta \mid (\lambda \xl. e)~y \right] : \phi$\` by def.\\
    (3) $\G, \Theta\uparrow^\omega; \D, \Theta\uparrow^1 \vdash ((\lambda \xl. e)~y, \Theta\downarrow^1_e) : \phi \otimes \bigotimes(\Theta\downarrow^1_\sigma)$\` by i.h.\\
    (4) $Let~\D_1,\D_2,\D_3 = \D, \Theta\uparrow^1 $\\
    (5) $\G, \Theta\uparrow^\omega; \D_3 \vdash \Theta\downarrow^1_e : \bigotimes(\Theta\downarrow^1_\sigma)$\\
    (6) $\G, \Theta\uparrow^\omega; \D_2 \vdash y : \s$\` by inv\\
    (7) $\G, \Theta\uparrow^\omega; \D_1 \vdash (\lambda \xl. e) : \s \to \phi $\` by inv\\
    (8) $\G, \Theta\uparrow^\omega; \D_1, \xl \vdash e : \phi $\` by inv\\
    (9) $\G, \Theta\uparrow^\omega; \D_1, \D_2, \xl \vdash (e, y) : \phi \otimes \sigma$\` by pair\\
    (10) $\G, \Theta\uparrow^\omega; \D, \xl \vdash (e, y, \Theta\downarrow^1_e) : \phi \otimes \sigma \otimes \bigotimes(\Theta\downarrow^1_\sigma)$\` by pair\\
\end{tabbing}
\item[Case:] $\left[\Theta, \xD = e' \mid e\right] \triangleq \left[\Theta \mid \llet{\xD = e'}{e}\right]$
\begin{tabbing}
    (1) $\G; \D' \vdash \left[ \Theta, \xD = y \mid e \right] : \phi$\\
    (2) $\G; \D' \vdash \left[ \Theta \mid \llet{\xD = e'}{e} \right] : \phi$\` by def.\\
    (3) $\G, \Theta\uparrow^\omega; \D', \Theta\uparrow^1 \vdash (\llet{\xD = e'}{e}, \Theta\downarrow^1_e) : \phi \otimes \bigotimes(\Theta\downarrow^1_\sigma)$\` by i.h.\\
    (4) $Let~\D,\D_2,\D_3 = \D', \Theta\uparrow^1 $\\
    (5) $\G, \Theta\uparrow^\omega; \D_3 \vdash \Theta\downarrow^1_e : \bigotimes(\Theta\downarrow^1_\sigma)$\\
    (6) $\G, \Theta\uparrow^\omega; \D \vdash e' : \s$\` by inv on let\\
    (7) $\G, \Theta\uparrow^\omega, \xD; \D,\D_2 \vdash e : \phi $\` by inv on let\\
    (8) $\G, \Theta\uparrow^\omega, \xD; \D,\D_2,\D_3 \vdash (e, \Theta\downarrow^1_e) : \phi \otimes \bigotimes(\Theta\downarrow^1_\sigma)$\` by pair\\
\end{tabbing}
\end{description}
\end{proof}

\subsection{Bisimulation}

We relate our two operational semantics through a bisimulation,
allowing us to transfer type safety of the instrumented semantics to
the lazy big-step semantics. The relation essentially allows the
execution environment to treat some unrestricted bindings
of the natural semantics as linear.

\begin{definition}[Bisimulation]
We define the relation $\gamma(\Theta : e)(\G ; \D \vdash \Theta' \mid
e' : \tau,\Sigma)$ between a well-typed state and an ordinary state of
the lazy big-step semantics if $e = e'$ and $\Theta'{\uparrow^\omega}
\leq \Theta$.
\end{definition}

\begin{lemma}~\label{lem:naturaltoinstr}
  \begin{itemize}
    \item For all $\gamma(\Theta : e)(\G ; \D \vdash \Theta' \mid
e : \tau,\Sigma)$ such that $\Theta : e \Downarrow \Theta'' : v$
there exists a well-typed state $\G ; \D \vdash (\Theta''' \mid v) :
\tau,\Sigma$ such that $\G ; \D \vdash (\Theta' \mid e) \Downarrow
(\Theta''' \mid v) : \tau,\Sigma$ and $\gamma(\Theta'' : v)(\G ; \D
\vdash \Theta''' \mid v : \tau,\Sigma)$
\item For all $\gamma(\Theta : e)(\G ; \D \vdash \Theta' \mid
e : \tau,\Sigma)$ such that $\Theta : e \Downarrow^* \Theta'' : v$
there exists a well-typed state $\G ; \D \vdash (\Theta''' \mid v) :
\tau,\Sigma$ such that $\G ; \D \vdash (\Theta' \mid e) \Downarrow^*
(\Theta''' \mid v) : \tau,\Sigma$ and $\gamma(\Theta'' : v)(\G ; \D
\vdash \Theta''' \mid v : \tau,\Sigma)$
\end{itemize}
\end{lemma}
\begin{proof}
By induction on the given operational semantics rules.
  \end{proof}

\begin{lemma}~\label{lem:instrtonatural}
  \begin{itemize}
\item  For all $\gamma(\Theta : e)(\G ; \D \vdash \Theta' \mid
e : \tau,\Sigma)$ such that $\G ; \D \vdash (\Theta' \mid e) \Downarrow
(\Theta'' \mid v) : \tau,\Sigma$, there exists a state $\Theta''' : v$
such that $\Theta : e \Downarrow \Theta''' : v$ and
$\gamma(\Theta''' : v)(\G ; \D \vdash (\Theta'' \mid v) :
\tau,\Sigma)$
\item For all $\gamma(\Theta : e)(\G ; \D \vdash \Theta' \mid
e : \tau,\Sigma)$ such that $\G ; \D \vdash (\Theta' \mid e) \Downarrow^*
(\Theta'' \mid v) : \tau,\Sigma$, there exists a state $\Theta''' : v$
such that $\Theta : e \Downarrow^* \Theta''' : v$ and
$\gamma(\Theta''' : v)(\G ; \D \vdash (\Theta'' \mid v) :
\tau,\Sigma)$
   \end{itemize}
 \end{lemma}
\begin{proof}
By induction on the given operational semantics rules.
  \end{proof}

\section{Type Safety Proofs}\label{app:proofs}

In the following we write $\Mapsto$ to stand for
$\Rightarrow_\mathsf{WHNF}$ and $\Rrightarrow$ for $\Rightarrow_\mathsf{NWHNF}$.

\subsection{Auxiliary Lemmas}
The following lemmas presuppose the variable names across the typing
contexts are distinct.

  \begin{lemma}[Linear to $\Delta$-bound]~\label{lem:onedelta-app}

    \begin{enumerate}
    
      \item If $\G; \D',\xl \vdash e :\vp$
        then $\G[\D/x],\xD; \D,\D' \vdash e : \vp$ (with $\Delta$ fresh).
      \item If $\G ; \D' , \xl \vdash_{alt} \rho \rightarrow e
        :^z_{\Delta_s} \sigma \Rightarrow \rho$ then
        $\G[\D/x],\xD;\D,\D' \vdash_{alt} \rho \rightarrow e
        :^z_{\Delta_s} \sigma \Rightarrow \rho$ (with $\Delta$ fresh).
   
      \end{enumerate}
\end{lemma}
\begin{proof}
Straightforward induction on the structure of the given derivation,
reconstructing the original derivation and applying rule
$\mathit{Var}_\Delta$ where rule $\mathit{Var}_1$ was applied on
$x$ in the original derivation.
\end{proof}
\onedelta*

\begin{proof}
By Lemma~\ref{lem:onedelta-app}(1).
\end{proof}

\begin{lemma}[$\Delta$-bound to Linear]~\label{lem:deltaone-app}

  \begin{enumerate}
 \item If $\G,\x[\irr{\D}]; \irr{\D},\D' \vdash e : \vp$
   then $\G[x/\irr{\D}]; \D',\xl \vdash e :\vp$.
 \item If $\G ,\x[\irr{\D}]; \irr{\D}, \D' \vdash_{alt} \rho \rightarrow e
   :^z_{\Delta_s} \sigma \Rightarrow \rho$ then
   $\G[x/\irr{\D}];\D',\xl \vdash_{alt} \rho \rightarrow e
   :^z_{\Delta_s} \sigma \Rightarrow \rho$.
  \end{enumerate}
\end{lemma}
\begin{proof}
Straightforward induction on the structure of the given derivation,
reconstructing the original derivation and applying rule
$\mathit{Var}_1$ where rule $\mathit{Var}_\D$ was applied on
$x$ in the original derivation.
\end{proof}

\deltaone*
\begin{proof}
By Lemma~\ref{lem:deltaone-app}(1).
\end{proof}

\begin{lemma}[Unrestricted and $\Delta$-bound]~\label{lem:undelta-app}
  \begin{enumerate}
    \item $\G,\xo; \D \vdash e : \vp$ iff $\G,\x[\cdot]; \D \vdash e :
      \vp$
    \item $\G,\xo;\D \vdash_{alt} \rho \rightarrow e
      :^z_{\Delta_s} \sigma \Rightarrow \rho$ iff
      $\G,\x[\cdot];\D \vdash_{alt} \rho \rightarrow e
   :^z_{\Delta_s} \sigma \Rightarrow \rho$
  \end{enumerate}
\end{lemma}
\begin{proof}
Straightforward induction on the structure of the given derivation,
noting that rules $\mathit{Var}_\Delta$ and $\mathit{Var}_\omega$
coincide when $\Delta$ is empty.
\end{proof}

\undelta*

\begin{proof}
By Lemma~\ref{lem:undelta-app}.
\end{proof}

\subsection{Irrelevance\label{sec:proof:irrelevance}}

\irrelevancelemma*

\begin{proof}
By structural induction on the case alternative typing derivation.

\begin{description}

\item[Case:] $Alt\_$
\begin{tabbing}
    (1) $\G,\z[\irr{\D}];\irr{\D}, \D' \vdash \_ \ROUNDTWO{\Rightarrow} e :^z_{\irr{\D}} \s $\\
    (2) $\G, \z[\irr{\D}];\irr{\D},\D' \vdash e : \s$\\
    (3) $\irr{\D}$ is used through $z$\`since $\irr{\D}$ can't otherwise be used\\\`and is introduced in this alternative uniquely \\\`(since we have multi-tier proof irrelevance, i.e. $\irr{\irr{\D}}\neq\irr{\D}$)\\
    (4) $\G[z/\irr{\D}]; \z[1], \D' \vdash e : \s$\` by Lemma~\ref{lem:deltaone} (2,3)\\
    (5) $\G[\ov{\D_i}/\irr{\D}], \z[\ov{\D_i}]; \ov{\D_i}, \D' \vdash e : \s$\` by Lemma~\ref{lem:onedelta} (4)\\
    (6) $\G, \z[\ov{\D_i}]; \ov{\D_i}, \D' \vdash e : \s$\` by (3), resources in $\G$ with $\irr{\D}$ (parts) in\\\` the usage environment cannot be used in $e$\\
    (7) $\G, \z[\ov{\D_i}]; \ov{\D_i}, \D' \vdash \_ \,\ROUNDTWO{\Rightarrow}\, e :^z_{\ov{\D_i}} \s $\` by $Alt\_$\\
\end{tabbing}

\item[Case:] $Alt0$
\begin{tabbing}
    (1) $\G, \z[\irr{\D}]; \irr{\D}, \D' \vdash K~\ov{\xo}\ROUNDTWO{\Rightarrow} e :^z_{\irr{\D}} \s $\\
    (2) $\G, \z[\cdot], \ov{\xo}; \D' \vdash e : \s$\` by inv. on $Alt0$ and def. of empty subst.\\
    (3) $\G, \z[\ov{\D_i}]; \ov{\D_i}, \D' \vdash K~\ov{\xo}\ROUNDTWO{\Rightarrow} e :^z_{\ov{\D_i}} \s $\` by $Alt0$ (2)\\
\end{tabbing}

\item[Case:] $AltN_{\textrm{WHNF}}$
\begin{tabbing}
    Not applicable since $\Rrightarrow$ is only generalized by $\Rightarrow$, not $\Mapsto$.
\end{tabbing}

\item[Case:] $AltN_{\textrm{Not WHNF}}$
    We prove the theorem for constructing both an $AltN_{\textrm{Not WHNF}}$
        and an $AltN_{\textrm{WHNF}}$ from a proof-irrelevant $AltN_{\textrm{Not
        WHNF}}$, to prove the statement holds for any $\Rightarrow$ kind for
        $AltN$, rather than just $\Rrightarrow$ or $\Mapsto$.
\begin{tabbing}
    (1) $\G, \z[\irr{\D}]; \irr{\D}, \D' \vdash K~\ov{\xo},\ov{y_i{:}_{1}\s_i}^n\ROUNDTWO{\Rightarrow} e :^z_{\irr{\D}} \s $\\
    (2) $\G,\z[\irr{\D}],\ov{\xo},\ov{y_j{:}_{\D_j}\s_j}^n; \irr{\D}, \D' \vdash e : \s$\`by inv. on $AltN_{\textrm{Not WHNF}}$\\
    (3) $\ov{\D_j} = \ov{\lctag{\irr{\D}}{K_j}}$\`by inv. on $AltN_{\textrm{Not WHNF}}$\\
    Subcase $\irr{\D}$ is consumed through $z$\\
    (4) $\G,\ov{\xo},\ov{y_j{:}_{\D_j[z/\irr{\D}]}\s_j}^n; \z[1], \D' \vdash e : \s$\`by Lemma~\ref{lem:deltaone} (2,subcase)\\\`and vars in $\G$ with $\irr{\D}$ cannot be used by subcase\\
    Subcase constructing $AltN_{\textrm{WHNF}}$ ($\Mapsto$)\\
    (5) $\G,\ov{\xo},\ov{y_j{:}_{\D_i}\s_j}^n; \z[1], \D' \vdash e : \s$\`by $\ov{y}$ is not used, and $Weaken$\\
    (6) $\G,\ov{\xo},\ov{y_j{:}_{\D_i}\s_j}^n, \z[\ov{\D_i}]; \ov{\D_i}, \D' \vdash e : \s$\`by Lemma~\ref{lem:onedelta} (5)\\\` and $\G$ vars do not mention $z$ by (4) \\
    (7) $\G, \z[\ov{\D_i}]; \ov{\D_i}, \D' \vdash K~\ov{\xo},\ov{y_i{:}_{1}\s_i}^n \ROUNDTWO{\Mapsto} e :^z_{\ov{\D_i}} \s $\`by $AltN_{\textrm{WHNF}}$ (6)\\
    Subcase constructing $AltN_{\textrm{Not WHNF}}$ ($\Rrightarrow$)\\
    (5) $\G,\ov{\xo},\ov{y_j{:}_{\lctag{\ov{\D_i}}{K_j}}\s_j}^n; \z[1], \D' \vdash e : \s$\`by $\ov{y}$ does not consume resources\\
    (6) $\G,\ov{\xo},\ov{y_j{:}_{\lctag{\ov{\D_i}}{K_j}}\s_j}^n, \z[\ov{\D_i}]; \ov{\D_i}, \D' \vdash e : \s$\`by Lemma~\ref{lem:onedelta} (5)\\\` and $\G$ vars do not mention $z$ by (4) \\
    (7) $\G, \z[\ov{\D_i}]; \ov{\D_i}, \D' \vdash K~\ov{\xo},\ov{y_i{:}_{1}\s_i}^n \ROUNDTWO{\Rightarrow} e :^z_{\ov{\D_i}} \s $\`by $AltN_{\textrm{Not WHNF}}$ (6)\\
    Subcase $\irr{\D}$ is (fully) consumed by $\ov{y}$ (after splitting)\\
    (4) $\G,\z[\ov{y}],\ov{\xo};\ov{y_j{:}_{1}\s_j}^n, \D' \vdash e : \s$\`by Lemma~\ref{lem:deltaone} (2, subcase)\\\` and vars in $\G$ with $\irr{\D}$ cannot be used by subcase\\
    (5) $\G,\z[\ov{\D_i}],\ov{\xo};\ov{y_j{:}_{1}\s_j}^n, \D' \vdash e : \s$\`by $Weaken$ and $z$ does not consume resources\\
    Subcase constructing $AltN_{\textrm{WHNF}}$ ($\Mapsto$)\\
    (6) $\G,\z[\ov{\D_i}],\ov{\xo},\ov{y_j{:}_{\D_i}\s_j}^n; \ov{\D_i}, \D' \vdash e : \s$\`by Lemma~\ref{lem:onedelta} (5)\\\`and vars in $\G$ do not mention $\ov{y}$ (4,5)\\
    (7) $\G,\z[\ov{\D_i}]; \ov{\D_i}, \D' \vdash K~\ov{\xo},\ov{y_i{:}_{1}\s_i}^n\ROUNDTWO{\Mapsto}e :^z_{\ov{\D_i}} \s$\`by $AltN_{\textrm{WHNF}}$\\
    Subcase constructing $AltN_{\textrm{Not WHNF}}$ ($\Rrightarrow$)\\
    (6) $\G,\z[\ov{\D_i}],\ov{\xo},\ov{y_j{:}_{\lctag{\ov{\D_i}}{K_j}}\s_j}^n; \ov{\D_i}, \D' \vdash e : \s$\`by $1 \Rightarrow \D$ lemma (5)\\\`and vars in $\G$ do not mention $\ov{y}$ (4,5)\\
    (7) $\G,\z[\ov{\D_i}]; \ov{\D_i}, \D' \vdash K~\ov{\xo},\ov{y_i{:}_{1}\s_i}^n\ROUNDTWO{\Rrightarrow} e :^z_{\ov{\D_i}} \s $\`by $AltN_{\textrm{Not WHNF}}$\\
\end{tabbing}

\end{description}
\end{proof}

\subsection{Substitution Lemmas\label{sec:proof:substitution-lemmas}}

The linear substitution lemma states that a well-typed expression $e$ with a
linear variable $x$ of type $\s$ remains well-typed if
occurrences of $x$ in the $e$ are replaced by an expression $e'$ of the same
type $\s$, and occurrences of $x$ in the linear context and in usage
environments of $\D$-bound variables are replaced by the linear context $\D'$
used to type $e'$:
\LinearSubstitutionLemma

\noindent Where $\G[\D'/x]$ substitutes all occurrences of $x$ in the usage
environments of $\D$-variables in $\G$ by the linear variables in
$\D'$.
We further require that the environment annotated in the case
alternative judgement, $\D_s$, is a subset of the environment used to type the
whole alternative $\D_s \subseteq \D$. In all occurrences of the alternative
judgement (in $Case_{\textrm{WHNF}}$ and $Case_{\textrm{Not WHNF}}$), the
environment annotating the alternative judgement is \emph{always} a subset of
the alternative environment.

\begin{proof}
By structural induction on the given derivation. We omit most trivial cases.

Statement (1):

\begin{description}
  
\item[Case:] $\Lambda I$
\begin{tabbing}
  (1) $\Gamma; \D, x{:}_1\sigma \vdash \Lambda p.~e : \forall p.~\varphi$\\
  (2) $\G; \D' \vdash e' : \sigma$\\
  (3) $\G, p; \D,x{:}_1\sigma \vdash e : \varphi$ \` by inversion on $\Lambda I$\\
  (4) $p \notin \Gamma$ \` by inversion on $\Lambda I$\\
  (5) $\G[\D'/x],p; \D,\D' \vdash e[e'/x] : \varphi$ \` by i.h.(1) by (2,3)\\
  (6) $\G[\D'/x];\D,\D' \vdash \Lambda p.~e[e'/x] : \forall p.~\varphi$ \` by $\Lambda I$ (4,5)\\
  (7) $(\Lambda p.~e)[e'/x] = (\Lambda p.~e[e'/x])$ \` by def. of substitution\\
\end{tabbing}

\item[Case:] $\lambda I_1$
\begin{tabbing}
  (1) $\G; \D, \xl \vdash \lambda y{:}_1\sigma'.~e : \sigma' \to_1 \varphi$\\
  (2) $\G; \D' \vdash e' : \sigma$\\
  (3) $\G; \D, \xl, y{:}_1\sigma' \vdash e : \varphi$ \` by inversion on $\lambda I$\\
  (4) $\G[\D'/x]; \D, y{:}_1\sigma', \D' \vdash e[e'/x] : \varphi$ \` by i.h.(1) by (2,3)\\
  (5) $\G[\D'/x]; \D,\D' \vdash \lambda y{:}_1\sigma'.~e[e'/x] : \sigma' \to_1 \varphi$ \` by $\lambda I$ (4)\\
  (6) $(\lambda y{:}_1\sigma'.~e)[e'/x] = (\lambda y{:}_1\sigma'.~e[e'/x])$ \` by def. of substitution\\
\end{tabbing}

\item[Case:] $Var_1$
\begin{tabbing}
  (1) $\G;x{:}_1\sigma \vdash x : \sigma$\\
  (2) $\G;\D' \vdash e' : \sigma$\\
  (3) $\G[\D'/x]; \D' \vdash e' : \s$\` by weaken\\
  (4) $x[e'/x] = e'$ \` by def. of substitution\\
(5) $\G[\D'/x];\D' \vdash e' : \sigma$ \` by (3)\\
\end{tabbing}

\item[Case:] $Var_\omega$\\
  (1) Impossible. $x{:}_1\sigma$ can't be in the context.\\

\item[Case:] $Var_\Delta$
\begin{tabbing}
  (1) $\G,y{:}_{\Delta,\xl}\varphi; \D, x{:}_1\sigma \vdash y : \varphi$\\
  (2) $\G; \D' \vdash e' : \sigma$\\
  (3) $y[e'/x] = y$\\
  (4) $\G[\D'/x],y{:}_{\Delta,\D'}\varphi; \D,\D'; \vdash y : \varphi$ \` by $Var_\Delta$\\
\end{tabbing}

\item[Case:] $\lambda E_1$
\begin{tabbing}
  (1) $\G; \D,\D'',\xl \vdash e~e'' : \vp$\\
  (2) $\G; \D' \vdash e' : \s$\\
  Subcase $x$ occurs in $e$\\
  (3) $\G; \D,\xl \vdash e : \s' \to_1 \vp$\` by inversion on $\lambda E_1$\\
  (4) $\G; \D'' \vdash e'' : \s'$\` by inversion on $\lambda E_1$\\
  (5) $\G[\D'/x]; \D, \D' \vdash e[e'/x] : \s' \to_1 \vp$\` by i.h.(1) (2,3)\\
  (6) $\G[\D'/x]; \D, \D', \D'' \vdash e[e'/x]~e'' : \vp$\` by $\lambda E_1$\\
  (7) $(e[e'/x]~e'') = (e~e'')[e'/x]$ \` because $x$ does not occur in $e''$\\
  Subcase $x$ occurs in $e''$\\
  (3) $\G; \D \vdash e : \s' \to_1 \vp$\` by inversion on $\lambda E_1$\\
  (4) $\G; \D'', \xl \vdash e'' : \s'$\` by inversion on $\lambda E_1$\\
  (5) $\G[\D'/x]; \D'',\D' \vdash e''[e'/x] : \s'$\` by i.h.(1) (2,4)\\
  (6) $\G[\D'/x]; \D,\D',\D'' \vdash e~e''[e'/x] : \vp$\` by $\lambda E_1$\\
(7) $(e~e''[e'/x]) = (e~e'')[e'/x]$ \` because $x$ does not occur in $e$\\
\end{tabbing}

\item[Case:] $\lambda E_\omega$
\begin{tabbing}
  (1) $\G; \D, \xl \vdash e~e'' : \vp$\\
  (2) $\G; \D' \vdash e' : \s$\\
  (3) $x$ does not occur in $e''$\` by $e''$ linear context is empty\\
  (4) $\G; \D, \xl \vdash e : \s' \to_\omega \vp$\` by inversion on $\lambda E_\omega$\\
  (5) $\G; \cdot \vdash e'' : \s'$\` by inversion on $\lambda E_\omega$\\
  (6) $\G[\D'/x]; \D, \D' \vdash e[e'/x] : \s' \to_\omega \vp$\` by i.h.(1) (2,4)\\
  (7) $\G[\D'/x]; \D, \D' \vdash e[e'/x]~e'' : \vp$\` by $\lambda E_\omega$\\
(8) $(e[e'/x]~e'') = (e~e'')[e'/x]$ \` because $x$ does not occur in $e''$\\
\end{tabbing}

\item[Case:] $Let$
\begin{tabbing}
  (1) $\G; \D' \vdash e' : \s$\\
  Subcase $x$ occurs in $e$\\
  (2) $\G; \D, \xl, \D'' \vdash \llet{\y[\D,\xl][\s'] = e}{e''} : \vp$\\
  (3) $\G,\y[\D,\xl][\s']; \D, \xl, \D'' \vdash e'' : vp$\` by inversion on $Let$\\
  (4) $\G; \D, \xl \vdash e : \s'$\` by inversion on $Let$\\
  (5) $\G[\D'/x],\y[\D,\D'][\s']; \D, \D', \D'' \vdash e''[e'/x]$\` by i.h.(1) $(1,3)$\\
  (6) $\G[\D'/x]; \D, \D' \vdash e[e'/x] : \s'$ \` by i.h.(1) $(1,4)$\\
(7) $\G[\D'/x]; \D, \D', \D'' \vdash \llet{\y[\D,\D'][\s'] = e[e'/x]}{e''[e'/x]} : \vp$ \` (5,6) by $Let$\\
(8) $(\llet{\y[\D,\D'][\s'] = e[e'/x]}{e''[e'/x]}) = (\llet{\y[\D,\D'][\s'] = e}{e''})[e'/x]$ \` by subst.\\
  Subcase $x$ does not occur in $e$\\
  (2) $\G; \D, \D'', \xl  \vdash \llet{\yD = e}{e''} : \vp$\\
  (3) $\G,\yD; \D, \D'', \xl \vdash e'' : \vp$ \` by inversion on $Let$\\
  (4) $\G; \D \vdash e : \s'$ \` by inversion on $Let$\\
  (5) $\G[\D'/x],\yD; \D, \D', \D'' \vdash e''[e'/x] : \vp$\` by i.h.(1) (1,3)\\
  (6) $\G[\D'/x]; \D, \D', \D'' \vdash \llet{\yD = e}{e''[e'/x]} : \vp$ \` by $Let$ (2,5,6)\\
(7) $\llet{\yD = e}{e''[e'/x]} = (\llet{\yD = e}{e''})[e'/x]$\` by $x$ does
  not occur in $e$\\
\end{tabbing}

\item[Case:] $LetRec$
\begin{tabbing}
  (1) $\G; \D' \vdash e' : \s$\\
Subcase $\xl$ occurs in some $e_i$\\
  (2) $\G; \D, \xl, \D'' \vdash \lletrec{\ov{\var[y_i][\D,\xl][\s_i] = e_i}}{e''} : \vp$\\
  (3) $\G, \ov{\var[y_i][\D,\xl][\s_i]}; \D, \xl, \D'' \vdash e'' : \vp$\` by inversion on $LetRec$\\
  (4) $\ov{\G,\ov{\var[y_i][\D, \xl][\s_i]}; \D, \xl \vdash e_i : \s_i}$\` by inversion on $LetRec$\\
  (5) $\G[\D'/x], \ov{\var[y_i][\D, \D'][\s_i]}; \D, \D', \D'' \vdash e''[e'/x] : \vp$ \` by i.h.(1) (1,3)\\
  (6) $\ov{\G, \ov{\var[y_i][\D, \D'][\s_i]}; \D, \D' \vdash e_i[e'/x] : \s_i}$\` by i.h.(1) (1,4)\\
  (7) $\G[\D'/x]; \D, \D', \D'' \vdash \lletrec{\ov{\var[y_i][\D,\G'_1][\s_i] = e_i[e'/x]}}{e''[e'/x]} : \vp$\` by $LetRec$\\
  (8) $(\lletrec{\ov{\var[y_i][\D,\D'][\s_i] = e_i}}{e''})[e'/x] = \lletrec{\ov{\var[y_i][\D,\D'][\s_i] = e_i[e'/x]}}{e''[e'/x]}$\\
  Subcase $\xl$ does not occur in any $e_i$\\
  (2) $\G; \D, \xl, \D'' \vdash \lletrec{\ov{\var[y_i][\D][\s_i] = e_i}}{e''} : \vp$\\
  (3) $\G, \ov{\var[y_i][\D][\s_i]}; \D, \xl, \D'' \vdash e'' : \vp$\` by inversion on $LetRec$\\
  (4) $\ov{\G, \ov{\var[y_i][\D][\s_i]}; \D \vdash e_i : \s_i}$\` by inversion on $LetRec$\\
  (5) $\G[\D'/x], \ov{\var[y_i][\D][\s_i]}; \D, \D', \D'' \vdash e''[e'/x] : \vp$\` by i.h.(1) (1,3)\\
  (6) $\G[\D'/x]; \D, \D', \D'' \vdash \lletrec{\ov{\var[y_i][\D][\s_i] = e_i}}{e''[e'/x]} : \vp$\` by $LetRec$\\
  (7) $\lletrec{\ov{\var[y_i][\D][\s_i] = e_i}}{e''[e'/x]} = (\lletrec{\ov{\var[y_i][\D][\s_i] = e_i}}{e''})[e'/x]$\` by subcase\\
\end{tabbing}

\item[Case:] $CaseWHNF$
\begin{tabbing}
    (1) $\judg[\G][\D']{e'}{\s}$\\
    Subcase $x$ occurs in $e$\\
    (2) $\judg[\G][\D,\xl,\D'']{\ccase{e}{\z[\D,\xl][\s']~\{\ov{\rho\Rightarrow e''}\}}}{\vp}$\\
    (3) $e$ is in WHNF\\
    (4) $\judg[\G][\D,\xl]{e}{\s'}$\\
    (5) $\rho_j~\textrm{matches}~e$
    (6) $\judg[\G,\zr{\D,\xl}{\s'}][\D,\xl,\D'']{\rho_j \Mapsto  e''}{\vp}[][\D,\xl][z]$\\
    (7) $\ov{\judg[\G,\zr{\irr{\D,\xl}}{\s'}][\irr{\D,\xl},\D'']{\rho\Rrightarrow e''}{\vp}[][\irr{\D,\xl}][z]}$\`by inv.\\
    (8) $\judg[\subst{\G}{\D'}{x}][\D,\D']{e[e'/x]}{\vp}$\`by i.h.(1)\\
    (9) $\judg[\subst{\G}{\D'}{x},\zr{\D,\D'}{\s'}][\D,\D',\D'']{\rho_j \Mapsto e''[e'/x]}{\ \vp}[][\D,\D'][z]$\` by i.h.(2)\\
    (10) $\ov{\judg[\subst{\G}{\D'}{x},\zr{\irr{\D,\D'}}{\s'}][\irr{\D,\D'},\D'']{\rho\Rrightarrow e''[e'/x]}{ \vp}[alt][\irr{\D,\D'}][z]}$ \` by Irrelevance\\
    (11) $\judg[\subst{\G}{\D'}{x}][\D,\D',\D'']{\ccase{e[e'/x]}{\z[\D,\D'][\s']~\{\ov{\rho\to e''[e'/x]}\}}}{\vp}$\\
    Subcase $x$ occurs in $\ov{e''}$\\
    (2) $\judg[\G][\D,\D'',\xl]{\ccase{e}{\z[\D][\s']~\{\ov{\rho\Rightarrow e''}\}}}{\vp}$\\
    (3) $e$ is in WHNF\\
    (4) $\rho_j$ matches $e$\\
    (5) $\judg[\G][\D]{e}{\s'}$\\
    (6) $\judg[\G,\zr{\D}{\s'}][\D,\D'',\xl]{\rho_j \to e''}{\s' \Mapsto \vp}[alt][\D][z]$\\
    (7) $\ov{\judg[\G,\zr{\irr{\D}}{\s'}][\irr{\D},\D'',\xl]{\rho\Rrightarrow e''}{\vp}[][\irr{\D}][z]}$\`by inv.\\
    (8) $e[e'/x] = e$\` by $x$ does not occur in $e$\\
    (9) $\judg[\subst{\G}{\D'}{x},\zr{\D}{\s'}][\D,\D'',\D']{\rho_j \Mapsto e''[e'/x]}{\vp}[][\D][z]$\`by i.h.(2)\\
    (10) $\ov{\judg[\subst{\G}{\D'}{x},\zr{\irr{\D}}{\s'}][\irr{\D},\D'',\D']{\rho\Rrightarrow e''[e'/x]}{ \vp}[][\irr{\D}][z]}$\`by i.h.(2)\\
    (11) $\judg[\subst{\G}{\D'}{x}][\D,\D'',\D']{\ccase{e}{\z[\D][\s']~\{\ov{\rho\Rightarrow e''[e'/x]}\}}}{\vp}$\\
\end{tabbing}

\item[Case:] $CaseNotWHNF$
\begin{tabbing}
  (1) $\judg[\G][\D']{e'}{\s}$\\
  Subcase $x$ occurs in $e$\\
    (2) $\judg[\G][\D,\xl,\D'']{\ccase{e}{\z[\irr{\D,\xl}][\s']~\{\ov{\rho\Rightarrow e''}\}}}{\vp}$\\
    (3) $e$ is definitely not in WHNF\\
    (4) $\judg[\G][\D,\xl]{e}{\s'}$\`by inv.\\
    (5) $\ov{\judg[\G,\zr{\irr{\D,\xl}}{\s'}][\irr{\D,\xl},\D'']{\rho\Rrightarrow e''}{\vp}[][\irr{\D,\xl}][z]}$\`by inv.\\
    (6) $\judg[\subst{\G}{\D'}{x}][\D,\D']{e[e'/x]}{\vp}$\`by i.h.(1)\\
    (7) $\ov{\judg[\subst{\G}{\D'}{x},\zr{\irr{\D,\D'}}{\s'}][\irr{\D,\D'},\D'']{\rho\Rrightarrow e''[e'/x]}{ \vp}[][\irr{\D,\D'}][z]}$ \` by Irrelevance\\
(8) $\judg[\subst{\G}{\D'}{x}][\D,\D',\D'']{\ccase{e[e'/x]}{\z[\D,\D'][\s']~\{\ov{\rho\Rightarrow e''[e'/x]}\}}}{\vp}$\\
  Subcase $x$ occurs in $\ov{e''}$\\
    (2) $\judg[\G][\D,\D'',\xl]{\ccase{e}{\z[\irr{\D}][\s']~\{\ov{\rho\Rightarrow e''}\}}}{\vp}$\\
    (3) $e$ is definitely not in WHNF\\
    (4) $\judg[\G][\D]{e}{\s'}$\`by inv.\\
    (5) $\ov{\judg[\G,\zr{\irr{\D}}{\s'}][\irr{\D},\D'',\xl]{\rho\Rrightarrow e''}{ \vp}[][\irr{\D}][z]}$\`by inv.\\
    (6) $e[e'/x] = e$\` by $x$ does not occur in $e$\\
    (7) $\ov{\judg[\subst{\G}{\D'}{x},\zr{\irr{\D}}{\s'}][\irr{\D},\D'',\D']{\rho\Rrightarrow e''[e'/x]}{\vp}[][\irr{\D}][z]}$\`by i.h.(2)\\
    (8) $\judg[\subst{\G}{\D'}{x}][\D,\D'',\D']{\ccase{e}{\z[\irr{\D}][\s']~\{\ov{\rho\Rightarrow e''[e'/x]}\}}}{\vp}$\\
\end{tabbing}

\end{description}

Statement (2):

\begin{description}
\item[Case:] $AltN_{WHNF}$
\begin{tabbing}
    (1) $\judg[\G][\D']{e'}{\s}$\\
    (2) $\judg[\G][\D,\xl]{\konstructor~\Mapsto e}{\vp}[][\D_s][z][\dag]$\\
    (3) $\judg[\G,\ov{\xo},\ov{y_i{:}_{\D_i}\s_i}][\D,\xl]{e}{\vp}$\`by inv.\\
    (4) $\judg[\subst{\G}{\D'}{x},\ov{\xo},\ov{y_i{:}_{\subst{\D_i}{\D'}{x}}\s_i}][\D,\D']{e[e'/x]}{\vp}$\` by i.h.(1)\\
    (5) $\judg[\subst{\G}{\D'}{x}][\D,\D']{\rho \Mapsto e[e'/x]}{\vp}[][\subst{\D_s}{\D'}{x}][z][\dag]$\` by (4)\\
\end{tabbing}

\item[Case:] $AltN_{Not WHNF}$
\begin{tabbing}
    (1) $\judg[\G][\D']{e'}{\s}$\\
    (2) $\judg[\G][\D,\xl]{\konstructor~\Rrightarrow e}{\vp}[][\D_s][z][\ddag]$\\
(3) $\ov{\D_i}=\ov{\lctag{\D_s}{K_j}}^n$\`by inv.\\
(4) $\judg[\G,\ov{\xo},\ov{y_i{:}_{\D_i}\s_i}][\D,\xl]{e}{\vp}$\`by inv.\\
    (5) $\judg[\subst{\G}{\D'}{x},\ov{\xo},\ov{y_i{:}_{\subst{\D_i}{\D'}{x}}\s_i}][\D,\D']{e[e'/x]}{\vp}$\` by i.h.(1)\\
    (6) $\ov{\D_i[\D'/x]}=\ov{\lctag{\D_s[\D'/x]}{K_j}}^n$\`by (3) and cong.\\
    (7) $\judg[\subst{\G}{\D'}{x}][\D,\D']{\rho \Rrightarrow e[e'/x]}{ \vp}[][\subst{\D_s}{\D'}{x}][z][\ddag]$\` by (5,6)\\
\end{tabbing}

\item[Case:] $Alt0$
    This is one of the most interesting proof cases, and challenging to prove.
        \begin{itemize}
            \item The first insight is to divide the proof into two subcases, accounting
                for when the scrutinee (and hence $\D_s$) contains the linear resource and when it does not.
            \item The second insight is to recall that $\D$ and $\D'$ are
                disjoint to be able to prove the subcase in which $x$ does not
                occur in the scrutinee
            \item The third insight is to \emph{create} linear resources
                seemingly out of nowhere \emph{under a substitution that
                removes them}. We see this happen in the case where $x$ occurs
                in the scrutinee, for both the linear and affine contexts (see (5,6)).
                We must also see that we can swap $x$ for $\D'$ if neither can occur (see (7)).
        \end{itemize}
\begin{tabbing}
  (1) $\judg[\G][\D']{e'}{\s}$\\
  Subcase $x$ occurs in scrutinee\\
    (2) $\judg[\G][\D,\xl]{K~\ov{\xo}~\Rightarrow e}{\vp}[][\D_s,\xl][z]$\\
    (2.5) $\judg[\subst{\G}{\cdot}{\D_s,x}_z,\ov{\xo}][\subst{(\D,\xl)}{\cdot}{\D_s,x}]{e}{\vp}$\`by inv.\\
    (3) $\judg[\subst{\G}{\cdot}{\D_s,x}_z,\ov{\xo}][\subst{\D}{\cdot}{\D_s}]{e}{\vp}$\\
    (4) $e[e'/x] = e$\` since $x$ cannot occur in $e$ (erased from cx)\\
    (5) $\D[\cdot/\D_s] = (\D,\D')[\cdot/\D_s,\D']$   \` by cong. of subst.\\
    (6) $\G[\cdot/\D_s,x]_z[\D'/x] = \G[\D'/x][\cdot/\D_s,\D']_z$\` by cong. of subst.\\(7) $\forall x,\D,\D',\G:x\notin\Delta \land \Delta'\not\subset\D \land \G;\D \vdash e :\s \Rightarrow \G[\D'/x];\D \vdash e : \s$\` by Weaken\\
    \` and variables in $\G$ cannot occur in $e$ if they mention $x$ nor if they mention $\D'$ \\
(8) $\judg[\subst{\subst{\G}{\D'}{x}}{\cdot}{\D_s,\D'}_z,\ov{\xo}][\subst{(\D,\D')}{\cdot}{\D_s,\D'}]{e[e'x]}{\vp}$\` by (4,5,6,7)\\\` and $x$ and $\D'$ are erased from ctx\\
(9) $\judg[\subst{\G}{\D'}{x}][\D,\D']{K~\ov{\xo}~\Rightarrow e[e'/x]}{ \vp}[alt][\D_s,\D'][z]$\` by $Alt0$\\
  Subcase $x$ does not occur in scrutinee\\
    (2) $\judg[\G][\D,\xl]{K~\ov{\xo}~\Rightarrow e}{\vp}[][\D_s][z]$\\
    (3) $\judg[\subst{\G}{\cdot}{\D_s}_z,\ov{\xo}][\subst{\D}{\cdot}{\D_s},\xl]{e}{\vp}$\`by $x$ does not occur in $\D_s$ and inv.\\
    (4) $\judg[\subst{\subst{\G}{\D'}{x}}{\cdot}{\D_s}_z,\ov{\xo}][\subst{\D}{\cdot}{\D_s},\D']{e[e'/x]}{\vp}$\\\`by i.h.(1) and $x$ does not occur in $\D_s$\\
    (5) $\judg[\subst{\subst{\G}{\D'}{x}}{\cdot}{\D_s}_z,\ov{\xo}][\subst{(\D,\D')}{\cdot}{\D_s}]{e[e'/x]}{\vp}$\\
    \`by $\D$ and $\D'$ being disjoint by hypothesis,\\
    \`and $\D_s$ being a subset of $\D$\\ (6) $\D_s[\D'/x] = \D_s$\`by $x$ does not occur in $\D_s$\\
    (7) $\judg[\subst{\G}{\D'}{x}][\D,\D']{K~\ov{\xo}~\Rightarrow e[e'/x]}{\vp}[][\subst{\D_s}{\D'}{x}][z]$\\
\end{tabbing}

\end{description}

\end{proof}

The substitution lemma for unrestricted variables follows the usual
formulation, with the added restriction (common to linear type systems) that
the expression $e'$ that is going to substitute the unrestricted variable $x$
is typed on an empty linear environment:
\UnrestrictedSubstitutionLemma

\begin{proof}
By structural induction on the given derivation. We omit many trivial cases.

Statement (1):
\begin{description}

\item[Case:] $\lambda I_1$
\begin{tabbing}
  (1) $\G, \xo; \D \vdash \lambda\y[1][\s'].~e : \s'\to_1\vp$\\
  (2) $\G; \cdot \vdash e' : \sigma$\\
  (3) $\G, \xo; \D, \y[1][\s'] \vdash e : \vp$ \` by inversion on $\lambda I_1$\\
  (4) $\G; \D, \y[1][\s'] \vdash e[e'/x] : \vp$ \` by i.h.(1) (2,3)\\
  (5) $\G; \D \vdash \lambda\y[1][\s'].~e[e'/x] : \s' \to_1\vp$ \` by $\lambda I_1$\\
  (6) $(\lambda \y[\pi][\s'].~e)[e'/x] = (\lambda \y[\pi][\s'].~e[e'/x])$ \` by def. of subst.\\
\end{tabbing}

\item[Case:] $Var_\omega$
\begin{tabbing}
  (1) $\G, x{:}_\omega; \cdot \sigma \vdash x : \sigma$\\
  (2) $\G; \cdot \vdash e' : \s$\\
  (4) $x[e'/x] = e'$ \` by def. of substitution\\
  (5) $\G; \cdot \vdash e' : \sigma$ \` by (2)\\
\end{tabbing}

\item[Case:] $Var_\omega$
\begin{tabbing}
  (1) $\G, \xo; \cdot \vdash y : \vp$\\
  (2) $\G; \cdot \vdash e' : \s$\\
  (3) $y[e'/x] = y$ \` by def. of substitution\\
  (4) $\G; \cdot \vdash y : \vp$ \` by inversion on $Weaken_\omega$ (1)\\
\end{tabbing}

\item[Case:] $Var_1$
\begin{tabbing}
  (1) Impossible. The context in $Var_1$ is empty.
\end{tabbing}

\item[Case:] $Var_\Delta$
\begin{tabbing}
  (1) Impossible. The context in $Var_\Delta$ only contains linear variables.
\end{tabbing}

\item[Case:] $\lambda E_{1,p}$
\begin{tabbing}
  (1) $\G, \xo; \D,\D' \vdash e~e'' : \vp$\\
  (2) $\G; \cdot \vdash e' : \s$\\
(3) $\G, \xo; \D \vdash e : \s'\to_{1,p}\vp$ \` by inversion on $\lambda E_{1,p}$\\
  (4) $\G, \xo; \D' \vdash e'' : \s'$ \` by inversion on $\lambda E_{1,p}$\\
  (5) $\G; \D \vdash e[e'/x] : \s'\to_{1,p}\vp$ \` by i.h.(1) (2,3)\\
  (6) $\G; \D' \vdash e''[e'/x] : \s'$ \` by i.h.(1) (2,4)\\
  (7) $\G; \D, \D' \vdash e[e'/x]~e''[e'/x] : \vp$ \` by $\lambda E_{1,p}$ (5,6)\\
(8) $(e~e'')[e'/x] = (e[e'/x]~e''[e'/x])$ \` by def. of subst.\\
\end{tabbing}

\item[Case:] $Let$
\begin{tabbing}
    (1) $\G, \xo; \D, \D' \vdash \llet{\yD = e}{e''} : \vp$\\
    (2) $\G; \cdot \vdash e' : \sigma$\\
    (3) $\G,\xo,\yD; \D, \D' \vdash e'' \vp$ \` by inversion on $Let$\\
    (4) $\G,\xo; \D \vdash e : \s'$ \` by inversion on $Let$\\
    (5) $\G,\yD; \D \vdash e''[e'/x] : \vp$ \` by i.h.(1) (2,3)\\
    (6) $\G; \D \vdash e[e'/x] : \s'$ \` by i.h.(1) (2,4)\\
    (7) $\G;\D,\D' \vdash \llet{\yD = e[e'/x]}{e''[e'/x]}$ \` by $Let$ (5,6)\\
    (8) $(\llet{\yD = e}{e''})[e'/x] = (\llet{\yD = e[e'/x]}{e''[e'/x]})$
\end{tabbing}

\item[Case:] $LetRec$
\begin{tabbing}
    (1) $\G, \xo; \D, \D' \vdash \lletrec{\ov{\yD = e}}{e''} : \vp$\\
    (2) $\G'; \cdot \vdash e' : \sigma$\\
    (3) $\G,\xo,\ov{\yD}; \D, \D' \vdash e'' : \vp$ \` by inversion on $LetRec$\\
    (4) $\ov{\G,\xo, \ov{\yD}; \D \vdash e : \s'}$ \` by inversion on $LetRec$\\
    (5) $\G, \ov{\yD}; \D, \D' \vdash e''[e'/x] : \vp$ \` by i.h.(1) (2,3)\\
    (6) $\ov{\G,\ov{\yD}; \D \vdash e[e'/x] : \s'}$ \` by i.h.(1) (2,4)\\
    (7) $\G; \D, \D' \vdash \lletrec{\ov{\yD = e[e'/x]}}{e''[e'/x]} : \vp$ \` by $LetRec$ (5,6)\\
    (8) $(\lletrec{\ov{\yD = e}}{e''})[e'/x] = (\lletrec{\ov{\yD = e[e'/x]}}{e''[e'/x]})$
\end{tabbing}

\end{description}

Statement (2):

\begin{description}

\item[Case:] $AltN_{WHNF}$ (trivial induction)
\begin{tabbing}
    (1) $\G; \cdot \vdash e : \s$\\
    (2) $\G,\xo;\D \vdash_{alt} \konstructor~\to e \prescript{\dag}{}{:^z_{\ov{\D_i}^n}} \s' \Mapsto \vp$\\
    (3) $\G,\xo,\ov{\xo},\ov{y_i{:}_{\D_i}\s_i}; \D \vdash e : \vp $\\
    (4) $\G,\ov{\xo},\ov{y_i{:}_{\D_i}\s_i}; \D \vdash e[e'/x] : \vp $\` by i.h.(1)\\
    (5) $\G;\D \vdash_{alt} \konstructor~\to e[e'/x] \prescript{\dag}{}{:^z_{\ov{\D_i}^n}} \s' \Mapsto \vp$\` by $AltN$\\
\end{tabbing}

\item[Case:] $AltN_{Not WHNF}$ (trivial induction)
\begin{tabbing}
    (1) $\G; \cdot \vdash e : \s$\\
    (2) $\G,\xo;\D \vdash_{alt} \konstructor~\to e \prescript{\ddag}{}{:^z_{\D_s}} \s' \Rrightarrow \vp$\\
    (3) $\ov{\D_i} = \ov{\lctag{\D_s}{K_j}}^n$\\
    (4) $\G,\xo,\ov{\xo},\ov{y_i{:}_{\D_i}\s_i}; \D \vdash e : \vp $\` by inv.\\
    (5) $\G,\ov{\xo},\ov{y_i{:}_{\D_i}\s_i}; \D \vdash e[e'/x] : \vp $\` by i.h.(1)\\
    (6) $\G;\D \vdash_{alt} \konstructor~\to e[e'/x] \prescript{\ddag}{}{:^z_{\D_s}} \s' \Rrightarrow \vp$\` by $AltN$\\
\end{tabbing}

\end{description}
\end{proof}

Finally, we introduce the lemma stating that substitution of $\D$-bound
variables by expressions of the same type preserves the type of the original
expression.
What distinguishes this lemma from traditional substitution lemmas is that the
usage environment $\D$ of the variable $x$ being substituted by expression $e'$
must match exactly the typing environment $\D$ of $e'$ and the
environment of the original expression doesn't change with the substitution:
\DeltaSubstitutionLemma
\noindent Intuitively, if $x$ is well-typed with $\D$ in $e$, substituting $x$
by an expression $e'$ which is typed in the same environment $\D$ allows the
distribution of resources $\D,\D'$ used to type $e$ across sub-derivations to remain
unchanged. To prove the theorems, we don't need a ``stronger'' substitution of
$\D$-vars lemma (allowing arbitrary resources $\D''$ to type $e'$, as in other
substitution lemmas), as we only ever substitute $\D$-variables by expressions
whose typing environment matches the variables usage environment. However, it
is not obvious whether such a lemma is possible to prove for $\D$-variables
(e.g. let $\G;\D \vdash e :\s$ and $\G; \D' \vdash \llet{x = e'}{x}$, if we
substitute $e$ for $x$ the resources $\D'$ are no longer consumed).

\begin{proof}
By structural induction on the given derivation. We omit many trivial cases.

Statement (1):
\begin{description}

\item[Case:] $\Lambda I$
\begin{tabbing}
    (1) $\G, \xD; \D, \D' \vdash \Lambda p.~e : \forall p.~\vp$\\
    (2) $\G; \D \vdash e' : \s$\\
    (3) $\G, p, \xD; \D, \D' \vdash e : \vp$ \` by inversion on $\Lambda I$\\
    (4) $\G, p; \D, \D' \vdash e[e'/x]$ \` by i.h.(1) (2,3)\\
    (5) $\G; \D, \D' \vdash \Lambda p.~e[e'/x] : \forall p.~\vp$ \` by $\Lambda I$ \\
    (6) $(\Lambda p.~e)[e'/x] = (\Lambda p.~e[e'/x])$ \` by def. of subst.\\
    (7) $\G; \D, \D' \vdash (\Lambda p.~e)[e'/x] : \forall p.~\vp$ \` by (5,6)\\
\end{tabbing}

\item[Case:] $\lambda I_1$
\begin{tabbing}
    (1) $\G, \xD; \D, \D' \vdash \lambda \y[1][\s'].~e : \s' \to_1 \vp$\\
    (2) $\G; \D \vdash e' : \s$\\
    (3) $\G, \xD; \D, \y[1][\s'], \D' \vdash e : \vp$ \` by inversion on $\lambda I$\\
    (4) $\G; \D, \y[1][\s'], \D' \vdash e[e'/x] : \vp$ \` by i.h.(1) (2,3)\\
    (5) $\G; \D, \D' \vdash \lambda \y[1][\s'].~e[e'/x] : \s' \to_1 \vp$ \` by $\lambda I$\\
    (6) $(\lambda \y[1][\s'].~e[e'/x]) = (\lambda \y[1][\s'].~e)[e'/x]$ \` by def. of subst.\\
    (7) $\G; \D, \D' \vdash \lambda (\y[1][\s'].~e)[e'/x] : \s' \to_1 \vp$ \` by (4,5)\\
\end{tabbing}

\item[Case:] $Var_\omega$
\begin{tabbing}
    (1) $\G, \yo, \x[\cdot][\s]; \cdot \vdash y : \s'$\\
    (2) $\G; \cdot \vdash e' : \s$\\
    (3) $y[e'/x] = y$\\
    (4) $\G,\yo;\cdot \vdash y : \s'$\` by (1) and $Weaken_\D$\\
\end{tabbing}

\item[Case:] $Var_1$
\begin{tabbing}
  (1) $\G,\x[\y]; \y \vdash y : \s$\\
  (2) $\G; \y \vdash e' : \s$\\
  (3) $y[e'/x] = y$\\
  (4) $\G,\x[\y]; \y \vdash y : \s$\`by 1\\ 
  (5) $\G; \y\vdash y : \s$\` by $Weaken_\Delta$\\
\end{tabbing}

\item[Case:] $Var_\Delta$
\begin{tabbing}
    (1) $\G, \xD; \D \vdash y : \s$\\
    (2) $\G'; \D \vdash e' : \s$\\
(3) $x[e'/x] = e'$\\
    (4) $\G'; \D \vdash e' : \s$\` by (2)\\
\end{tabbing}

\item[Case:] $\lambda E_{1,p}$
\begin{tabbing}
    (1) $\G, \xD; \D,\D',\D'' \vdash e~e'' : \varphi$\\
    (2) $\G; \D \vdash e' : \sigma$\\
    Subcase $\D$ occurs in $e$\\
    (3) $\G, \xD; \D,\D' \vdash e : \s' \to_{1,p} \vp$\\
    (4) $\G, \xD; \D'' \vdash e'' : \s'$\\
    (5) $\G; \D'' \vdash e'' : \s'$\` by $Weaken_\D$\\
    (6) $\G; \D,\D' \vdash e[e'/x] : \s' \to_{1,p} \vp$ \` by i.h.(1)\\
    (7) $\G;\D, \D', \D'' \vdash e[e'/x]~e'' : \vp$\` by ($\lambda E_{1,p}$)\\
    (8) $(e[e'/x]~e'')=(e~e'')[e'/x]$\` since $x$ cannot occur in $e''$\\
    Subcase $\D$ occurs in $e''$\\
    (3) $\G, \xD; \D' \vdash e : \s' \to_{1,p} \vp$\\
    (4) $\G; \D' \vdash e : \s' \to_{1,p} \vp$\` by $Weaken_\Delta$\\
    (5) $\G,\xD; \D, \D'' \vdash e'' : \s'$\\
    (6) $\G; \D, \D'' \vdash e''[e'/x] : \s'$\` by i.h.(1)\\
    (7) $\G; \D,\D',\D'' \vdash (e~e''[e'/x]) : \vp$\` by ($\lambda E_{1,p}$)\\
    (8) $e~e''[e'/x] = (e~e'')[e'/x]$\`since $x$ does not occur in $e$\\
    Subcase $\D$ is split between $e$ and $e''$\\
    $x$ cannot occur in either, so the substitution is trivial, and $x$ can be weakened.

\end{tabbing}

\item[Case:] $\lambda E_\omega$
\begin{tabbing}
    (1) $\G, \xD; \D, \D' \vdash e~e'' : \vp$\\
    (2) $\G; \D \vdash e' : \s$\\
    (3) $\D$ cannot occur in $e''$\\
    (4) $\G, \xD; \D, \D' \vdash e : \s' \to_\omega \vp$\` by inv. on $\lambda E_\omega$\\
    (5) $\G; \cdot \vdash e'' : \s'$ \` by inv. on $\lambda E_\omega$\\
    (6) $\G; \D, \D' \vdash e[e'/x] : \s' \to_\omega \vp$ \` by i.h.(1) (2,4)\\
    (7) $\G; \D, \D' \vdash e[e'/x]~e'' : \vp$\` by $\lambda E_\omega$ (5,6)\\
    (8) $e[e'/x]~e'' = (e~e'')[e'/x]$\` $x$ does not occur in $e''$ by (3)\\
\end{tabbing}

\item[Case:] $Let$
\begin{tabbing}
    (1) $\G; \D \vdash e' : \s$\\
    Subcase $\D$ occurs in $e$\\
    (2) $\G, \xD; \D, \D', \D'' \vdash \llet{\y[\D,\D'][\s'] = e}{e''} : \vp$\\
    (3) $\G, \xD; \D, \D' \vdash e : \s'$\` by inversion on (let)\\
    (4) $\G, \xD,\y[\D,\D'][\s']; \D,\D', \D'' \vdash e'' : \vp$\` by inversion on (let)\\
    (5) $\G,\y[\D,\D'][\s']; \D,\D', \D'' \vdash e'' : \vp$\` by $Weaken_\Delta$ (admissible)\\
    (6) $\G; \D,\D' \vdash e[e'/x] : \s'$\` by i.h.(1)  (1,3)\\
    (7) $\G; \D,\D',\D'' \vdash \llet{\y[\D,\D'][\s'] = e[e'/x]}{e''} : \vp$\` by (let) (5,6)\\
    (8) $\llet{\y[\D,\D'][\s'] = e[e'/x]}{e''} = (\llet{\y[\D,\D'][\s'] = e}{e''})[e'/x]$\\\` since $x$ cannot occur in $e''$\\
    Subcase $\D$ occurs in $e''$\\
    (2) $\G, \xD; \D, \D', \D'' \vdash \llet{\y[\D'][\s'] = e}{e''} : \vp$\\
    (3) $\G, \xD; \D' \vdash e : \s'$\` by inversion on (let)\\
    (4) $\G; \D' \vdash e : \s'$\` by $Weaken_\Delta$\\
    (5) $\G,\xD,\y[\D'][\s']; \D,\D',\D'' \vdash e'' : \vp$\` by inversion on (let)\\
    (6) $\G,\y[\D'][\s']; \D,\D',\D'' \vdash e''[e'/x] : \vp$\` by i.h.(1) (1,5)\\
    (7) $\G;\D,\D',\D'' \vdash \llet{\y[\D'][\s'] = e}{e''[e'/x]} : \vp$\` by (let)\\
    (8) $\llet{\y[\D'][\s'] = e}{e''[e'/x]} = (\llet{\y[\D'][\s'] = e}{e''})[e'/x]$\\\` since $x$ cannot occur in $e$\\
    Subcase $\D$ is split between $e$ and $e''$\\
    $x$ cannot occur in either, so the substitution is trivial, and $x$ can be weakened.
\end{tabbing}

\item[Case:] LetRec
\begin{tabbing}
    (1) $\G; \D \vdash e' : \s$\\
    Subcase $\D$ occurs in $\ov{e_i}$\\
    (2) $\G,\xD; \D, \D', \D'' \vdash \lletrec{\ov{\var[y_i][\D,\D'][\s_i'] = e_i}}{e''} : \vp$\\
    (3) $\G,\xD,\ov{\var[y_i][\D,\D'][\s'_i]}; \D,\D',\D'' \vdash e'' : \vp$\` by inversion on (let)\\
    (4) $\G,\ov{\var[y_i][\D,\D'][\s'_i]}; \D,\D',\D'' \vdash e'' : \vp$\` by $Weaken_\Delta$\\
    (5) $\ov{\G,\xD,\ov{\var[y_i][\D,\D'][\s'_i]}; \D, \D' \vdash e_i : \s'_i}$\` by inversion on (let)\\
    (6) $\ov{\G,\ov{\var[y_i][\D,\D'][\s'_i]}; \D,\D' \vdash e_i[e'/x] : \s'_i}$\` by i.h.(1) (1,5)\\
    (7) $e''[e'/x] = e''$\` since $x$ cannot occur in $e''$\\
    (8) $\G; \D,\D',\D'' \vdash \lletrec{\ov{\var[y_i][\D,\D'][\s_i'] = e_i[e'/x]}}{e''} : \vp$\` by (let) (4,6)\\
    Subcase $\D$ occurs in $e''$\\
    (2) $\G, \xD; \D, \D', \D'' \vdash \lletrec{\ov{\var[y_i][\D'][\s_i'] = e_i}}{e''} : \vp$\\
    (3) $\ov{\G,\xD,\ov{\var[y_i][\D'][\s'_i]}; \D' \vdash e_i : \s'_i}$\` by inversion on (let)\\
    (4) $\ov{\G,\ov{\var[y_i][\D'][\s'_i]}; \D' \vdash e_i : \s'_i}$\` by $Weaken_\Delta$\\
    (6) $\G,\xD,\ov{\var[y_i][\D'][\s'_i]}; \D,\D',\D'' \vdash e'' : \vp$\` by inversion on (let)\\
    (7) $\G,\ov{\var[y_i][\D'][\s'_i]}; \D,\D',\D'' \vdash e''[e'/x] :
    \vp$\` by i.h.(1) (1,6)\\
    (8) $\ov{e_i[e'/x] = e_i}$\` since $x$ cannot occur in $\ov{e_i}$\\
    (9) $\G;\D,\D',\D'' \vdash \lletrec{\ov{\var[y_i][\D'][\s'_i] = e_i}}{e''[e'/x]} : \vp$\` by (let)\\
    Subcase $\D$ is split between $e$ and $e''$\\
    $x$ cannot occur in either, so the substitution is trivial, and $x$ can be weakened.
\end{tabbing}

\item[Case:] CaseWHNF
\begin{tabbing}
    (1) $\G; \D \vdash e' : \s$\\
    Subcase $\D$ occurs in $e$\\
    (2) $\G,\xD;\D,\D',\D'' \vdash \ccase{e}{\z[\D,\D'][\s']~\{\ov{\rho \Rightarrow e''}\}} : \vp$\\
    (3) $e$ is in WHNF\\
    (4) $\rho_j$ matches $e$\\
    (5) $\G,\xD;\D,\D' \vdash e : \s'$\\
    (6) $\G,\xD,\z[\D,\D'][\s']; \D,\D',\D'' \vdash \rho_j \Mapsto e'' :^z_{\D,\D'} : \vp$\\
    (7) $\ov{\G,\xD,\z[\irr{\D,\D'}][\s']; [\D,\D'],\D'' \vdash\rho \Rrightarrow e'' :^z_{\irr{\D,\D'}} \vp}$\` by inv.\\
    (8) $\G;\D,\D' \vdash e[e'/x] : \s'$\` by i.h.(1)\\
    (9) $\G,\z[\D,\D'][\s']; \D,\D',\D'' \vdash\rho_j \Rightarrow e''[e'/x] :^z_{\D,\D'} \vp $\` by i.h.(2)\\
    (10) $\ov{\G,\z[\irr{\D,\D'}][\s']; [\D,\D'],\D'' \vdash \rho \Rrightarrow e'' :^z_{\irr{\D,\D'}} \vp}$\` by (7) and proof steps below\\\`from case $CaseNotWHNF$\\
    (11) $\G;\D,\D',\D'' \vdash \ccase{e[e'/x]}{\z[\D,\D'][\s']~\{\ov{\rho \Rightarrow e''[e'/x]}\}} : \vp$\` by $CaseWHNF$\\
    Subcase $\D$ does not occur in $e$\\
    (2) $\G,\xD;\D,\D',\D'' \vdash \ccase{e}{\z[\D'][\s']~\{\ov{\rho \Rightarrow e''}\}} : \vp$\\
    (3) $e$ is in WHNF\\
    (4) $\rho_j$ matches $e$\\
    (5) $\G,\xD;\D' \vdash e : \s'$\\
    (6) $\G;\D' \vdash e : \s'$\` by (admissible) $Weaken_\D$\\
    (7) $e[e'/x] = e$\` by $x$ cannot occur in $e$\\
    (8) $\G,\xD,\z[\D'][\s']; \D,\D',\D'' \vdash \rho_j \Mapsto e'' :^z_{\D'} : \vp$\\
    (9) $\ov{\G,\xD,\z[\irr{\D'}][\s']; \D,[\D'],\D'' \vdash_{alt} \rho \Rrightarrow e'' :^z_{\irr{\D'}} \vp}$\` by inv.\\
    (10) $\G,\z[\D'][\s']; \D,\D',\D'' \vdash_{alt} \rho_j \Mapsto e''[e'/x] :^z_{\D'} \vp$\` by i.h.(2)\\
    (11) $\ov{\G,\z[\irr{\D'}][\s']; \D,[\D'],\D'' \vdash_{alt} \rho \Rrightarrow e''[e'/x] :^z_{\irr{\D'}} \vp}$\\\` by i.h.(2)\\
    (12) $\G;\D,\D',\D'' \vdash \ccase{e[e'/x]}{\z[\D'][\s']~\{\ov{\rho \Rightarrow e''[e'/x]}\}} : \vp$\` by $CaseWHNF$\\
    Subcase $\D$ is partially used in $e$\\
    This is like the subcase above, but consider $\D'$\\
    to contain some of part of $\D$ and $\D$ to refer to the other part only.
\end{tabbing}

\item[Case:] CaseNotWHNF
\begin{tabbing}
    (1) $\G;\D\vdash e' :\s$\\
    Subcase $\D$ occurs in $e$\\
    (2) $\G,\xD; \D, \D', \D'' \vdash \ccase{e}{\z[\irr{\D,\D'}][\s']~\{\ov{\rho\Rightarrow e''}\}}$\\
    (3) $\G,\xD; \D, \D' \vdash e : \s'$\` by inv.\\
    (4) $\G; \D,\D' \vdash e[e'/x] : \s'$\` by i.h.(1)\\
    (5) $\ov{\G,\xD,\z[\irr{\D,\D'}][\s']; [\D,\D'],\D'' \vdash \rho \Rrightarrow e'' :^z_{\irr{\D,\D'}} \vp}$\` by inv.\\
    (6) $e''[e'/x] = e$\` by $x$ cannot occur in $\ov{e''}$ since $\D$ is not available (only $\irr{\D}$)\\
    (7) $\ov{\G,\z[\irr{\D,\D'}][\s']; [\D,\D'],\D'' \vdash \rho \Rrightarrow e'' :^z_{\irr{\D,\D'}} \vp}$\` by (admissible) $Weaken_\Delta$\\
    (8) $\G; \D, \D', \D'' \vdash \ccase{e[e'/x]}{\z[\irr{\D,\D'}][\s']~\{\ov{\rho\Rightarrow e''}\}}$\`by $CaseNotWHNF$\\
    Subcase $\D$ does not occur in $e$\\
    (2) $\G,\xD; \D, \D', \D'' \vdash \ccase{e}{\z[\irr{\D'}][\s']~\{\ov{\rho\Rightarrow e''}\}}$\\
    (3) $\G,\xD; \D' \vdash e : \s'$\` by inv.\\
    (4) $\G; \D' \vdash e : \s'$\` by weaken\\
    (5) $e[e'/x] = e$\` by x cannot occur in $e$\\
    (5) $\ov{\G,\xD,\z[\irr{\D'}][\s']; \D,[\D'],\D'' \vdash\rho \Rrightarrow e'' :^z_{\irr{\D'}} \vp}$\` by inv.\\
    (7) $\ov{\G,\z[\irr{\D'}][\s']; \D,[\D'],\D'' \vdash_{alt} \rho \Rrightarrow e''[e'/x] :^z_{\irr{\D'}} \vp}$\\\` by i.h.(2)\\
    (8) $\G; \D, \D', \D'' \vdash \ccase{e}{\z[\irr{\D,\D'}][\s']~\{\ov{\rho\Rightarrow e''[e'/x]}\}}$\`by $CaseNotWHNF$\\
    Subcase $\D$ is partially used in $e$\\
    This is like the subcase above, but consider $\D'$\\
    to contain some of part of $\D$ and $\D$ to refer to the other part only.
\end{tabbing}

\end{description}

Statement (2):
\begin{description}

\item[Case:] $Alt0$
\begin{tabbing}
(1) $\G; \D \vdash e' : \s$\\
    Subcase $\D$ occurs in scrutinee\\
    (2) $\G,\xD; \D,\D',\D'' \vdash K~\ov{\x}~\Rightarrow e :^z_{\D,\D'} \vp$\\
    (3) $(\G,\xD)[\cdot/\D,\D']_z; (\D,\D',\D'')[\cdot/\D,\D'] \vdash e : \vp$\\
    (4) $\G[\cdot/\D,\D']_z, \xD; \D'' \vdash e : \vp$\\\`by def. of $[]_z$ subst. and $[]$ subst.\\
    (5) $\G[\cdot/\D,\D']_z, \xD; \D'' \vdash e[e'/x] : \vp$\\\`by $x$ cannot occur in $e$ by $\D$ not available\\
    (6) $\G[\cdot/\D,\D']_z; \D'' \vdash e[e'/x] : \vp$\\\`by (admissible) $Weaken_\D$\\
    (7) $\cdot = (\D,\D')[\cdot/\D,\D']$\\
    (8) $\G[\cdot/\D,\D']_z; (\D,\D',\D'')[\cdot/\D,\D'] \vdash e[e'/x] : \vp$\\\`by (6,7)\\
    (9) $\G;\D,\D',\D'' \vdash e[e'/x] :^z_{\D,\D'} \vp$\\
    Subcase $\D$ does not occur in the scrutinee\\
    (2) $\G,\xD; \D,\D',\D'' \vdash K~\ov{\x}~\Rightarrow e :^z_{\D'} \vp$\\
    (3) $(\G,\xD)[\cdot/\D']_z; (\D,\D',\D'')[\cdot/\D'] \vdash e : \vp$\\
    (4) $\G[\cdot/\D']_z, \xD; \D,\D'' \vdash e : \vp$\\\`by def. of $[]_z$ subst. and $[]$ subst.\\
    (5) $\G[\cdot/\D']_z; \D,\D'' \vdash e[e'/x] : \vp$\` by i.h.(1)\\
    (6) $\cdot = \D'[\cdot/\D']$\\
    (7) $\G[\cdot/\D']_z; (\D,\D',\D'')[\cdot/\D'] \vdash e[e'/x] : \vp$\` by (5,6)\\
    (8) $\G;\D,\D',\D'' \vdash \Rightarrow   e[e'/x] :^z_{\D'} \vp$\`by $Alt0$\\
    Subcase $\D$ is partially used in the scrutinee\\
    This is like the subcase above, but consider $\D'$\\
    to contain some of part of $\D$ and $\D$ to refer to the other part only.
\end{tabbing}

\end{description}

\end{proof}

\subsection{Type Preservation\label{sec:proof:type-preservation}}

\begin{theorem}[Type Preservation]\label{thm:typepres}
  Let $\G ; \D \vdash (\Theta \mid e) \Downarrow (\Theta' \mid v) :
  \tau , \Sigma$ or $\G ; \D \vdash (\Theta \mid e) \Downarrow^* (\Theta' \mid v) :
  \tau , \Sigma$.
  If $\G ; \D \vdash \Theta \mid e : \tau , \Sigma$ then $\G ; \D \vdash \Theta' \mid v : \tau , \Sigma$
  
\end{theorem}

\begin{proof}
By structural induction on the instrumented operational semantics. We
omit many trivial cases.
\begin{description}
    \item[Case:] $\Gamma; \Delta \vdash (\Theta, \xl = y \mid x) \Downarrow (\Theta' \mid v) : \s, \Sigma$
\begin{tabbing}
    (1) $\G; \D \vdash \Theta, \xl = y \mid x : \s, \Sigma$\\
    (2) $\G, \Theta{\uparrow}^\omega; \D, \Theta{\uparrow}^1, \xl \vdash (x, y, \Theta{\downarrow}^1_e, \Sigma) : \s \otimes \s \bigotimes(\Theta{\downarrow^1_e}) \bigotimes(\Sigma)$\`by \ref{lem:unwrapenv}\\
    (3) $Let~\G';\D_1,\D_2,\D_3 = \G, \Theta{\uparrow}^\omega; \D, \Theta{\uparrow}^1$\` s.t. $(\D_2,\D_3)$ types $(\Theta{\downarrow}^1_e,\Sigma)$\\
    (4) $\G'; \xl \vdash x : \s$\\
    (5) $\G'; \D_1 \vdash y : \s$\\
    (6) $\G'; \D_1, \D_2, \D_3 \vdash (y, \Theta{\downarrow}^1_e, \Sigma) : \s \bigotimes(\Theta{\downarrow^1_e}) \bigotimes(\Sigma)$\\
    (7) $\G; \D \vdash \Theta \mid y : \s, \Sigma$\`by well-typed state\\
    (8) $\G; \D \vdash \Theta' \mid v : \s, \Sigma$\` by induction hypothesis\\
\end{tabbing}
    \item[Case:] $\Gamma; \Delta \vdash (\Theta, x{:}_{\D'}\s = e \mid x) \Downarrow (\Theta', x{:}_{\D_a}\s = v \mid v) : \s, \Sigma$
\begin{tabbing}
    (1) $\Gamma; \Delta \vdash (\Theta, x{:}_{\D'}\s = e \mid x) : \s, \Sigma$\\
    (2) $\G, \Theta{\uparrow}^\omega; \D, \Theta{\uparrow}^1 \vdash \llet{x{:}_{\D'} = e}{(x, \Theta{\downarrow}^1_e, \Sigma)} : \s \bigotimes(\Theta{\downarrow}^1_\s) \bigotimes(\Sigma)$\`by \ref{lem:unwrapenv}\\
    (3) $Let~\D',\D'' = \D,\Theta{\uparrow}^1$\\
    (4) $\G, \Theta{\uparrow}^\omega; \D' \vdash e : \s$\`by inv on let (2)\\
    (5) $\G, \Theta{\uparrow}^\omega; \D', \D'' \vdash (x, \Theta{\downarrow}^1_e, \Sigma) : \s \bigotimes(\Theta{\downarrow}^1_\s) \bigotimes(\Sigma)$\`by inv on let (2)\\
    (6) $\G, \Theta{\uparrow}^\omega; \D' \vdash x : \s$\\
    (7) $\G, \Theta{\uparrow}^\omega; \D'' \vdash (\Theta{\downarrow}^1_e, \Sigma) : \bigotimes(\Theta{\downarrow}^1_\s) \bigotimes(\Sigma)$\\
    (8) $\G, \Theta{\uparrow}^\omega; \D', \D'' \vdash (e, \Theta{\downarrow}^1_e, \Sigma) : \tau \bigotimes(\Theta{\downarrow}^1_\s) \bigotimes(\Sigma)$\` by pair (4,7)\\
    (9) $\G; \D \vdash \Theta \mid e : \tau,\Sigma$\` by well-typed state\\
    (10) $\G; \D \vdash \Theta \mid v : \tau,\Sigma$\` by induction hypothesis\\
    (11) $Let~\G';\D_a,\D_b = \G, \Theta'{\uparrow}^\omega, \D, \Theta'{\uparrow}^1$\\
    (11a) $\G'; \D_a \vdash v : \s$
    (11b) $\G'; \D_b \vdash (\Theta'{\downarrow}^1_e, \Sigma) : \bigotimes(\Theta'{\downarrow}^1_\s) \bigotimes(\Sigma)$\\
    (12) $\G'; \D_a \vdash \llet{x{:}_{\D_a}\s = v}{v} : \s$\`by let (11a,11a)\\
    (13) $\G; \D \vdash \Theta' \mid \llet{x{:}_{\D_a}\s = v}{v} : \s, \Sigma$\`by well-typed state\\
    (14) $\G; \D \vdash \Theta', x{:}_{\D_a}\s = v \mid v : \s, \Sigma$\`by env. expansion\\
\end{tabbing}
Note: It's slightly more elaborate to accommodate recursion. The runtime
environment bindings must be tagged with the recursive group such that the
environment expansion can decide whether to unwrap to a letrec group or
not. Unlike the case proved here, when the $\D$-binding belongs to a let group,
the whole group is unwrapped and we can invert the assumption to get an
$e$ typed with the $x$s from the group.\\
\item[Case:] $\Gamma; \Delta \vdash (\Theta \mid \ccase{e}{z{:}_{\Delta'}\s~\{K~\ov{y_i}\Rightarrow e'}\}) \Downarrow (\Theta'' \mid v) : \tau$
\begin{tabbing}
    (1) $\Gamma; \Delta \vdash \Theta \mid \ccase{e}{z{:}_{\Delta_a}\s~\{K~\ov{y_i}\Rightarrow e'}\} : \tau, \Sigma$\\
    (2) $\G, \Theta{\uparrow}^\omega; \D, \Theta{\uparrow}^1 \vdash (\ccase{e}{z{:}_{\Delta_a}\s~\{K~\ov{y_i}\Rightarrow e'}\}, \Theta{\downarrow}^1_e, \Sigma) : \tau \otimes \bigotimes(\Theta{\downarrow}^1_\sigma) \otimes \bigotimes(\Sigma)$\`by \ref{lem:unwrapenv}\\
    (3) $Let~\G' = \G, \Theta\uparrow^\omega; \D_a,\D_b,\D_2,\D_3 = \D,\Theta\uparrow^1$\`s.t. $(\D_2,\D_3)$ types $(\Theta\downarrow^1_e, \Sigma)$\\
    (4) $\G'; \D_a, \D_b \vdash \ccase{e}{z{:}_{\D_a}\s\{K~\ov{y_i}\Rightarrow e'\}} : \tau$\\
    Subcase $e$ is in WHNF (thus $= K~\ov{x_i}$)\\
    (5) $\G'; \D_a \Vdash K~\ov{x_i} : \s \gtrdot \ov{\D_i}$\` by subcase\\
    (6) $\G', z{:}_{\D_a}\s; \D_a, \D_b \vdash_{alt} K~\ov{y_i} \to e' :^z_{\D_i} \s \Rightarrow \tau$\` by inv\\
    Subsubcase $\vdash$ $Alt_{WHNF}$\\
    (7) $\G', z{:}_{\D_a}\s, \ov{y_i{:}_\omega\s_i}, \ov{y_j{:}_{\Delta_j}\s_j}; \D_a, \D_b \vdash e' : \tau$\` by inv\\
    (8) $\Theta' = \Theta$\` in this subcase because $(\Theta \mid K~\ov{x_i})\Downarrow(\Theta \mid K~\ov{x_i})$\\
    (9) $\G'; \ov{\D_i} \vdash K~\ov{x_i} : \s$\` by inv (5)\\
    (10) $\ov{\G'; \cdot \vdash x_i : \s_i}$\` by inv (9) for unrestricted $x$s\\
    (11) $\ov{\G'; \D_j \vdash x_j : \s_j}$\` by inv (9) for linear fields $x$s\\
    (12) $\G', z{:}_{\D_a}\s; \D_a, \D_b \vdash e'\ov{[x_{i,j}/y_{i,j}] : \tau}$\` by $\D$-subst and $\omega$-subst lemmas (7,10,11)\\
    (13) $\G'; \D_a, \D_b \vdash e'\ov{[x_{i,j}/y_{i,j}]}[K~\ov{x_{i,j}}/z] : \tau$\` by $\D$-subst lemma and $\ov{\D_i} = \D_a$ (9,12)\\
    (14) $\G; \D \vdash \Theta \mid e'\ov{[x_{i,j}/y_{i,j}]}[K~\ov{x_{i,j}}/z] : \tau, \Sigma$\` by well-typed state\\
    (15) $\G; \D \vdash \Theta'' \mid v : \tau, \Sigma$\` by induction hypothesis\\
    Subsubcase $\vdash$ $Alt_{0}$\\
    (7) $(\G', z{:}_{\D_a}\s)[\cdot/\D_a]_z, \ov{y_i{:}_\omega\sigma_i}; (\D_a, \D_b)[\cdot / \D_a] \vdash e' : \tau$\` by inv subcase\\
    (8) $\G', z{:}_{\cdot}, \ov{y_i{:}_\omega\sigma_i}; \D_b \vdash e' : \tau$\` by $z$ cannot occur in $\G'$\\
    (9) Since $\G';\D_a \Vdash K~\ov{x_i}$, in this subcase $\D_a = \cdot$\\
    (10) $\G', \z{:}_\cdot\sigma, \ov{y{:}_\omega\sigma}; \D_b \vdash e' : \tau$\` by (9)\\
    (11) $\G', \z{:}_\cdot\sigma; \D_b \vdash e'\ov{[x_i/y_i]} : \tau$\` by $\omega$-subst (10)\\
    (12) $\G'; \D_b \vdash e'\ov{[x_i/y_i]}[K~\ov{x_i}/z] : \tau$\` by $\D$-subst\\
    (13) $\G; \D \vdash \Theta \mid e'\ov{[x_i/y_i]}[K~\ov{x_i}/z] : \tau, \Sigma$\` by well-typed state\\
    (14) $\G; \D \vdash \Theta'' \mid v : \tau, \Sigma$\`by induction hypothesis\\
    Subcase $e$ is not in WHNF\\
    (5) $\G'; \D_a \vdash e : \s$\` by inv on case\\
    (6) $\G', z{:}_{\irr{\D_a}}\s; \irr{\D_a}, \D_b \vdash_{alt} K~\ov{y_{i,j}} \to e' :^z_{\irr{\D_a}} \s $\` by inv on case\\
    Subsubcase $\vdash$ $Alt_{NWHNF}$\\
    (7) $\G', z{:}_{\irr{\D_a}}\s, \ov{y_i{:}_\omega\s_i}, \ov{y_j{:}_{\Delta_j}\s_j}; \irr{\D_a}, \D_b \vdash e' : \tau$\` for \ov{\D_j = \lctag{\irr{\D_a}}{K_j}}, by inv.\\
    (8) $\G', z{:}_{\ov{y_j{:}_1\s_j}}\s, \ov{y_i{:}_\omega\s_i}; \ov{y_j{:}_1\s_j}, \D_b \vdash e' : \tau$\` by $\D$-bound to linear\\
    (9) $\G', z{:}_{\ov{y_j{:}_1\s_j}}\s, \ov{y_i{:}_\omega\s_i}; \D_a, \D_b, \ov{y_j{:}_1\s_j} \vdash (e, e') : \sigma \otimes \tau$\` by pair (5,8)\\
    (10) $\G, z{:}_{\ov{y_j{:}_1\s_j}}\s, \ov{y_i{:}_\omega\s_i}; \D, \ov{y_j{:}_1\s_j} \vdash \Theta \mid e : \s, e':\tau, \Sigma$\`by well-typed state\\
    (11) $\G, z{:}_{\ov{y_j{:}_1\s_j}}\s, \ov{y_i{:}_\omega\s_i}; \D, \ov{y_j{:}_1\s_j} \vdash \Theta' \mid K~\ov{x_i} : \s, e':\tau, \Sigma$\`by induction hypothesis\\
    (12) $\G, \Theta'{\uparrow}^\omega, \ov{y_i{:}_\omega\s_i}; \D, \ov{y_j{:}_1\s_j}, \Theta'{\uparrow}^1 \vdash (K~\ov{x_{i,j}}, \Theta'{\downarrow}^1_e, e', \Sigma) : \s \bigotimes(\Theta'\downarrow^1_\sigma) \otimes \tau \bigotimes(\Sigma)$\`by \ref{lem:unwrapenv}\\
    (13) $Let~\G''=\G,\Theta'{\uparrow}^\omega~and~\D_c\D_d\D_4\D_5=\D,\Theta'{\uparrow}^1$\` where $(\D_4, \D_5)$ types $(\Theta'{\downarrow^1_e}, \Sigma)$\\
    (13a) $\G''; \D_c \vdash K~\ov{x_{i,j}} : \s$\\
    (13b) $\G'', z{:}_{\ov{y_j{:}_1\s_j}}\s, \ov{y_i{:}_\omega\s_i}; \D_d, \ov{y_j{:}_1\s_j} \vdash e' : \tau$\\
    (14) $\ov{\G''; \cdot \vdash x_i : \s_i}$\\
    (15) $\ov{\G''; \D_j \vdash x_j : \s_j}$\` for $\ov{D_j}=\D_c$, by inv\\
    (16) $\G'', z{:}_{\D_c}\s; \D_c, \D_d \vdash e'\ov{[x_{i,j}/y_{i,j}]} : \tau$\`by $1$-substs and $\omega$-substs lemmas\\
    (17) $\G''; \D_c, \D_d \vdash e'\ov{[x_{i,j}/y_{i,j}]}[K~\ov{x_{i,j}}/z] : \tau$\`by $\D$-subst lemma (13a,16)\\
    (18) $\G; \D \vdash \Theta' \mid e'\ov{[x_{i,j}/y_{i,j}]}[K~\ov{x_{i,j}}/z] : \tau, \Sigma$\`by well-typed state\\
    (19) $\G; \D \vdash \Theta'' \mid v : \tau, \Sigma$\`by induction hypothesis\\
    Subsubcase $\vdash$ $Alt_0$\\
    (7) $\G', z{:}_\cdot\s, \ov{y_i{:}_\omega\s_i}; \D_b \vdash e' : \tau$\` by inv\\
    (8) $\G', z{:}_\cdot\s, \ov{y_i{:}_\omega\s_i}; \D_a, \D_b \vdash (e, e') : \sigma \otimes \tau$\` by pair (5,7)\\
    (9) $\G, z{:}_\cdot\s, \ov{y_i{:}_\omega\s_i}; \D \vdash \Theta \mid e : \sigma, e' : \tau, \Sigma$\` by well-typed state\\
(10) $\G, z{:}_\cdot\s, \ov{y_i{:}_\omega\s_i}; \D \vdash \Theta' \mid K~\ov{x_i} : \sigma, e' : \tau, \Sigma$\` by induction hypothesis\\
    (11) $Let~\G'' = \G,\Theta'{\uparrow}^\omega~and~\D_d,\D_4,\D_5 = \D, \Theta'{\uparrow}^1$\` where $(\D_4, \D_5)$ types $(\Theta'{\downarrow^1_e}, \Sigma)$\\
    (11a) $\G''; \cdot \vdash K~\ov{x_i} : \s$\`since $K$ has no linear fields in this subcase\\
    (11b) $\G'', z{:}_{\cdot}\s, \ov{y_i{:}_\omega\s_i}; \D_d \vdash e' : \tau$\\
    (12) $\ov{\G''; \cdot \vdash x_{i} : \s_i}$\`by inv (11a)\\
    (13) $\G'', z{:}_{\cdot}\s; \D_d \vdash e'\ov{[x_i/y_i]} : \tau$\` by $\omega$-subst lemma\\
    (14) $\G''; \D_d \vdash e'\ov{[x_i/y_i]}[K~\ov{x_i}/z] : \tau$\` by $\D$-subst lemma\\
    (15) $\G; \D \vdash \Theta' \mid e'\ov{[x_i/y_i]}[K~\ov{x_i}/z] :
    \tau, \Sigma$\` by well-typed state\\
    (16) $\G; \D \vdash \Theta'' \mid v : \tau, \Sigma$\` by induction hypothesis\\
\end{tabbing}
\item[Case:] $\Gamma; \Delta \vdash (\Theta \mid \ccase{e}{z{:}_{\Delta'}\s~\{\_ \Rightarrow e'}\}) \Downarrow (\Theta'' \mid v) : \tau$
\begin{tabbing}
    (1) $\Gamma; \Delta \vdash \Theta \mid \ccase{e}{z{:}_{\Delta_a}\s~\{\_ \Rightarrow e'}\} : \tau, \Sigma$\\
    (2) $\G, \Theta{\uparrow}^\omega; \D, \Theta{\uparrow}^1 \vdash (\ccase{e}{z{:}_{\Delta_a}\s~\{\_ \Rightarrow e'}\}, \Theta{\downarrow}^1_e, \Sigma) : \tau \otimes \bigotimes(\Theta{\downarrow}^1_\sigma) \otimes \bigotimes(\Sigma)$\`by \ref{lem:unwrapenv}\\
    (3) $Let~\G' = \G, \Theta\uparrow^\omega; \D_a,\D_b,\D_2,\D_3 = \D,\Theta\uparrow^1$\`s.t. $(\D_2,\D_3)$ types $(\Theta\downarrow^1_e, \Sigma)$\\
    (4) $\G'; \D_a, \D_b \vdash \ccase{e}{z{:}_{\D_a}\s\{\_ \Rightarrow e'\}} : \tau$\\
    (5) $\G'; \D_a \vdash e : \sigma$\\
    Subcase $e$ is in WHNF\\
    (-) Similar to the same subcase in the case-with-pattern-vars reduction,\\
        except only $Alt_{\_}$ is relevant which is far simpler\\
    Subcase $e$ is not in WHNF\\
    (6) $\G', z{:}_{\irr{\D_a}}\s; \irr{\D_a}, \D_b \vdash e' : \tau$\\
    (7) $\G'; \D_b, z{:}_1\s \vdash e' : \tau$\` by $\D$-bound to linear\\
    (8) $\G'; \D_a, \D_b, z{:}_1\s \vdash (e, e') : \sigma \otimes \tau$\` by pair\\
    (9) $\G; \D, z{:}_1\s \vdash \Theta \mid e : \sigma, e' : \tau,
    \Sigma$\` by well-typed state\\
    (10) $\G; \D, z{:}_1\s \vdash \Theta' \mid K~\ov{x_i} : \sigma, e' : \tau, \Sigma$\` by induction hypothesis\\
    (11) $Let~\G';\D_c,\D_d,\D_4,\D_5 = \G, \Theta'{\uparrow}^\omega; \Delta, \Theta'{\uparrow}^1$\`s.t. $(\D_4,\D_5)$ types $(\Theta'{\downarrow}^1_e, \Sigma)$\\
    (12) $\G';\D_c \vdash K~\ov{x_i} : \s$\\
    (13) $\G';\D_d,z{:}_1\s \vdash e' : \tau$\\
    (14) $\G';\D_c, \D_d \vdash e'[K~\ov{x_i}/z] : \tau$\` by $1$-subst lemma\\
    (15) $\G;\D \vdash \Theta' \mid e'[K~\ov{x_i}/z] : \tau, \Sigma$\` by well-typed state\\
    (16) $\G;\D \vdash \Theta'' \mid v : \tau, \Sigma$\` by induction hypothesis\\
\end{tabbing}
\end{description}
\end{proof}

We also derive type preservation for the lazy natural semantics.

  \begin{theorem}[Type Preservation]
For any well-typed $[\Theta \mid e]$, if $\Theta : e \Downarrow
\Theta' : e'$ or $\Theta : e \Downarrow^* \Theta' : e'$ then we have
that $[\Theta \mid e']$ is well-typed.
\end{theorem}
\begin{proof}
By Lemma~\ref{lem:naturaltoinstr} and Theorem~\ref{thm:typepres}.
\end{proof}

\subsection{Progress\label{sec:proof:progress}}

\begin{theorem}[Progress]\label{thm:prog}
Let $\G ; \D \vdash \Theta \mid e : \tau,\Sigma$. Any partial derivation $\G; \D \vdash (\Theta \mid e)
\Downarrow?,\Sigma$ can be extended. 
\end{theorem}

\begin{proof}
 Essentially all cases are standard. The only potential issue is the
 (linear) variable rule.
It suffices to show that whenever a variable is evaluated it is
present in the evaluation environment. This situation is impossible
since the state
$\G ; \D \vdash \Theta \mid x : \tau , \Sigma$ where $x$ is absent
from $\Theta$ is not well-typed and so cannot be the head of a partial
derivation. 
\end{proof}

\begin{theorem}[Progress]
For any partial derivation of $\Theta : e \Downarrow?$, for $[\Theta |
e]$ well-typed, the derivation can be extended.
\end{theorem}
\begin{proof}
Follows from Lemma~\ref{lem:instrtonatural} and Theorem~\ref{thm:prog}.
 \end{proof}

\section{Optimisations preserve linearity}\label{app:optimisations}

\subsection{Inlining}
Inlining substitutes occurrences of a let-bound $\D$-variable $x$
with the expression $e$ it is bound to, which determines its usage environment
$\D$. Intuitively, in the let body $e'$, $x$ can occur once or not at all: if
$x$ occurs, then the linear resources $\D$ used indirectly through $x$ are used
via $e$ instead; if $x$ does not occur, then the resources $\D$ are already
used linearly in $e'$ and the substitution is a no-op.

\InliningTheorem

\begin{proof}~

\begin{tabbing}
    (1) $\G;\D,\D' \vdash \llet{\xD = e}{e'} : \vp$\\
    (2) $\G,\D \vdash e : \s$\` by inv. on (let)\\
    (3) $\G,\xD; \D, \D' \vdash e' : \vp$\` by inv. on (let)\\
    (4) $\G;\D,\D' \vdash e'[e/x] : \vp$\` by $\D$-subst. lemma (2,3)\\
    (5) $\G,\xD; \D,\D' \vdash e'[e/x] : \vp$ \` by (admissible) $Weaken_\Delta$\\
    (6) $\G;\D,\D' \vdash \llet{\xD = e}{e'[e/x]} : \vp$\` by (let) (2,5)\\
\end{tabbing}
\end{proof}

\subsection{\texorpdfstring{$\beta$}{Beta}-reduction}

$\beta$-reduction evaluates a $\lambda$-abstraction application by substituting
the $\lambda$-bound variable $x$ with the argument $e'$ in the body of the
$\lambda$-abstraction $e$.
We consider two definitions of $\beta$-reductions, one that substitutes all
occurrences of a variable by an expression, as in call-by-name, and other which
creates a lazy let binding to share the result of computing the argument
expression amongst uses of the variable, as in call-by-need.

The first kind of $\beta$-reduction on term abstractions can be seen to
preserve linearity by case analysis. When the function is linear, the binding
$x$ is used exactly once in the body of the lambda, thus can be substituted by
an expression typed with linear resources, since the expression is guaranteed
to be used exactly once in place of $x$. The proof is direct by type preservation.

\BetaReductionTheorem

\begin{proof}
~
\begin{tabbing}
(1) $\G; \D \vdash (\lambda \x[\pi].~e)~e' : \vp$\\
    (2) $(\lambda \x[\pi].~e) e' \longrightarrow e[e'/x]$\\
    (3) $\G; \D \vdash e[e'/x] : \vp$ \` by type preservation theorem (1,2)\\

\end{tabbing}
\end{proof}

\noindent We assume the $\beta$-reduction with sharing (i.e. the one that creates a let
binding) is only applicable when the $\lambda$-abstraction has an unrestricted
function type. Otherwise, the call-by-name $\beta$-reduction is always
favourable, as we know the resource to be used exactly once and hence sharing
would be counterproductive, and result in an unnecessary heap allocation.
Consequently, the argument to the function must be unrestricted (hence use no
linear resources) for the term to be well-typed, and so it vacuously follows
that linearity is preserved by this transformation.

\BetaReductionSharingTheorem

\begin{proof}
~
\begin{tabbing}
    (1) $\G;\D \vdash (\lambda \xo.~e)~e' : \vp$\\
    (2) $\G; \D \vdash (\lambda \xo.~e) : \s \to_\omega \vp$\` by inv. on $\lambda E_\omega$\\
    (3) $\G; \cdot \vdash e' : \s$ \` by inv. on $\lambda E_\omega$\\
    (4) $\G,\xo; \D \vdash e : \vp$\` by inv. on $\lambda I$\\
    (5) $\G, \x[\cdot][\s]; \D \vdash e : \vp$\` by Lemma~\ref{lem:undelta}\\
    (6) $\G,\G' \vdash \llet{x = e'}{e} : \vp$\` by let (5,3)\\
\end{tabbing}
\end{proof}

\noindent Finally, $\beta$-reduction on multiplicity abstractions is also type
preserving. The argument of the application is a multiplicity rather than an
expression, so no resources are needed to type it, and, since the body of the
lambda must treat the multiplicity $p$ as though it were linear, the body uses
the argument linearly regardless of the instantiation of $p$ at $\pi$.

\BetaReductionMultTheorem

\begin{proof}
    Trivial by invoking preservation using the $\Lambda$ application reduction and the assumption
\end{proof}

\subsection{Case of known constructor}

Case-of-known constructor is a transformation that essentially evaluates a case
expression of a known constructor at compile time, by substituting in the
matching alternative the case binder by the scrutinee and pattern variables by
constructor arguments.
Intuitively, either the case binder is used exactly once to consume the
resources of the scrutinee, or the pattern components whose matching
constructor argument uses linear resources are used exactly once, so one of the
substitutions is a no-op.  If the pattern variables are substituted by the
matching scrutinee expressions, the expressions are still only used once, and
if the case binder is substituted by the scrutinee, it is still used exactly
once.
The proof follows trivially from the preservation theorem.

\CaseOfKnownConstructorTheorem

\begin{proof}~
\begin{tabbing}
    (1) $\G; \D, \D' \vdash \ccase{K~\ov{e}}{\z[\D][\s]~\{..., K~\ov{x} \Rightarrow e_i\}} : \vp$\\
    (2) $\ccase{K~\ov{e}}{\z[\D][\s]~\{..., K~\ov{x} \Rightarrow e_i\}} \longrightarrow e_i\ov{[e/x]}[K~\ov{e}/z]$\\
    (3) $\G; \D,\D' \vdash e_i\ov{[e/x]}[K~\ov{e}/z] : \vp$\` by preservation theorem (1,2)\\
\end{tabbing}
\end{proof}

\subsection{Case of Case\label{sec:proof:caseofcase}}

The case of case transformation applies to case expressions whose scrutinee is
another case expression, and returns the innermost case expression transformed by repeating
the outermost case expression in each alternative of the innermost case,
scrutinizing the original alternative body.

Intuitively, since the scrutinee of the outermost case is not in WHNF, no
resources from it can directly occur in the outermost alternatives. By moving
the outermost alternatives inwards with a different scrutinee, the alternatives
remain well-typed because they are typed using either the case binder or the
pattern bound variables, which, by the \emph{Irrelevance} lemma, makes it
well-typed for any scrutinee consuming arbitrary resources.

\CaseOfCaseTheorem

\begin{proof}~
\begin{tabbing}
    (1) $\G; \D, \D', \D'' \vdash \ccase{\ccase{e_c}{\zD~\{\ov{\rho_{c_i} \Rightarrow e_{c_i}}\}}}{w{:}\s'~\{\ov{\rho_i \Rightarrow e_i}\}} : \vp$\\
    (2) $\G; \D, \D' \vdash \ccase{e_c}{\zD~\{\ov{\rho_{c_i} \Rightarrow e_{c_i}}\}} : \s'$ \`by inv.\\
    (3) $\ov{\G, \var[w][\irr{\D,\D'}][\s']; \irr{\D,\D'},\D'' \vdash \rho \Rightarrow e_i :^w_{\irr{\D,\D'}} \vp}$\`by inv.\\
    Subcase $e_c$ is not in WHNF\\
    (4) $\G;\D \vdash e_c : \s$\\
    (5) $\ov{\G,\var[z][\irr{\D}]; \irr{\D}, \D' \vdash \rho_{c_i} \Rightarrow e_{c_i} :^z_{\irr{\D}} \s'}$\`by inv. (2)\\
    For all alternatives\\
    Subcase $\rho_{c_i} = \_$\\
    (6) $\G,\var[z][\irr{\D}]; \irr{\D}, \D' \vdash e_{c_i} : \s'$\`by inv. on $Alt\_$ (5)\\
    (7) $\ov{\G, \var[w][\irr{\irr{\D},\D'}][\s'],\D''; \irr{\D,\D'},\D'' \vdash_{alt} \rho \Rightarrow e_i :^w_{\irr{\irr{\D},\D'}} \vp}$\`by irrelevance (3)\\
    Subcase $e_{c_i}$ is not in WHNF\\
    (8) $\G,\var[z][\irr{\D}]; \irr{\D},\D',\D'' \vdash \ccase{e_{c_i}}{w{:}_{\irr{\irr{\D},\D'}}\s'~\{\ov{\rho_i \Rightarrow e_i}\}} : \vp$\\\`by (6,7) CaseNotWHNF\\
    Subcase $e_{c_i}$ is in WHNF\\
    (8) $\G,z{:}_{\irr{\D}}\s; \irr{\D},\D' \Vdash e_{c_i} : \s' \gtrdot \ov{\D_i}$\`for some $\ov{\D_i}$ in this subcase\\
    (9) $\G, w{:}_{\ov{\D_i}}\s'; \ov{\D_i}, \D'' \vdash \rho_j \Mapsto e_i :^w_{\ov{\D_i}} \vp$\`by irrelevance (3)\\\` for some matching $\rho_j$\\
    (10) $\G, \z[\irr{\D}]; \irr{\D},\D',\D'' \vdash \ccase{e_{c_i}}{w{:}_{\ov{\D_i}}\s'~\{\ov{\rho_i \Rightarrow e_i}\}} : \vp$\` by (7,8,9) CaseWHNF\\
    Subcase $\rho_{c_i} = K~\ov{\xo}$\\
    (6) $\G,\z[\cdot], \ov{\xo} ; \cdot, \D' \vdash e_{c_i} : \s'$\`by inv. on $Alt0$ (5)\\
    (7) $\ov{\G; \var[w][\irr{\D'}][\s']; \irr{\D'}, \D'' \vdash \rho_i \Rrightarrow e_i :^w_{\irr{\D'}} \vp}$\`by irrelevance (3)\\
    Subcase $e_{c_i}$ is not in WHNF\\
    (8) $\G,\z[\cdot],\ov{\xo}; \D',\D'' \vdash \ccase{e_{c_i}}{\var[w][\irr{\D'}][\s']~\{\ov{\rho_i \Rightarrow e_i}\}} : \vp$\`by (6,7) CaseNotWHNF\\
    Subcase $e_{c_i}$ is in WHNF\\
    (8) $\G,\z[\cdot], \ov{\xo} ; \cdot, \D' \Vdash e_{c_i} : \s' \gtrdot \ov{\D_i}$\`by (6) in subcase\\
    (9) $\G, \var[w][\ov{\D_i}][\s']; \ov{\D_i}, \D'' \vdash_{alt} \rho_i \Mapsto e_i :^w_{\ov{\D_i}} \vp$\`by irrelevance (3)\\
    (10) $\G,\z[\cdot],\ov{\xo}; \ov{\D_i},\D'' \vdash \ccase{e_{c_i}}{\var[w][\ov{\D_i}][\s']~\{\ov{\rho_i \Rightarrow e_i}\}} : \vp$\`by (7,8,9) CaseWHNF\\
    Subcase $\rho_{c_i} = K~\ov{\xo}\ov{\y[1]}$, recalling that $e_c$ is not in WHNF\\
    (6) $\G,\z[\irr{\D}], \ov{\xo}, \ov{y{:}_{\lctag{\irr{\D}}{K_i}}\s}; \irr{\D},\D' \vdash e_{c_i} : \s'$\`by inv. on $AltN_\textrm{Not WHNF}$ (5)\\
    (7) $\G,\var[w][\irr{\irr{\D},\D'}][\s']; \irr{\irr{\D}, \D'}, \D'' \vdash_{alt} \rho_i \Rrightarrow e_i :^w_{\irr{\irr{\D},\D'}} \vp$\`by irrelevance (3)\\
    Subcase $e_{c_i}$ is not in WHNF\\
    (8) $\G,\z[\irr{\D}], \ov{\xo}, \ov{y{:}_{\lctag{\irr{\D}}{K_i}}\s}; \irr{\D},\D',\D'' \vdash \ccase{e_{c_i}}{\var[w][\irr{\irr{\D},\D'}][\s']~\{\ov{\rho_i \Rightarrow e_i}\}} : \vp$\\
    Subcase $e_{c_i}$ is in WHNF\\
    (8) $\G,\z[\irr{\D}], \ov{\xo}, \ov{y{:}_{\lctag{\irr{\D}}{K_i}}\s}; \irr{\D},\D' \Vdash e_{c_i} : \s' \gtrdot \ov{\D_i}$\`by (6) in subcase\\
    (9) $\G,\var[w][\irr{\D},\D'][\s']; \irr{\D}, \D', \D'' \vdash_{alt} \rho_i \Mapsto e_i :^w_{\irr{\D},\D'} \vp$\`by irrelevance (3)\\
    (10) $\G,\z[\irr{\D}], \ov{\xo}, \ov{y{:}_{\lctag{\irr{\D}}{K_i}}\s}; \irr{\D},\D',\D'' \vdash \ccase{e_{c_i}}{\var[w][\irr{\D},\D'][\s']~\{\ov{\rho_i \Rightarrow e_i}\}} : \vp$\\
    Commonly to all alternatives subcases:\\
    (11) $\G;\D,\D',\D'' \vdash \ccase{e_c}{\textrm{alternatives from}~(8)~\textrm{or}~(10)} : \vp$\\
    Subcase $e_c$ is in WHNF\\
    (4) $\G; \D \Vdash e_c : \s \gtrdot \D$\\
    (5) $\G, \zD; \D,\D' \vdash \rho_{c_j} \Rightarrow e_{c_i} :^z_\D \s \Mapsto \vp$\`for $\rho_{c_j}$ matches $e_c$\\
    (5) $\ov{\G, \z[\irr{D}]; \irr{\D},\D' \vdash \rho_{c_i} \Rightarrow e_{c_i} :^z_{\irr{\D}} \s \Rrightarrow \vp}$\\
    Continue as in the previous subcase, but with $\D$ instead of $\irr{\D}$\\

\end{tabbing}
\end{proof}
 
\subsection{Commuting let}

The let floating transformations move lazy let constructs, in and out of other
constructs, to further unblock more optimisations.
In essence, since let bindings consume resources lazily (by introducing a
$\D$-variable with usage environment $\D$, where $\D$ is the typing environment
of the bound expression), we can intuitively move them around without violating linearity.
We prove let floating transformations \emph{full-laziness} and three
\emph{local-transformations} preserve types and linearity.

\FullLazinessTheorem

\begin{proof}~

\begin{tabbing}
    (1) $\G; \D,\D' \vdash \lambda \y[\pi].~\llet{\xD = e}{e'} : \s' \to \vp$\\
    Subcase $\pi = 1$\\
    (2) $\G; \D,\D',\y[1] \vdash \llet{\xD = e}{e'} : \vp$\`by inv. on $\lambda I$\\
    (3) $\G; \D \vdash e : \s$\` by inv. on (let)\\
    (4) $\G, \xD; \D, \D', \y[1] \vdash e' : \vp$\`by inv. on (let)\\
    (5) $\G, \xD; \D, \D' \vdash \lambda \y[1].~e' : \s' \to \vp$\`by ($\lambda I$) (4)\\
    (6) $\G; \D,\D' \vdash \llet{\xD = e}{\lambda \y[1].~e'} : \s' \to \vp$ by (let) $(3,5)$\\
    Subcase $\pi = \omega$\\
    As above but $x$ is put in the unrestricted context $\G$
\end{tabbing}
\end{proof}

\LocalTransformationsTheorem

\begin{description}
\item[1.] Commuting Let-app
\begin{proof}~
\begin{tabbing}
    (1) $\G; \D, \D', \D'' \vdash (\llet{\xD = e_1}{e_2})~e_3 : \vp$\\
    (2) $\G; \D, \D' \vdash \llet{\xD = e_1}{e_2} : \s' \to_\pi \vp$ \`by inv. on 1\\
    (3) $\G; \D'' \vdash e_3 : \s'$\`by inv. on 1\\
    (4) $\G; \D \vdash e_1 : \s$\`by inv. on 2\\
    (5) $\G, \xD; \D, \D' \vdash e_2 : \s' \to_\pi \vp$\`by inv. on 2\\
    (6) $\G, \xD; \D, \D', \D'' \vdash e_2~e_3 : \vp$\` by $\lambda_\pi E$\\
    (7) $\G; \D, \D', \D'' \vdash \llet{\xD = e_1}{e_2~e_3} : \vp$\`by let (4,6)\\
\end{tabbing}
\end{proof}

\item[2.] Commuting let-case
\begin{proof}~
\begin{tabbing}
    (1) $\G; \D, \D', \D'' \vdash \ccase{\llet{\xD = e_1}{e_2}}{\z[\D,\D'][\s']~\{\ov{\rho \Rightarrow e_3}\}} : \vp$\\
    (2) $\G; \D, \D' \vdash \llet{\xD = e_1}{e_2} : \s'$\`by inv. on 1\\
    (3) $\G; \D \vdash e_1 : \s$\`by inv. on 2\\
    (4) $\G, \xD;\D,\D' \vdash e_2 : \s'$\`by inv. on 2\\
    Subcase $e_2$ is in WHNF\\
    (5) $\ov{\G,\z[\D,\D'][\s']; \D, \D', \D'' \vdash_{alt} \rho \Rightarrow e_3 :^z_{\D,\D'} \s' \Mapsto \vp}$\`by inv. on 1\\
    (6) $\G, \xD; \D, \D', \D'' \vdash \ccase{e_2}{\z[\D,\D'][\s']~\{\ov{\rho \Rightarrow e_3}\}} : \vp$\`by CaseWHNF (4,5)\\
    (7) $\G; \D, \D', \D'' \vdash \llet{\xD = e_1}{\ccase{e_2}{\z[\D,\D'][\s']~\{\ov{\rho \Rightarrow e_3}\}}} : \vp$\`by Let (3,6)\\
    Subcase $e_2$ is not in WHNF\\
    (5) $\ov{\G,\z[\irr{\D,\D'}][\s']; \irr{\D, \D'}, \D'' \vdash_{alt} \rho \Rightarrow e_3 :^z_{\irr{\D,\D'}} \s' \Rrightarrow \vp}$\`by inv. on 1\\
    (6) $\G, \xD; \D, \D', \D'' \vdash \ccase{e_2}{\z[\D,\D'][\s']~\{\ov{\rho \Rightarrow e_3}\}} : \vp$\`by CaseWHNF (4,5)\\
    (7) $\G; \D, \D', \D'' \vdash \llet{\xD = e_1}{\ccase{e_2}{\z[\D,\D'][\s']~\{\ov{\rho \Rightarrow e_3}\}}} : \vp$\`by Let (3,6)\\
\end{tabbing}
\end{proof}

\item[3.] Commuting let-let
\begin{proof}~
\begin{tabbing}
    (1) $\G, \D;\D',\D'' \vdash \llet{\x[\D,\D'][\s'] = (\llet{\yD = e_1}{e_2})}{e_3} : \vp$\\
    (2) $\G; \D, \D' \vdash \llet{\yD = e_1}{e_2} : \s'$\`by inv. on 1\\
    (3) $\G, \x[\D,\D'][\s'];\D,\D',\D'' \vdash e_3$\`by inv. on 1\\
    (4) $\G; \D \vdash e_1 : \s$\`by inv. on 2\\
    (5) $\G, \yD; \D,\D' \vdash e_2 : \s'$\`by inv. on 2\\
    (6) $\G, \yD; \D,\D',\D'' \vdash \llet{\x[\D,\D'][\s'] = e_2}{e_3} : \vp$\`by Let (3,5) and Weaken\\
    (7) $\G; \D,\D',\D'' \vdash \llet{\yD = e_1}{\llet{\x[\D,\D'][\s'] = e_2}{e_3}} : \vp$\`by Let (4,6)\\
\end{tabbing}
\end{proof}
\end{description}

\begin{theorem}
If $\G; \D_1,\D_2,\D_3 \vdash \ccase{e}{\{\overline{\rho_j\Rightarrow e_j}, \rho \Rightarrow
  E[e_1], \overline{\rho_k\Rightarrow e_k}\}}  : \vp$ and $e_1$ does not
use any pattern variables introduced by $\rho$, the case binder nor variables bound in
context $E[{-}]$
then 
$\G ; \D_1,\D_2,\D_3  \vdash \llet{x = e_1}{\ccase{e}{\{\overline{\rho_j\Rightarrow
           e_j}, \rho \Rightarrow E[x], \overline{\rho_k\Rightarrow e_k}\}}} : \vp$,
     for some fresh $x$.
              
   \end{theorem}
   
   \begin{proof}
   ~  
     \begin{tabbing}
       Subcase: $e$ is in WHNF\\
       (1) $\G;\D_1 \Vdash e:\s \gtrdot \ov{\D_i}$ \` by inversion\\
       (2) $\G ; \D_1 \vdash e : \s$ \` by inversion and constructor rule
       if needed\\
       (3) $\G,\var[z][\ov{\D_i}];\ov{\D_i},\D_2,\D_3 \vdash_{alt} \rho_j
       \Rightarrow e' :^z_{\ov{\D_i}} \s \Mapsto \vp$ \` by inversion\\
       (4) $\ov{\judg[\G,z{:}_{\irr{\D_1}}\s][\irr{\D_1},\D_2,\D_3]{\rho\Rightarrow e'}{\s
           \Rrightarrow \vp}[alt][\irr{\D_1}][z]}$ \` by inversion\\
       (5) $e$ matches $\rho_j$ \` by inversion\\
       Subsubcase: $\rho_j$ corresponds to branch containing $E[e_1]$\\
       Subsubsubcase: (3) derived by $\textrm{AltN}_{\textrm{WHNF}}$\\
       $\judg[\G, z{:}_{\ov{\Delta_i}} \sigma,\ov{\xo},\ov{y_i{:}_{\D_i}\s_i}^n][\ov{\D_i},\D_2,\D_3]{E[e_1]}{\vp}$
       \` by inversion\\
       
       $\judg[\G, z{:}_{\ov{\Delta_i}} \sigma,\ov{\xo},\ov{y_i{:}_{\D_i}\s_i}^n][\D',\D_2]{e_1}{\sigma'}$
       \` by inversion, with $\D' \subseteq \ov{\D_i} \subseteq \D_1$\\
       $\judg[\G, x{:}_{\D',\D_2}\sigma', z{:}_{\ov{\Delta_i}} \sigma,\ov{\xo},\ov{y_i{:}_{\D_i},\s_i}^n][\D',\D_2]{x}{\sigma'}$
       \` by Var$_\Delta$, with $x$ fresh\\
      $\judg[\G][\D',\D_2]{e_1}{\sigma'}$
       \` by assumption\\
       $\judg[\G,
       x{:}_{\D',\D_2}\sigma', z{:}_{\ov{\Delta_i}} \sigma,\ov{\xo},\ov{y_i{:}_{\D_i}\s_i}^n][\ov{\D_i},\D_2,\D_3]{E[x]}{\vp}$
       \` by context instantiation\\
       $\G,x{:}_{\D',\D_2}\sigma',\var[z][\ov{\D_i}];\ov{\D_i},\D_2,\D_3 \vdash_{alt} \rho_j
       \Rightarrow E[x] :^z_{\ov{\D_i}} \s \Mapsto \vp$ \` by
       $\textrm{AltN}_{\textrm{WHNF}}$\\
       $\G,x{:}_{\D',\D_2}\sigma'; \D_1,\D_2,\D_3 \vdash \ccase{e}{\{\overline{\rho_j\Rightarrow e_j}, \rho \Rightarrow
  E[x], \overline{\rho_k\Rightarrow e_k}\}} : \vp$ \` by $\textrm{Case}_\textrm{WHNF}$\\
       $\G ; \D_1 , \D_2 , \D_3 \vdash \llet{\x[\D',\D_2] = e_1}{\ccase{e}{\{\overline{\rho_j\Rightarrow e_j}, \rho \Rightarrow
           E[x], \overline{\rho_k\Rightarrow e_k}\}}} : \vp$ \` by $\mathit{Let}$\\
       Subsubsubcase: (3) derived by $\textrm{Alt}0$\\
       $\judg[\subst{\G, z{:}_{\ov{\Delta_i}} \sigma}{\cdot}{\ov{\D_i}}_z,\ov{\xo}][\subst{\ov{\D_i},\D_2,\D_3}{\cdot}{\ov{\D_i}}]{E[e_1]}{\vp}$
       \` by inversion\\
       $\judg[\G,z{:}_{\cdot}\s,\ov{\xo}][\D_2]{e_1}{\sigma'}$
       \` by inversion\\
       $\judg[\G][\D_2]{e_1}{\sigma'}$ \` by assumption\\
       $\judg[\G,x{:}_{\D_2}\sigma' , z{:}_{\cdot}\s,\ov{\xo}][\D_2]{x}{\sigma'}$ \` by Var$_\Delta$, with $x$ fresh\\
       $\judg[\subst{\G,x{:}_{\D_2}\sigma', z{:}_{\ov{\Delta_i}}
         \sigma }{\cdot}{\ov{\D_i}}_z,\ov{\xo}][\subst{\ov{\D_i},\D_2,\D_3}{\cdot}{\ov{\D_i}}]{E[x]}{\vp}$
       \` by context instantiation\\
             $\G,x{:}_{\D_2}\sigma'; \D_1,\D_2,\D_3 \vdash \ccase{e}{\{\overline{\rho_j\Rightarrow e_j}, \rho \Rightarrow
  E[x], \overline{\rho_k\Rightarrow e_k}\}} : \vp$ \` by $\textrm{Alt}0$ and
$\textrm{Case}_\textrm{WHNF}$\\
$\G ; \D_1 , \D_2 , \D_3 \vdash \llet{\x[\D_2] = e_1}{\ccase{e}{\{\overline{\rho_j\Rightarrow e_j}, \rho \Rightarrow
           E[x], \overline{\rho_k\Rightarrow e_k}\}}} : \vp$ \` by
       $\mathit{Let}$\\
       Subsubsubcase: (3) derived by $\textrm{Alt}\_$\\
       $\G , z{:}_{\ov{\Delta_i}} \sigma; \D_2,\D_3 \vdash E[e_1] : \vp$ \` by inversion\\
       $\judg[\G , z{:}_{\ov{\Delta_i}} \sigma][\D_2]{e_1}{\sigma'}$
       \` by inversion\\
       $\judg[\G, , z{:}_{\ov{\Delta_i}} \sigma ,x{:}_{\D_2}\sigma'][\D_2]{x}{\sigma'}$ \` by Var$_\Delta$, with $x$ fresh\\
       $\G, z{:}_{\ov{\Delta_i}} \sigma ,x{:}_{\D_2}\sigma'; \D_2,\D_3 \vdash E[e_1] : \vp$ \` by
       context instantiation\\
       $\G,x{:}_{\D_2}\sigma'; \D_1,\D_2,\D_3 \vdash \ccase{e}{\{\overline{\rho_j\Rightarrow e_j}, \rho \Rightarrow
  E[x], \overline{\rho_k\Rightarrow e_k}\}} : \vp$ \` by $\textrm{Alt}\_$ and
$\textrm{Case}_\textrm{WHNF}$\\
       $\G ; \D_1 , \D_2 , \D_3 \vdash \llet{\x[\D_2] = e_1}{\ccase{e}{\{\overline{\rho_j\Rightarrow e_j}, \rho \Rightarrow
           E[x], \overline{\rho_k\Rightarrow e_k}\}}} : \vp$ \` by
       $\mathit{Let}$\\
       Subsubcase: $\rho_j$ does not correspond to branch containing
       $E[e_1]$\\
       (6) $\judg[\G,z{:}_{\irr{\D_1}}\s][\irr{\D_1},\D_2,\D_3]{\rho\Rightarrow E[e_1]}{\s
         \Rrightarrow \vp}[alt][\irr{\D_1}][z]$ \` this subsubcase\\
       Subsubsubcase: (6) derived by $\textrm{AltN}_{\textrm{Not
           WHNF}}$\\
       $\judg[\G, z{:}_{\irr{\D_1}}\s,\ov{\xo},\ov{y_i{:}_{\D_i}\s_i}][\irr{\D_1},\D_2,\D_3]{E[e_1]}{\vp}$
       and $\ov{\D_i} = \ov{\lctag{\D_1}{K_i}}^n$ \` by inversion\\
       $\judg[\G,
       z{:}_{\irr{\D_1}}\s,\ov{\xo},\ov{y_i{:}_{\D_i}\s_i}][\D_2]{e_1}{\sigma'}$
       \` by inversion\\
       $\judg[\G,x{:}_{\D_2}\sigma',
       z{:}_{\irr{\D_1}}\s,\ov{\xo},\ov{y_i{:}_{\D_i}\s_i}][\D_2]{x}{\sigma'}$
       \` by Var$_\Delta$, with $x$ fresh\\
       $\judg[\G, x{:}_{\D_2}\sigma', z{:}_{\irr{\D_1}}\s,\ov{\xo},\ov{y_i{:}_{\D_i}\s_i}][\irr{\D_1},\D_2,\D_3]{E[x]}{\vp}$
       and $\ov{\D_i} = \ov{\lctag{\D_1}{K_i}}^n$ \` by context
       instantiation\\
         $\G,x{:}_{\D_2}\sigma'; \D_1,\D_2,\D_3 \vdash \ccase{e}{\{\overline{\rho_j\Rightarrow e_j}, \rho \Rightarrow
  E[x], \overline{\rho_k\Rightarrow e_k}\}} : \vp$ \` by $\textrm{AltN}_{\textrm{Not
           WHNF}}$ and
       $\textrm{Case}_\textrm{WHNF}$\\
         $\G ; \D_1 , \D_2 , \D_3 \vdash \llet{\x[\D_2] = e_1}{\ccase{e}{\{\overline{\rho_j\Rightarrow e_j}, \rho \Rightarrow
           E[x], \overline{\rho_k\Rightarrow e_k}\}}} : \vp$ \` by
       $\mathit{Let}$\\
       Subsubsubcase: (6) derived by $\textrm{Alt}\_$ or
       $\textrm{Alt}0$\\
       Identical to subsubsubcases above.\\
       Subcase: $e$ is not in WHNF\\
       (1) $\G ; \D_1 \vdash e : \s$ \` by inversion\\
       (2) $\ov{\judg[\G,z{:}_{\irr{\D_1}}\s][\irr{\D_1},\D_2,\D_3]{\rho\Rightarrow
           e'}{\s \Rrightarrow \vp}[alt][\irr{\D_1}][z]}$ \` by
       inversion\\
       (3) $\judg[\G,z{:}_{\irr{\D_1}}\s][\irr{\D_1},\D_2,\D_3]{\rho\Rightarrow
           E[e_1]}{\s \Rrightarrow \vp}[alt][\irr{\D_1}][z]$ \` by
         inversion\\
       Subsubcase: (3) derived by $\textrm{AltN}_{\textrm{Not
           WHNF}}$, $\textrm{Alt}\_$ or
       $\textrm{Alt}0$\\
       Identical to Subsubsubcases above.
     \end{tabbing}

   \end{proof}
 
\subsection{\texorpdfstring{$\eta$}{Eta}-conversions}

The $\eta$-conversion transformations are $\eta$-expansion and
$\eta$-reduction. In both transformations, linearity is preserved since the
resources used to type the function $f$ do not change neither when the lambda
and its argument are removed, nor when we add a lambda and apply $f$ to the
bound argument.

\EtaExpansionTheorem

\begin{proof}~
\begin{tabbing}
    Subcase $f$ is linear\\
    (1) $\G; \D \vdash f : \s \lolli \vp$\\
    (2) $\G; \xl \vdash x : \s$\\
    (3) $\G; \D, \xl \vdash f~x : \vp$\`by $\lambda E$\\
    (4) $\G; \D \vdash (\lambda \xl.~f~x) : \s \lolli \vp$\`by $\lambda I$\\
    Subcase $f$ is unrestricted\\
    As above but $x$ is introduced in $\G$ and functions are unrestricted
\end{tabbing}
\end{proof}

\EtaReductionTheorem

\begin{proof}~
\begin{tabbing}
    (1) $\G; \D \vdash (\lambda \x[\pi].~f~x) : \s \to_\pi \vp$\\
    Subcase $\pi = 1$\\
    (2) $\G; \D, \xl \vdash f~x : \vp$\`by inv. on $\lambda I$\\
    (3) $\G; \D \vdash f : \s \to_1 \vp$\`by inv. on $\lambda E$\\
    Subcase $\pi = \omega$\\
    As above but $x$ is introduced in $\G$
\end{tabbing}
\end{proof}

\subsection{Binder Swap}

The binder swap transformation applies to case expressions whose scrutinee is a
single variable $x$, and it substitutes occurrences of $x$ in the case
alternatives for the case binder $z$. If $x$ is a linear resource, $x$ cannot
occur in the case alternatives (as we conservatively consider variables are not
in WHNF), so the substitution preserves types vacuously. Otherwise, $x$ can be
freely substituted by $z$, since $z$ is also an unrestricted resource (it's
usage environment is empty because $x$ is unrestricted).

\BinderSwapTheorem

\begin{proof}~
\begin{tabbing}
    Subcase $x$ is linear\\
    (1) $\G; \D,\xl \vdash \ccase{x}{\z[\irr{x}]~\{\ov{\rho \Rightarrow e}\}} : \vp$\\
    (2) $\G; \xl \vdash x : \s$\`by inv. on $Case_\textrm{Not WHNF}$\\
    (3) $\ov{\G, \z[\irr{x}]; \D, \irr{\xl} \vdash_{alt} \rho \Rightarrow e :^z_{\irr{x}} \s \Rrightarrow \vp}$\`by inv. on $Case_\textrm{Not WHNF}$\\
    (4) $\ov{\G, \z[\irr{x}]; \D, \irr{\xl} \vdash_{alt} \rho \Rightarrow e[z/x] :^z_{\irr{x}} \s \Rrightarrow \vp}$\\\`by $x$ cannot occur in $e$ bc it's proof irrelevant\\
    (5) $\G;\D,\xl \vdash \ccase{x}{\z[\irr{x}]~\{\ov{\rho \Rightarrow e[z/x]}\}} : \vp$\`by $Case_\textrm{Not WHNF}$\\
    Subcase $x$ is unrestricted\\
    (1) $\G,\xo; \D \vdash \ccase{x}{\z[\cdot]~\{\ov{\rho \Rightarrow e}\}} : \vp$\\
    (2) $\G,\xo; \cdot \vdash x : \s$\`by inv. on $Case_\textrm{Not WHNF}$\\
    (3) $\ov{\G,\xo, \z[\cdot]; \D \vdash_{alt} \rho \Rightarrow e :^z_{\cdot} \s \Rrightarrow \vp}$\`by inv. on $Case_\textrm{Not WHNF}$\\
    (4) $\G,\z[\cdot]; \cdot \vdash z : \s$\`by $Var_\D$\\
    (5) $\ov{\G, \z[\cdot]; \D \vdash_{alt} \rho \Rightarrow e[z/x] :^z_{\cdot} \s \Rrightarrow \vp}$\`by unr. subst. lemma (3,4)\\
(6) $\G,\xo;\D \vdash \ccase{x}{\z[\cdot]~\{\ov{\rho \Rightarrow e[z/x]}\}} : \vp$\`by $Weaken_\omega$ and $Case_\textrm{Not WHNF}$\\
\end{tabbing}

\end{proof}

\end{document}